\documentclass[14pt]{article}
\pdfoutput=1
\usepackage{amsmath,amssymb,amsfonts,amsthm,bm,bbm,cancel,wasysym}
\usepackage{epsfig,graphics,graphicx,epstopdf,caption,subcaption}
\graphicspath{{Charts/}}
\usepackage{array,booktabs,colortbl,colordvi,multirow}
\usepackage{colordvi,color,xcolor}
\usepackage{hyperref}
\usepackage{rotating}
\usepackage{comment}
\usepackage{feynmf}
\usepackage{tikz}
\usetikzlibrary{shapes, arrows}
\tikzstyle{query} = [rectangle, draw, text centered, rounded corners, minimum height = 2em]
\tikzstyle{connector} = [draw, -latex']

\usepackage{verbatim}
\usepackage{cite}
\usepackage{amsmath}
\usepackage{physics}
\usepackage{mathtools}
\usepackage{amsthm}
\usepackage{algorithm2e}
\usepackage{setspace}
\usepackage{url}
\usepackage[percent]{overpic}
\usepackage{slashed}
\usepackage{authblk}
\usepackage{xspace}
\usepackage{fullpage}
\usepackage[numbers,sort&compress]{natbib}
\usepackage{multirow}
\usepackage{array}
\newcolumntype{P}[1]{>{\centering\arraybackslash}p{#1}}

\newtheorem{lemma}{Lemma}

\def\ie{{\it i.e.}}
\def\eg{{\it e.g.}}

\newskip\zatskip \zatskip=0pt plus0pt minus0pt
\def\matth{\mathsurround=0pt}
\def\lsim{\mathrel{\mathpalette\atversim<}}
\def\gsim{\mathrel{\mathpalette\atversim>}}
\def\atversim#1#2{\lower0.7ex\vbox{\baselineskip\zatskip\lineskip\zatskip
  \lineskiplimit 0pt\ialign{$\matth#1\hfil##\hfil$\crcr#2\crcr\sim\crcr}}}

\newif\ifdiagrams
\diagramstrue
\ifdiagrams

\else
  \excludecomment{fmffile}
\fi

\begin{document}


\begin{flushright}
\today
\end{flushright}
\vspace*{5mm}

\renewcommand{\thefootnote}{\fnsymbol{footnote}}
\setcounter{footnote}{1}
\begin{center}

{\Large {\bf BFBrain: Scalar Bounded-From-Below Conditions from Bayesian Active Learning}}\\

\vspace*{0.75cm}

{\bf George N. Wojcik~\footnote{gwojcik@wisc.edu}}

\vspace{0.5cm}

{Department of Physics, University of Wisconsin-Madison, Madison, WI 53706, USA}

\end{center}
\vspace{.5cm}

\begin{abstract}
 
\noindent  
We present a procedure leveraging Bayesian deep active learning to rapidly produce highly accurate approximate bounded-from-below conditions for arbitrary renormalizable scalar potentials, in the form of a neural network which may be saved and exported for use in arbitrary parameter space scans. We explore the performance of our procedure on three different scalar potentials with either highly nontrivial or unknown symbolic bounded-from-below conditions (the most general two-Higgs doublet model, the three-Higgs doublet model, and a version of the Georgi-Machacek model without custodial symmetry). We find that we can produce fast and highly accurate binary classifiers for all three potentials. Furthermore, for the potentials for which no known symbolic necessary and sufficient conditions on boundedness-from-below exist, our classifiers substantially outperform some common approximate analytical methods, such as producing tractable sufficient but not necessary conditions or evaluating boundedness-from-below conditions for scenarios in which only a subset of the theory's fields achieve vev's. Our methodology can be readily adapted to any renormalizable scalar field theory. For the community's use, we have developed a Python package, BFBrain, which allows for the rapid implementation of our analysis procedure on user-specified scalar potentials with a high degree of customizability.
\end{abstract}

\renewcommand{\thefootnote}{\arabic{footnote}}
\setcounter{footnote}{0}
\thispagestyle{empty}
\vfill
\newpage
\setcounter{page}{1}

\section{Introduction}\label{sec:intro}

Although the Standard Model (SM) contains only a single elementary scalar (the Higgs boson), given the SM's failure to adequately address a variety of ongoing theoretical and experimental questions, such as the identity of dark matter, the gauge-gravity hierarchy problem, and the origin of the intricate and nontrivial structure of fermion masses and mixings, it would be needlessly limiting to avoid considering theories with multiple additional scalar fields. Over the decades, significant model-building work has been and continues to be done on a variety of theories with extended scalar sectors (see, for example, \cite{Maniatis:2006fs, ivanov2017building, Georgi1985DoublyCH, Weinberg:1976hu, Moultaka:2020dmb, Boto:2022uwv, Kannike:2016fmd}), in which these beyond the Standard Model (BSM) scalars might spontaneously break new gauge symmetries, act as dark matter candidates or mediators, reproduce the observed SM Higgs mass with a greater degree of naturalness, and more. Unfortunately, multi-scalar BSM theories tend to present model builders with considerable difficulties-- in a significant subset of such constructions, the scalar sector is in fact the principal source of model complexity. Unlike gauge interactions, which have highly constrained forms for their interactions by virtue of the elegant structure suggested by local gauge invariance, there are comparatively few a priori theoretical constraints on most scalar coupling parameters. As such, the number of free parameters in a given theory will grow rapidly with the number of scalar fields, quickly making comprehensive explorations of these parameter spaces nontrivial at best and entirely intractable at worst. Among the more notorious problems arising in multi-scalar theories is the issue of boundedness-from-below of the scalar potential, a necessary precondition for stability of the vacuum. The problem is simple enough to express: For a given set of model parameters, the scalar potential function in the action must have an absolute minimum at some finite configuration of scalar vacuum expectation values (vev's). If this condition is not satisfied, then the model has no absolute minimum and any vacuum configuration we observe in our universe is necessarily unstable.

A necessary step to determining the bounded-from-below region of scalar potential parameter space is specifically identifying \emph{strict} boundedness-from-below conditions, namely the region of parameter space in which the sum of all quartic terms of the scalar potential are always positive. Once this criterion is resolved, the complete space of bounded-from-below scalar potentials can be determined simply by requiring that on the boundaries of the strictly bounded-from-below region (that is, where there exists a vev configuration such that the quartic coefficients sum to 0), the cubic and/or quadratic terms ensure boundedness-from-below is maintained. For many BSM studies, it in fact suffices to simply determine the strict bounded-from-below conditions, since a region satisfying the weak boundedness-from-below conditions but not the strict ones is likely infinitesimally close to a model for which the strict condition is observed-- this practice is so common that in subsequent sections in this paper we shall often for the sake of brevity refer only to ``bounded-from-below conditions'' when discussing the strict positivity criteria.\footnote{There are, of course, models where flat directions of the quartic potential are significant, perhaps most famously the Minimal Supersymmetric Standard Model (MSSM) \cite{Gherghetta:1995dv}.}

At tree level, identifying the strictly bounded-from-below region of parameter space amounts to determining if a given multivariable quartic polynomial, parameterized by the quartic scalar coupling coefficients, is everywhere positive, or equivalently, that a rank-4 tensor is positive-definite. Since many radiative corrections can be expressed simply as modifications of the tree-level quartic couplings, resolving the tree-level bounded-from-below region generally also addresses the problem at arbitrary loop level as well. Unfortunately, while it is always possible to determine the positive-definiteness of a given rank-4 tensor \cite{Kannike:2016fmd}, the general procedure to do so is NP-hard \cite{Kannike:2016fmd, Ivanov:2018jmz, Murty1987SomeNP}, and hence often not feasible to accomplish even for numerical scans of the model parameter space. While certain special cases can be treated more easily\footnote{for example, if the quartic terms in a potential can be split into biquadratic forms, such as in the inert doublet model \cite{Deshpande:1977rw}, the problem can be reduced to simply proving the positive-definiteness of a rank-2 tensor, which can be readily done numerically or symbolically.}, the only recourse for many scalar potentials is to embark on an enormously complicated symbolic analysis, which will generally not yield compact solutions if closed-form results are obtainable at all. Failing to find exact necessary and sufficient conditions, model builders have generally relied on tractable symbolic approximations to these conditions by finding criteria that are necessary but not sufficient (\eg, \cite{Blasi:2017xmc}, in which case the criteria will permit points that are not, in fact, stable) or sufficient but not necessary (\eg, \cite{Grzadkowski:2009bt}, in which case a stable allowed portion of parameter space will be omitted). Any attempt to create a generic strategy to approach the boundedness-from-below problem, therefore, naturally must address the punishing complexity of determining the positive-definiteness of a rank-4 tensor, even approximately. 

In this work, we propose mitigating this problem by using active learning: By training a sequential neural network on comparatively few explicitly labelled points and using the network itself to propose additional training examples, the full landscape of decision boundaries in the space of the potential's quartic coupling constants for boundedness-from-below can be rapidly approximated while probing only a fraction of the full parameter space-- recently in \cite{caron2019constraining,goodsell2023active,hammad2023exploration} these techniques have been used in other BSM model building contexts . Training can be done on the scale of hours on an individual personal computer and requires only publicly available machine learning frameworks, and once the decision boundary is well-explored, the trained neural network itself can be used as a fast classifer to evaluate the boundedness-from-below of the scalar potential for \emph{any} point in the model's parameter space to a remarkable degree of accuracy. On explicitly labelled test data, we find that our classifiers substantially outperform common examples of necessary but not sufficient and sufficient but not necessary approximate bounded-from-below conditions. Furthermore, because we employ a Bayesian neural network, the trained classifier also has a meaningful metric for uncertainty in its predictions, which can be used to gauge the reliability of its results.

Once trained, the classifier produced by our methodology represents a portable set of approximate bounded-from-below conditions, accurate enough to be usable for a wide range of phenomenological studies and with reliable uncertainty estimates on its predictions, that can in turn be saved, shared, and applied to new points in parameter space with minimal computational effort. In this paper, we shall outline the key components of our procedure and explore the results of applying it to several scalar potentials with nontrivial or unknown strict bounded-from-below conditions. The remainder of this paper is laid out as follows. In Section \ref{sec:problem-setup}, we outline the problem of strict boundedness-from-below in greater detail and establish notation. In Section \ref{sec:bayesian-nn}, we review some key concepts in Bayesian deep learning and reference their application in our setting, with a particular emphasis on uncertainty quantification, as well as arguing for the suitability of a Bayesian neural network for the classification task at hand. In Section \ref{sec:active-learning}, we review the fundamental concepts in active learning (particularly in the context of a Bayesian neural network) and describe the components of our active learning procedure in detail. In Section \ref{sec:experiments}, we present the results of applying our procedure to three different scalar potentials: The most general two-Higgs doublet model (2HDM), the Weinberg three-Higgs doublet model of \cite{Weinberg:1976hu}, and the Georgi-Machacek model \cite{Georgi1985DoublyCH} with custodial symmetry broken (as occurs at the loop level \cite{gunion:1990dt,Blasi:2017xmc,Moultaka:2020dmb}). Finally, in Section \ref{sec:conclusion}, we summarize our findings and discuss directions for future work. In addition to this paper, we have published a Python package on GitHub, BFBrain \cite{Wojcik2023}, which allows for our procedure to be easily implemented for arbitrary user-specified scalar potentials, and includes substantial customization options.

\section{Problem Setup}\label{sec:problem-setup}

In general, a renormalizable scalar potential with real degrees of freedom given as $\phi_i$ can be written as
\begin{align}
    V(\phi) = M_{i j} \phi_i \phi_j + \xi_{i j k} \phi_i \phi_j \phi_k + Q(\vec{\lambda})_{i j k l} \phi_i \phi_j \phi_k \phi_l,
\end{align}
where $M$, $\xi$, and $Q(\vec{\lambda})$ are a rank-2, rank-3, and rank-4 tensor, respectively, while $\vec{\lambda}$ is a vector of independent real coefficients of dimension-4 scalar operators-- for later notational convenience we keep these coefficients explicit, and without loss of generality let $Q(\vec{\lambda})$ be a linear function of $\vec{\lambda}$. A necessary precondition for vacuum stability of the model (barring additional nonrenormalizable operators) is that such a model is bounded from below-- that is, the global minimum of $V$ is greater than $- \infty$. As noted in Section \ref{sec:intro}, a necessary step in this analysis is confirming strict bounded-from-below conditions, namely
\begin{align}\label{eq:general-BFB-cond}
    Q(\vec{\lambda})_{i j k l} \phi_i \phi_j \phi_k \phi_l > 0 \; \forall \; \phi \in \mathbb{R}^d.
\end{align}
In tensor algebra terms, the rank-4 tensor $Q(\vec{\lambda})$ must be positive-definite. It is straightforward to see that many of the characteristics of positive-definiteness of matrices carry over to rank-n tensors-- for example, a positive rescaling of a positive-definite tensor remains positive-definite, and the sum of two positive-definite tensors must necessarily remain positive-definite. Two lemmas are of particular use when characterizing the space of $\vec{\lambda}$ values that satisfy Eq.(\ref{eq:general-BFB-cond}) for some scalar potential. First,
\begin{lemma}\label{prop:convex}
Let $Q(\vec{\lambda})_{ijkl}$ be a rank-4 tensor representing the quartic part of a scalar potential, and let $Q$ be a linear function of $\vec{\lambda}$. Then, the set of $\vec{\lambda}$ such that $Q(\vec{\lambda})_{ijkl}$ is positive-definite is convex.\footnote{Thank you to Matthew Sullivan for pointing this out.}
\end{lemma}
\begin{proof}
    A region $\mathcal{C}$ is convex by definition if the line segment separating any two points $\vec{\lambda}_a, \vec{\lambda}_b \in \mathcal{C}$ are also in $\mathcal{C}$-- we can write this as $t \vec{\lambda}_a + (1 - t)\vec{\lambda}_b \in \mathcal{C}$ for $t \in [0, 1]$. From the linearity of $Q(\vec{\lambda})$, we know that $Q(t \vec{\lambda}_a + (1 - t)\vec{\lambda}_b) = t Q(\vec{\lambda}_a) + (1-t) Q(\vec{\lambda}_b)$. We know that $Q(\vec{\lambda}_{a,b})$ are both positive-definite, and since $t \in [0,1]$, we know that $t, 1-t \geq 0$. Therefore, we know that both $t Q(\vec{\lambda}_a)$ and $(1-t)Q(\vec{\lambda}_b)$ are positive-definite, and so their sum, $Q(t \vec{\lambda}_a + (1-t) \vec{\lambda}_b)$ will also be positive-definite.
\end{proof}

Furthermore, we can straightforwardly see that
\begin{lemma}\label{prop:hypersphere}
Let $Q(\vec{\lambda})_{ijkl}$ be a rank-4 tensor representing the quartic part of a scalar potential, and let $Q$ be a linear function of $\vec{\lambda}$. Then, for some real number $r > 0$, $Q(r \vec{\lambda})$ is positive-definite if and only if $Q(\vec{\lambda})$ is positive-definite.
\end{lemma}
\begin{proof}
Both directions of the proof follow immediately from the linearity of $Q(\vec{\lambda})$. If $Q(\vec{\lambda})$ is positive-definite, then we know that $Q(r \vec{\lambda}) = r Q(\vec{\lambda})$ must also be positive-definite for $r > 0$, and similarly if $Q(r \vec{\lambda})$ is positive-definite, then $Q(\vec{\lambda}) = r^{-1} Q(r \vec{\lambda})$ is also positive-definite.
\end{proof}

Combining the two lemmas \ref{prop:convex} and \ref{prop:hypersphere}, we readily see that characterizing the entire strictly bounded-from-below region for a given scalar potential is simply the task of identifying a single, geodesically convex region on the surface of the unit hypersphere in $\vec{\lambda}$ space.

The convexity of the bounded-from-below region suggests a convenient feature: In order to fully characterize that region, we need merely to find a small ensemble of points within it from a random search, and then scan intelligently in the local vicinity of those points in order to locate the decision boundaries. We don't need to worry, for example, that any other valid region might exist, hugely separated in $\vec{\lambda}$ space from the convex region the classifier learns. It is therefore not unreasonable to expect that a sufficiently well-trained classifier can to good approximation characterize the \emph{entire} bounded-from-below region for a given scalar potential, and allow a model builder to easily include bounded-from-below constraints in arbitrary parameter scans even with potentials for which numerical or algebraic techniques for determining these constraints are unresolved or impractically computationally intensive. To realize such a classifier, we merely need to identify a classification technique that suits our needs and a training strategy to efficiently explore $\vec{\lambda}$.

\section{Bayesian Deep Learning and Uncertainty: A Review}\label{sec:bayesian-nn}

To answer the first of our needs, defining an appropriate classifier, we propose a Bayesian neural network \cite{mackay1992bayesian, magris2023bayesian}, a form of learner which has previously proven effective in addressing other problems arising in high energy physics, for example jet classification \cite{Araz:2021wqm,Bollweg:2019skg}, predictions for supersymmetric (SUSY) theories \cite{Kronheim:2020vct}, and analyzing galactic gamma ray observations \cite{List:2020mzd}. Since neural networks can in principle approximate any continuous function,\footnote{This is only rigorously true for infinitely wide or deep neural networks, and we will make some architecture choices later on which will further limit expressivity, but for practical purposes a finite neural network is capable of expressing essentially any decision boundary we are likely to come across.} a neural network should be extremely well-suited to learning an arbitrary classification rule for boundedness-from-below, where continuity and analyticity of the scalar potential ensure that a continuous function evaluating ``boundedness-from-below'' as a scalar function of the input quartic potential couplings, theoretically exists. In a Bayesian neural network, probability distributions rather than point estimates of the neural network's parameters are learned during training, affording greater uncertainty quantification abilities which we shall find useful. This section briefly reviews the concepts underlying a Bayesian deep neural network, outlining its suitability for our task and summarizing some important results leveraged in our analysis.

To maximize readability, the discussion here will be primarily qualitative and intuitive with as little mathematics as possible-- more detailed and mathematically rigorous discussions of Bayesian neural networks are deferred to Appendix \ref{appendix:BayesianDropout}. Before diving into the Bayesian neural network paradigm, it is useful to discuss the nature of predictive uncertainty in neural networks. Because the black-box nature of any neural network, it is virtually inconceivable that any classifier we produce will be perfectly accurate. Furthermore, even if it were, we would have no means of rigorously proving that accuracy. Therefore, in order to render our classifier useful, it must have some notion of its predictive uncertainty-- that is, given an input, the neural network must produce not only a label, but also some metric for how confident the classifier is about that label. We also note that a high degree of classifier uncertainty may stem from one of two principal sources in our problem: First, a point in $\vec{\lambda}$ space that is extremely close to the boundary between the bounded-from-below region and the unbounded region will presumably have uncertain classification because it resembles points in both possible classes. Second, a point anywhere in $\vec{\lambda}$ space may have uncertain classification simply because the model is insufficient-- either its training data is sparse in the vicinity of that point or the model isn't complex enough to perfectly capture the physics. These two sources of uncertainty-- from inherent ambiguity in the data and from insufficiency of the model-- are often classified as \emph{aleatoric} and \emph{epistemic} uncertainty, respectively.

It is clear that any classifier we wish to use for our task must have a notion of both epistemic and aleatoric uncertainty, and be capable of distinguishing between the two-- we can see this by considering how differently a model builder might consider points with significant levels of one or the other uncertainty. A point with high aleatoric uncertainty likely denotes a region that may be of physical interest, perhaps with a metastable scalar potential, for example. A point with high epistemic uncertainty, meanwhile, solely suggests a shortcoming of our training data or neural network that should trigger skepticism in the classification from a physicist, but can be corrected by expanding the training set or increasing the complexity of the model. Unfortunately, a conventional neural network-based classifier lacks any adequate machinery to track epistemic uncertainty. For clarity, we will discuss this problem in the case of a simple binary classifier. In this case, a neural network takes an input value $\mathbf{x}$ and outputs a confidence score (really, a likelihood) between 0 and 1 given as
\begin{align}\label{eq:deterministic-confidence}
    c_w (\mathbf{x}) = \frac{1}{1 + \exp{-f_w(\mathbf{x})}},
\end{align}
where $f_w$ is a real function specified by neural network's weights (trainable parameters) $w$.\footnote{Because this score ranges between 0 and 1 and represents a likelihood of the positive label $y$ given the input $\mathbf{x}$ and the weights $w$, it is often written as  $p(y | \mathbf{x}, w)$. Given that misinterpretation of the confidence as the genuine probability of the label $y$ given input $\mathbf{x}$ (that is, $p (y | \mathbf{x})$) is rife, however, we have opted for a different somewhat nonstandard notation.} This score indicates the model's confidence that $\mathbf{x}$ is in the ``positive'' class (which class is defined as ``positive'' is of course arbitrary)-- a score close to 1 indicates high confidence that $\mathbf{x}$ belongs to this class, while a score close to 0 indicates high confidence that it belongs instead to the other classification. Translating this into labels is then trivial: Points with $c_w(\mathbf{x}) > 0.5$ are classified as positive and all others are placed in the opposite class. Points with high aleatoric uncertainty are inherently ambiguous and will output confidence scores near 0.5, since similar training data points will incentivize the weights to modify predictions around the uncertain point in opposite directions. However, since the confidence score is the sole output from the neural network, it is also clear that there is no further information from our prediction that we might use to quantify epistemic uncertainty. In fact, it's been empirically demonstrated \cite{gal2016dropout} that neural networks can often produce highly confident but incorrect predictions for points that are dissimilar from any training data-- obviously, the confidence score alone not only lacks any way for us to separate aleatoric and epistemic uncertainty, it appears to disregard the latter entirely.

\subsection{Bayesian Neural Networks}

Bayesian neural networks offer a solution to this shortcoming. While a conventional deterministic neural network learns a series of point estimates for its weights, based on their maximum likelihood values given the training data, a Bayesian neural network is presented with a prior probability weight distribution, and learns a posterior distribution given the training data. As a result, a Bayesian neural network doesn't offer a deterministic point estimate for its prediction, but instead offers a probability distribution over outputs. The variance of this distribution provides a notion of epistemic uncertainty: If more data is provided, the posteriors of the model weights will become sharper (reflecting our diminishing ignorance), in turn sharpening the distribution of the prediction.

In practical terms, translating the above notions into tractable approaches for deep learning is nontrivial, and for details we refer readers to Appendix \ref{appendix:BayesianDropout} and a review such as \cite{magris2023bayesian}. In this work, we approximate Bayesian inference using a technique known as Monte Carlo dropout \cite{gal2016dropout}. Originally devised as a regularization technique for neural networks, dropout randomly sets the outputs of some neurons to zero during each pass through the network during training. In \cite{gal2016dropout}, a correspondence was proved between a conventional neural network trained with dropout and an (approximate) Bayesian neural network. The dropout-trained neural network will produce prediction distributions with the same mean and variance as the corresponding Bayesian neural network as long as dropout is also applied when making predictions. In the original proposal for Monte Carlo dropout, the dropout probability of the neurons remained a free input parameter in the training process, which required optimization for any given set of training data via simple trial and error (or in other words, repeated training of the network). Because our task involves adding new training data throughout the learning process, we instead opt for a modified approach known as concrete dropout \cite{gal2017concrete}, which treats the dropout probability of each network layer's neurons as a learnable parameter which is automatically optimized during training. The details of the implementation of concrete (and Monte Carlo) dropout are discussed in Appendix \ref{appendix:BayesianDropout}.

To glean information about the prediction distribution for an input $\mathbf{x}$, then, we only need to repeatedly query our trained network with the dropout enabled. For example, a modified confidence score for an input $\mathbf{x}$ can be attained by taking the mean of the confidence scores given in some number of trials $T$, or more succinctly
\begin{align}\label{eq:bayesian-confidence}
    \overline{c}(\mathbf{x}) = \frac{1}{T} \sum_{t=1}^{T} c_t (\mathbf{x}),
\end{align}
where $t$ here represents some specific forward pass of the input through the model, and $c_t(\mathbf{x})$ represents the output confidence score from the $t^{\textrm{th}}$ pass. Now we can consider how this confidence score might assess uncertain inputs. An input with high aleatoric uncertainty will, similar to the deterministic case, result in most passes through the neural network outputting values near 0.5, but a point about with high epistemic uncertainty (and therefore a high variance in the weights which contribute to the output) will instead give outputs which are highly certain in opposing directions, \ie, some outputs will be close to 1 and others close to 0. $\overline{c}(\mathbf{x})$, will therefore still be close to 0.5. In contrast to the confidence score given in Eq.(\ref{eq:deterministic-confidence}), then, this Bayesian version has a notion of \emph{both} aleatoric and epistemic uncertainty.

\subsection{Quantifying Uncertainty in Bayesian Deep Learning}\label{sec:uncertainty-quantification}

Having argued that the probability distributions predicted by a Bayesian neural network reflect both epistemic and aleatoric uncertainty, we can now summarize how these components are estimated in our specific application. Information theory suggests tractable estimates of both total uncertainty (that is, epistemic and aleatoric uncertainty combined) and epistemic uncertainty. A measurement of total uncertainty is the Shannon entropy \cite{shannon1948communication}, which we can estimate as
\begin{align}\label{eq:shannon-entropy}
    H(\overline{c}) = - \overline{c} \log \overline{c} - (1 - \overline{c}) \log (1 - \overline{c}),
\end{align}
where $\overline{c} = \overline{c}(\mathbf{x})$, as defined in Eq.(\ref{eq:bayesian-confidence}), for some input $\mathbf{x}$. The expression in Eq.(\ref{eq:shannon-entropy}) is maximized for $\overline{c} = 1/2$ and is precisely analogous to the Gibbs formula for entropy in thermodynamics. The epistemic uncertainty can be estimated via mutual information, given by
\begin{align}\label{eq:mutual-information}
    I(c_t) = H(\overline{c}) + \frac{1}{T} \sum_{t=1}^{T} \big[ c_t \log c_t + (1 - c_t) \log(1 - c_t) \big],
\end{align}
where $c_t$ follows the same notation as Eq.(\ref{eq:bayesian-confidence}). Formally, mutual information measures the expected information gained about the weights through knowing the label of the queried input. Because uncertainty in the weights in turn corresponds to a Bayesian neural network's measure of epistemic uncertainty, we follow the practice of \cite{depeweg2018decomposition} and identify the mutual information as an estimate of that quantity. We can see that mutual information isn't sensitive to low-confidence predictions in the individual passes $c_t$ through the neural network, as long as they're consistent: Even if $\overline{c} = 1/2$, $I(c_t) = 0$ if all $c_t = \overline{c}$. Instead, mutual information is maximized for points where there is a high degree of disagreement between different predictions of the neural network, regardless of the inherent ambiguity of the input.

We finally invoke one further candidate metric for model uncertainty, the variation ratio. This quantity is given by
\begin{align}\label{eq:variation-ratios}
    VR(\mathbf{x}) = 1 - \frac{f_\mathbf{x}}{T},
\end{align}
where $f_\mathbf{x}$ is the number of the $T$ forward passes through the neural network that specify a classification \emph{other} than the most common label for some input $\mathbf{x}$. In the Bayesian case, this approximates the likelihood that the mode label is \emph{not} the true label, given the input and the training data. This measure doesn't cleanly correspond to either epistemic or aleatoric uncertainty, but its intuitive usefulness is clear, and it is often included in studies on uncertainty in Bayesian deep learning \cite{rakesh2021efficacy,gal2017deep}. We note that, unlike mutual information, points with higher aleatoric uncertainty (that is, with predicted confidence scores all near 1/2) will tend to also have a high variation ratio, since a smaller variability of the confidence scores will result in more instances of conflicting classification if $\overline{c}$ is near 1/2 than if it instead were highly certain. At the same time, a point with high aleatoric uncertainty might have a smaller variation ratio than a point with lower aleatoric uncertainty, simply because the latter point has a greater variance in its predictions.

An important caveat to each of these metrics of uncertainty is that they are not, a priori, \emph{calibrated}, especially in active learning where the training set is disproportionately composed of inputs for which the neural network is uncertain. This means that the likelihood estimated from, \eg, the variation ratio is \emph{not} going to be generically equal to the probability that a given point drawn from some distribution is incorrectly classified. Of course, since at the time of training we generally can't know the distribution of quartic couplings in some user's phenomenological scan (which will likely have a complicated prior depending on a variety of parameters other than the quartic scalar couplings, such as physical particle masses or vev mixing parameters), the notion of a generally well-calibrated uncertainty is nonsensical here. It is feasible that an uncalibrated model produced by analyses of the type presented here can be calibrated for a particular distribution of inputs using techniques discussed in, \eg, \cite{laves2019well}, but we do not pursue this possibility here or in our public package BFBrain. Uncalibrated uncertainties, however, are hardly useless-- they contain information about the relative confidence for some input that the neural network has about its outputs, compared to other inputs. In our empirical studies, we shall find that the uncertainty metrics outlined in this section serve as an excellent predictor of the classifier accuracy for a given input point. In a numerical scan of parameter space points, therefore, we can readily apply the neural network to classify scalar potential boundedness-from-below while identifying points where the classification may be unreliable, either excluding these points from the scan or subjecting them to further (more computationally expensive) analysis.

\section{Active Learning}\label{sec:active-learning}

Having discussed the mechanics of our classifier, we can now address the training strategy we employ to explore our parameter space. A key issue with which we are presented is that the overwhelming majority of points in the full parameter space (namely the hypersphere in $\vec{\lambda}$ space) will for most scalar potentials be thoroughly uninteresting, while only a very small portion will yield bounded-from-below potentials. Furthermore, explicitly labelling any points to create training data will, by the nature of the problem, be computationally expensive. Active learning is a paradigm that efficiently addresses these problems, by dynamically generating a training set consisting of only the most informative samples from parameter space over the course of training.\footnote{For a review of active learning, see, \eg, \cite{ren2021survey}} In this section, we shall review some of the core concepts of active learning and outline the application of these concepts in the context of our particular problem.

An implementation of an active learning strategy generally consists of three components:
\begin{itemize}
    \item An \emph{oracle} computes the label for given input data. This calculation, as noted earlier, will generally be computationally expensive, and so calls to the oracle will represent a significant bottleneck in the program and should be minimized.
    \item Some sort of statistical \emph{learner} (in our case a Bayesian neural network functioning as a binary classifier) which will be trained on points labelled by the oracle. This learner will be trained on data labelled by the oracle, and then its outputs on additional unlabelled data will suggest new points for the oracle to label and incorporate into the next round of training.
    \item A \emph{query strategy} by which the trained learner suggests new points to the oracle for labelling, which are then incorporated into the training set for the next round of learning. This strategy generally selects new training points on which the classifier is highly uncertain.
\end{itemize}

In our analysis, then, the active learning program flow is as follows:

\begin{enumerate}
    \item Generate a random initial sample of $\vec{\lambda}$'s and label them with the oracle. This shall become the initial training data $P_{\textrm{train}}$.
    \item Train the classifier on $P_{\textrm{train}}$.
    \item Generate an additional sample $L$ of candidate $\vec{\lambda}$'s in the vicinity of points which the oracle has labelled as bounded-from-below.
    \item Score the points in $L$ based on the algorithm's query strategy, and add those with the top percentile of scores to $P_{\textrm{train}}$.
    \item Repeat Steps 2-4 until a predetermined number of active learning iterations have completed.
\end{enumerate}

For convenience and clarity, we have also summarized the program flow of the active learning loop in Figure \ref{fig:active-learning}. With the active learning strategy outlined, we can spend the remainder of this section describing the individual components of this strategy in greater detail-- in particular the oracle and the query strategy. The specific architecture of our learner, a Bayesian neural network of the type discussed in Section \ref{sec:bayesian-nn}, may be of less interest to the general reader and so information on its structure is located in Appendix \ref{appendix:Learner}

\begin{figure}
    \centering
    \begin{tikzpicture}
        \node at (0,2) [query, fill=blue!20] (p-id) {Initial $P_{\textrm{train}}$};
        \node at (0,0) [query, fill=red!20] (c-id) {Train \textbf{learner} on all $P_\textrm{train}$};
        \node at (8,0) [query, fill=blue!20] (l-id) {Generate unlabelled data $L$};
        \node at (8,-2) [query, fill=violet!20] (q-id) {Select $K \subset L$ using \textbf{query strategy}};
        \node at (0,-2) [query, fill=green!20] (o-id) {Label $K$ with \textbf{oracle} and add to $P_\textrm{train}$};
        \path [connector] (p-id) -- (c-id);
        \path [connector] (c-id) -- (l-id);
        \path [connector] (l-id) -- (q-id);
        \path [connector] (q-id) -- (o-id);
        \path [connector] (o-id) -- (c-id);
    \end{tikzpicture}
    \caption{A schematic diagram of the flow of the active learning program, as described in the text. }
    \label{fig:active-learning}
\end{figure}
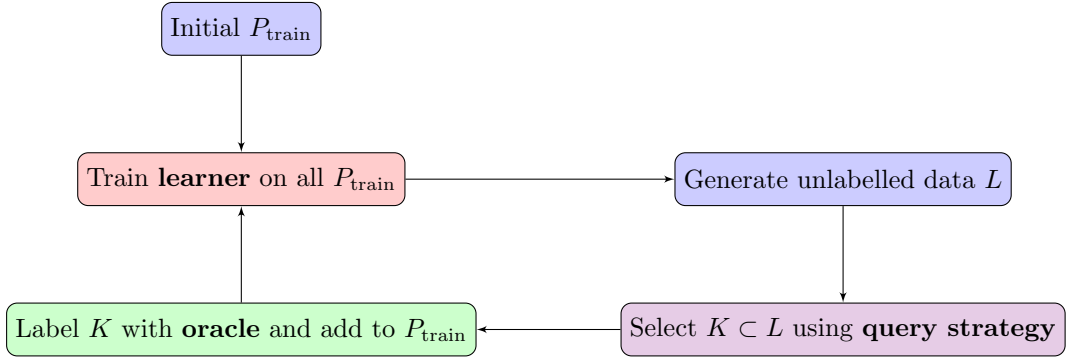

\subsection{Oracle}\label{sec:oracle}
The problem of creating an oracle to label scalar potentials as bounded from below is a somewhat nontrivial one-- in large part because the lack of a computationally efficient and highly accurate method for assigning these labels is precisely the problem that our analysis is devised to address. A possible choice would be to implement a version of the general algorithm discussed in \cite{Ivanov:2018jmz}, which determines whether a given scalar potential is bounded from below with perfect accuracy using the theory of resultants, but this method's accuracy comes at a price: Labelling a single data point for even an extremely simple model (\eg, the inert doublet model) will take hours-- a consequence of the fact that the algorithm is executed in exponential time. 

Instead, for this analysis, we opt for introducing a small degree of inaccuracy into our oracle in exchange for a substantially more time-efficient labelling procedure. A necessary and sufficient condition for satisfying the positive definiteness condition of Eq.(\ref{eq:general-BFB-cond}) (and therefore establishing boundedness from below) is the non-negativity of the potential for every value of the vev vector $\phi$ on $S^{d-1}$, the surface of the $d$-dimensional unit hypersphere, where we remind the reader that $\phi$ is a real $d$-dimensional vector parameterizing the vev configuration of the scalar field(s) in the theory. So, we can approximate this condition, theoretically arbitrarily well, by simply repeatedly locally minimizing the quartic potential on $S^{d-1}$ an arbitrary number of times with random starting points, and labelling a point as bounded from below if all the local minima found by the oracle are positive. In Algorithm \ref{alg:oracle}, we depict this simple strategy schematically. In the practical implementation of our analysis (and in the public package BFBrain), the execution of this algorithm is optimized using Jax \cite{jax2018github} and the JaxOpt constrained numerical optimization package \cite{jaxopt_implicit_diff}, leveraging these tools' automatic differentiation and parallelization capabilities to greatly speed up computation.

\begin{algorithm}[hbt!]
\caption{A schematic depiction of our approximate oracle's classification strategy. Here, $PGD(f, \phi_0)$ refers to a projected gradient descent minimization over $\phi$ on the unit hypersphere $S^{d-1}$, with starting point $\phi_0$.}\label{alg:oracle}
\KwIn{$Q(\vec{\lambda)}, \; \vec{\lambda} \in \mathbb{R}^m, \; n_{\textrm{iter}} \gsim O(100)$}
\For{$i \leq n_{\textrm{iter}}$}{
$\phi_0 \gets $ Random $\in \mathbb{R}^d$\;
$V = PGD(Q(\vec{\lambda}) \phi^4, \phi_0)$\;
\If{$V \leq 0$}{return False\;}
}
return True\;
\end{algorithm}

It is clear that formally, our approximate oracle represents necessary but not sufficient conditions for boundedness-from-below. The accuracy with which these conditions approximate necessary \emph{and} sufficient boundedness-from-below conditions is dependent on three parameters. Two of these are simply parameters for the local minimization algorithm: The maximum number of gradient descent steps taken and the error tolerance of the minimizer (\ie, how small of a gradient norm does the minimizer accept as being indicative of a local extremum). We have found that for our experiments, these parameters have a negligible effect for reasonable choices of 10000 maximum steps and a tolerance of $10^{-3}$. The third parameter which affects our algorithm's precision is the number of local minimization attempts that the oracle makes for each potential, $n_{\textrm{iter}}$-- in practice this parameter overwhelmingly controls the oracle's accuracy. Fortunately, this parameter is (at least within hardware and software constraints) arbitrarily tunable, and as $n_{\textrm{iter}} \rightarrow \infty$, the approximation will approach perfect necessary and sufficient conditions. Of course, the fact that the calculation is always approximate introduces a degree of label noise into our training data-- some points which are in reality not bounded-from-below will be mislabelled. We will investigate the effects of this noise empirically in subsequent sections.

In order to use this oracle effectively, we must obviously estimate a value of $n_{\textrm{iter}}$ which dependably excludes points which are not bounded-from-below, but remains computationally tractable. Because we wish our strategy to readily apply to any possible renormalizable scalar potential, it is not realistic to make a general statistical guarantee of the oracle reliability for given values of $n_{iter}$-- at best, we might estimate the probability that a given local minimization iteration will yield a negative minimum. Instead, we rely on the robustness of results to increasing $n_{\textrm{iter}}$ to estimate an optimal value of this parameter for each scalar potential we consider: For each scalar potential, starting from $n_{\textrm{iter}} = 50$, we repeatedly use the oracle to label the same random sample of $10^5$ points (uniformly sampled from the surface of the unit hypersphere in $\vec{\lambda}$ space), increasing $n_{\textrm{iter}}$ by 50 after each labelling attempt. We estimate that the optimal value of $n_{\textrm{iter}}$ is the value for which the oracle produces identical labels for all $10^5$ points for at least 5 consecutive iterations. The basic strategy is depicted in Algorith \ref{alg:oracle-test}, and a flexible implementation of the strategy is included in the BFBrain package.

\begin{algorithm}[hbt!]
\caption{The empirical method for testing the approximate oracle $\Omega$ and identifying the best value of $n_{\textrm{iter}}$, using $\mathbf{\Lambda}$, a list of sets of scalar quartic coefficients in $\vec{\lambda}$}\label{alg:oracle-test}
\KwIn{$n_{\textrm{iter}} = 50, \; \mathbf{\Lambda}$}
$n_{\textrm{best}} \gets 50$\;
$count \gets 0$\;
old\_labels $\gets \Omega(n_{\textrm{iter}}, \mathbf{\Lambda})$\;
$n_{\textrm{iter}} \gets n_{\textrm{iter}} + 50$\;
\While{$count < 5$}{
$new\_labels \gets \Omega(n_{\textrm{iter}}, \mathbf{\Lambda})$\;
    \eIf{new\_labels == old\_labels}{
        $count \gets count + 1$
    }{$old\_labels \gets new\_labels$\; 
      $count \gets 0$\; 
      $n_{\textrm{best}} \gets n_{\textrm{iter}}$}
$n_{\textrm{iter}} \gets n_{\textrm{iter}} + 50$\;
}
$\textrm{return}$ $n_{\textrm{best}}$
\end{algorithm}

A natural question emerges at this point-- namely, why do we not simply employ this approximate oracle directly to label points in a phenomenological scan, rather than resorting to a neural network? The first and most obvious answer to this question lies in the fact that our oracle can rapidly become enormously computationally expensive, as the number of vev parameters increases and the number of local minimizations necessary to achieve robustness increases with it. After training, the neural network's performance is \emph{entirely independent} of the oracle's computational cost-- therefore, depending on the model builder's requirements for precision, an oracle of arbitrary expense may be used without affecting the computational efficiency of the neural network in regular use-- one may even use the NP-hard algorithm of \cite{Ivanov:2018jmz} as an oracle. Furthermore, in the case of a noisy (that is, suboptimal $n_{\textrm{iter}}$) oracle, the oracle itself lacks any capacity for uncertainty estimation. We shall find empirically that even in the presence of significant label noise (that is, with $n_{\textrm{iter}}$ such that $\sim 10 \%$ of the points labelled as bounded-from-below are false positives), metrics for model and predictive uncertainty on data outside of the training set remain effective indicators of the reliability of a prediction. We also find in some circumstances that a neural network trained on such noisy data achieves consistently better performance than the noisy oracle itself. Finally, the use of a neural network to classify boundedness-from-below in turn makes certain aspects of exploring the parameter space of the scalar potential considerably simpler: For example, the neural network's output now constitutes a differentiable function describing the boundedness-from-below of a scalar potential.

\subsection{Query Strategy}\label{sec:query-strategy}

The efficacy of our strategy for analyzing scalar potentials will clearly be highly dependent on the manner in which the training set is constructed. To that end, here we discuss our strategy for generating our training data in significant detail.
First, before active learning can take place, we need an initial (small) set of labelled training data, which we shall call $P_{\textrm{train}}$. To generate the initial $P_{\textrm{train}}$, we begin by uniformly sampling 1000 points from the unit hypersphere in $\vec{\lambda}$ space, and querying our oracle about their labels. In all of the scalar potentials we consider here (and likely in most scalar potentials that might be of interest), the overwhelming majority of the points generated in this manner will \emph{not} be bounded-from-below. To avoid incentivizing our neural network to simply label everything as not bounded-from-below, we must then either reweight the points which \emph{are} bounded-from-below or augment our training data with considerably more bounded-from-below points. We can leverage the convexity of the bounded-from-below region to readily do the latter: By randomly sampling points on the line segments between the known bounded-from-below points in $P_{\textrm{train}}$ (and then projecting these back onto the unit hypersphere), we can generate an arbitrary number of additional bounded-from-below points without any further need to label them (as long as we can be confident that our oracle has sufficiently small noise-- we shall investigate this effect later).\footnote{Strictly speaking, we could achieve a greater diversity of points by randomly sampling positive linear combinations of bounded-from-below points, but in practice this minor generalization would not be especially meaningful, since points generated by this procedure represent a small fraction of the training set once a sufficient number of active learning iterations have been performed.} Adding the newly-generated bounded-from-below points to $P_{\textrm{train}}$, we then have a balanced initial training set.

After generating the initial $P_{\textrm{train}}$, our query strategy must also identify additional points to be added to the training set during each active learning iteration. We shall be considering a paradigm known as pool-based active learning, where a pool of new unlabelled points $L$ are proposed to the classifier, and then the trained classifier selects some subset of them to be labelled and added to the training set. The next task for our query strategy, then, is to generate $L$. To begin, we consider which points are of interest in our task: Points in and near the single convex region in $\vec{\lambda}$ space in which the potential is bounded from below. Ideally, then, our pool should focus on this region already, rather than uniformly sampling the entire parameter space. Similar to \cite{goodsell2023active, hammad2023exploration}, we accomplish this emphasis by generating $L$ from sampling in the vicinity of bounded-from-below points in the training set. Specifically, in order to generate a new point from an existing training point $p$, we rotate $p$ in a random direction (uniformly sampled from all possible directions in $\vec{\lambda}$ space) by an angle (in radians) $\delta$, where $\delta$ is randomly chosen from a normal distribution $\mathcal{N}(0, \Delta^2)$. The initial points $p$ are randomly selected (with uniform probability) from the bounded-from-below points in $P_{\textrm{train}}$.

The procedure for generating $L$ leaves only a single free parameter that we must select: The scale of the rotation angles $\Delta$. To estimate an appropriate $\Delta$, we note that we wish $L$ to be, to good approximation, a pool of points drawn from a region that at least somewhat tightly encompasses the entire bounded-from-below region of $\vec{\lambda}$ space, while not oversampling points which are far from this region and therefore of little interest to us. We should therefore anticipate that an appropriate $\Delta$ value should equate to a length scale characteristic of the size of the bounded-from-below region-- then, we would anticipate that $L$ will likely (after several active learning iterations to expand the initially small pool of bounded-from-below points) sample a region that approximately covers the entire bounded-from-below region in parameter space and its immediate vicinity, without needing to sample from the entire parameter space. To estimate the characteristic length scale of the bounded-from-below region, we can simply approximate it (very roughly) as the surface of the hypersphere subtended by a single angle (which we shall suggestively also refer to as $\Delta$). Then, if a fraction $f$ of uniformly-distributed random points on the $\vec{\lambda}$ hypersphere is bounded-from-below, we can estimate that $\Delta$ satisfies the equation
\begin{align}\label{eq:delta-eq}
    f = \frac{\sqrt{\pi} \Gamma(\frac{n}{2})}{2 \Gamma(\frac{n + 1}{2})} (\sin \Delta)^{n - 1} \, {}_2 F_{1}  \bigg( \frac{1}{2}, \frac{n - 1}{2}; \frac{n + 1}{2}; \sin^2 \Delta \bigg), \; n \equiv \textrm{dim}(\vec{\lambda}),
\end{align}
where ${}_2 F_1(a,b;c;z)$ denotes the ordinary hypergeometric function and $\Gamma(x)$ is the Euler gamma function, as long as $f \leq 1/2$ (which for practical problems will almost certainly always be the case). Since $f$ can always be estimated by finding the fraction of the initially generated training points that are bounded-from-below (before we rebalance the initial training data by adding additional bounded-from-below points), we therefore can estimate $\Delta$ by numerically solving the above expression.

We should note that the expression in Eq.(\ref{eq:delta-eq}), being based on somewhat unrealistic assumptions about the geometry of the bounded-from-below region and predicated on an initial training sample that can give a highly uncertain measurement of $f$ (for example, if the training sample includes only $O(\textrm{several})$ bounded-from-below points), gives us only a very approximate characterization of the optimal value for $\Delta$. In general, it only suggests an order of magnitude. However, because we only require $L$ to provide adequate coverage of the full bounded-from-below region (because our query strategy will then identify which points in $L$ are most informative regardless of the pool's distribution), this approximate knowledge is all that is required-- in other words, we should anticipate that our learner's performance should be robust against $O(1)$ modifications of $\Delta$ (in fact, cursory experiments have borne this expectation out). We find that this strategy yields good results for the scalar potentials that we consider here, which suggests its applicability in a broader range of scalar potentials that users of our public code may want to analyze.

With a pool of candidate points $L$ generated, the final step of our query strategy is to identify the points in $L$ which are most informative, so that they can be labelled and added to $P_{\textrm{train}}$ for the next iteration of active learning. Of course, ``most informative'' is hardly a rigorously defined term, so we must arrive at a definition that yields a performant classifier after training. A common criterion for the informativeness of a point in active learning scenarios is the degree of uncertainty the classifier has about that point's label-- fortunately in the Bayesian neural network paradigm, we have ample metrics of uncertainty, discussed in Section \ref{sec:uncertainty-quantification}. Therefore, given $L$ if we wish to add $k$ points to our training data, we can select the $k$ points from $L$ which have the highest uncertainty measure based on one of the metrics discussed in that Section. This line of thinking leads us to four selection criteria, all of which we will consider in our analysis:
\begin{itemize}
\item \textbf{Maximum Entropy:} Points with the largest Shannon entropy, defined in Eq.(\ref{eq:shannon-entropy}), are selected.
\item \textbf{Bayesian Active Learning by Disagreement (BALD):} Points with the largest mutual information, defined in Eq.(\ref{eq:mutual-information}), are selected. This strategy's utility in active learning problems was discussed (and the name was coined) in \cite{houlsby2011bayesian}.
\item \textbf{Variation Ratios:} Points with the largest variation ratios, defined in Eq.(\ref{eq:variation-ratios}), are selected.
\item \textbf{Random:} As a control to gauge the efficacy of our other strategies, points are assigned a random ``uncertainty'' score (in reality just a random number sampled from the uniform distribution between 0 and 1), and those points with the highest scores are selected.
\end{itemize}

With each active learning iteration we add $5 \times 10^3$ new points to $P_{\textrm{train}}$, out of a pool $L$ consisting of $5\times 10^5$ points generated as described in this Section. In contrast to our approach when generating our initial training data, we do \emph{not} augment the newly added data with additional positive points to balance the data set labels-- in practice we find that approximate parity between the label classes is preserved in any event.

\section{Experiments}\label{sec:experiments}

In this Section, we shall present the results of employing our procedure for training a bounded-from-below classifier with several example scalar potentials.
We find with all examples that our methodology results in consistently high-performance classifiers that exhibit accuracy likely to be sufficient for most parameter point scans, as well as robust uncertainty estimates that permit more careful evaluation of the points most likely to be incorrectly labelled.

\subsection{Experiments: Scalar Potentials}\label{sec:scalar-potentials}

Before presenting our results, we must, of course, specify some scalar potentials to analyze with our techniques. For our purposes here, we shall only consider potentials with SM-like (that is, $SU(2)_L \times U(1)_Y$) gauge symmetry-- this is for simplicity and because much of the work on nontrivial boundedness-from-below conditions for scalar potentials has been done in this regime, for example regarding multi-Higgs doublet models. We stress, however, that the techniques outlined in this paper and implemented in our public code are theoretically applicable to \emph{any} renormalizable scalar potential with any symmetry group. Having limited ourselves to $SU(2)_L \times U(1)_Y$ potentials, then, we further narrow our considerations by selecting three different classes of scalar potentials to analyze in detail.
The first of these is the most general Two-Higgs doublet model (2HDM), where we write the quartic part of the scalar potential as 
\begin{align}\label{eq:V-2HDM}
    V^{(4)}_{\textrm{2HDM}} = &\frac{\lambda_1}{2} |H_1|^4 +  \frac{\lambda_2}{2} |H_2|^4 + \lambda_3 |H_1|^2 |H_2|^2 + \lambda_4 |H_1^\dagger H_2|^2\\ 
    & + \bigg[ \frac{\lambda_5}{2} (H_1^\dagger H_2)^2 + \lambda_6 |H_1|^2 (H_1^\dagger H_2) + \lambda_7 |H_2|^2 (H_1^\dagger H_2) + h.c. \bigg], \nonumber
\end{align}
$H_1$ and $H_2$ are two complex $SU(2)_L$ doublets. Because $\lambda_5$, $\lambda_6$, and $\lambda_7$ are all complex parameters, there are a total of 10 real quartic coefficients, or in our terminology the $\vec{\lambda}$ space which characterizes the potential is 10-dimensional. Meanwhile, after leveraging gauge invariance we can see that a given vev configuration in the model has 5 independent real parameters. Next, we consider a three-Higgs doublet model (3HDM) with a $Z_2 \times Z_2$ discrete symmetry imposed, initially proposed in \cite{Weinberg:1976hu} as a model with CP violation in the scalar sector and suppressed flavor-changing neutral currents. The quartic part of the potential function in this case is given by
\begin{align}\label{eq:V-3HDM}
    V^{(4)}_{\textrm{3HDM}} = &\lambda_1 |H_1|^4 + \lambda_2 |H_2|^4 + \lambda_3 |H_3|^4 + \lambda_4 |H_1|^2 |H_2|^2 + \lambda_5 |H_1|^2 |H_3|^2\\
    &+ \lambda_6 |H_2|^2 |H_3|^2 + \lambda_7 |H_1^\dagger H_2|^2 + \lambda_8 |H_1^\dagger H_3|^2 + \lambda_9 |H_2^\dagger H_3|^2 \nonumber\\
    &+ \frac{1}{2}\bigg[ \lambda_{10} (H_1^\dagger H_2)^2 + \lambda_{11} (H_1^\dagger H_3)^2 + \lambda_{12} (H_2^\dagger H_3)^2 + h.c.\bigg], \nonumber
\end{align}
where now the Higgs doublets $H_1$ and $H_2$ are joined by a third Higgs doublet $H_3$. Since $\lambda_{10}$, $\lambda_{11}$, and $\lambda_{12}$ are all complex parameters, the $\vec{\lambda}$ space for this potential is 15-dimensional. Meanwhile, a vev configuration for the model is entirely specified by 9 independent real parameters. Finally, we consider the ``precustodial'' variant of the Georgi-Machacek (GM) model \cite{Georgi1985DoublyCH} as presented in \cite{Moultaka:2020dmb}, in which the usual custodial $SU(2)$ of the GM model is omitted.\footnote{This may occur, for example, if custodial symmetry-violating terms are generated at the loop level.} This potential is given as
\begin{align}\label{eq:V-PC}
    V^{(4)}_{PC} = &\frac{\lambda_1}{4} |H|^4 + \frac{\lambda_2}{4} (Tr A^\dagger A)^2 + \frac{\lambda_3}{4} Tr (A^\dagger A)^2 + \frac{\lambda_4}{4!} [Tr B^2]^2 + \lambda_5 |H|^2 Tr A^\dagger A \\
    &+ \lambda_6 H^\dagger A A^\dagger H + \frac{\lambda_7}{2} |H|^2 Tr B^2 + \frac{\lambda_8}{2} (Tr A^\dagger A)(Tr B^2) + \frac{\lambda_9}{2} (Tr A B)(Tr A^\dagger B) \nonumber\\
    &+\frac{i \lambda_{10}}{2} (H^T \sigma^2 A^\dagger B H - H^\dagger B A \sigma^2 H^*), \nonumber
\end{align}
where $H$ is the SM Higgs doublet, and $A$ and $B$ are real and complex triplets of $SU(2)_L$, respectively. Since all the coefficients in Eq.(\ref{eq:V-PC}) are real, the $\vec{\lambda}$ space here is 10-dimensional, while a vev configuration is fully specified by 10 real parameters. For the convenience of the reader, we have collected key information on our three scalar potentials in Table \ref{tab:potentials}.

\begin{table}[]
    \centering
    \begin{tabular}{| c | c | c | c |}
        \hline
        Potential &  Equation & Quartic Couplings & Vev Components\\
        \hline
        2HDM & Eq.(\ref{eq:V-2HDM}) & 10 & 5\\
        \hline
        3HDM & Eq.(\ref{eq:V-3HDM}) & 15 & 9\\
        \hline
        Precustodial & Eq.(\ref{eq:V-PC}) & 10 & 10\\
        \hline
    \end{tabular}
    \caption{\footnotesize The three different scalar potentials that we consider in our experiments for this work, along with the equations giving their scalar potential values, the number of independent real quartic couplings, and the number of independent real parameters to specify a vev in each model.}
    \label{tab:potentials}
\end{table}

Several remarks are in order regarding the scalar potentials we have chosen for this exploration. First, in spite of not extending the SM gauge group, all three of our potentials have highly nontrivial conditions for boundedness-from-below.
Of the three, only the 2HDM potential of Eq.(\ref{eq:V-2HDM}) has known exact symbolic bounded-from-below conditions, first presented in \cite{Maniatis:2006fs} and later expressed compactly as conditions on the eigenvalues of a Minkowski matrix in \cite{Ivanov:2006yq}. We have found these exact conditions to be a useful cross-check to our results for this potential, where we find that our oracle delivers perfect accuracy on all of our validation sets, and misclassifies only a handful of the $O(10^5)$ points used in training sets, which by the nature of active learning will be inherently more ambiguous. In the case of the 3HDM potential, only sufficient, but not necessary, conditions are known precisely (derived in \cite{Grzadkowski:2009bt}, with an alternative set discussed in \cite{Boto:2022uwv}, the latter of which applies only to real $\lambda_{10-12}$), in spite of the model having been originally proposed nearly 50 years ago. Partially resolved symbolic expressions for boundedness-from-below of the precustodial potential of Eq.(\ref{eq:V-PC}) are derived in \cite{Moultaka:2020dmb,Chen:2023ins}, but elements of the procedure the authors have derived still require establishing the positive-definiteness of a system of quartic polynomials, albeit a lower-dimensional one than the full space of possible vev configurations. 
The 3HDM and precustodial potentials therefore exhibit interesting use cases for the procedure we are exploring in this work-- closed-form symbolic expressions for their bounded-from-below conditions are unknown, and, because of the $O(10)$ number of parameters specifying a vev configuration in these models, unlikely to be tractable. Meanwhile, the 2HDM, although its bounded-from-below conditions have been solved, allows us a reassuring validation of the procedure's performance in one of the most complicated scenarios in which these conditions have been fully resolved.

\subsection{Experiments: Classifier Performance on Uniformly Sampled Test Sets}\label{sec:exp-performance}
For our first experiments, we simply implement our strategy, as described in Section \ref{sec:active-learning}, to create classifiers for the three scalar potentials we are considering. To get a better sense of their performance, for each potential we have performed the analysis five times (with five different random number seeds for generating and labelling training and validation data) and present the mean and variance of each performance metric we depict here. For validation data, each trial uses $10^6$ points sampled uniformly from the $\vec{\lambda}$ space hypersphere of the corresponding scalar potential. With an eye toward practical applications of our proposed analysis procedure, we note that this method of validation, which relies on labelling a large number of instances which aren't used in training, may be inadvisable or impractical if a particularly computationally expensive oracle is used, such as the NP-hard exact algorithm in \cite{Ivanov:2018jmz}. In such cases, a user can estimate improvements in $F_1$ score (that is, $\Delta F_1$) on \emph{unlabelled} data by measuring the agreement between classifier predictions after successive active learning iterations, following \cite{altschuler2019stopping}. As our current oracle is efficient enough to permit large validation sets, we do not explore this possibility here, but the BFBrain package contains multiple methods of tracking model performance in the absence of a labelled validation set, and we refer a curious reader to the package's documentation.

Returning to the parameters of our current experiment, for our results in this Section, we have performed active learning for 20 iterations for the 2HDM and precustodial potential, and for 40 iterations for the 3HDM potential, which we find produces highly performant classifiers. Our experiments will involve a variety of training hyperparameters-- for convenience we summarize the most relevant ones in Table \ref{tab:default-hyperparams}, as well as including ``default values'' for these parameters, which our training experiments will use unless otherwise specified. Beginning our experiments, in Figures \ref{fig:2HDM-F}, \ref{fig:3HDM-F}, and \ref{fig:Precustodial-F}, we get a sense for the overall performance of the method on our different scalar potentials by plotting the $F_1$ score, precision, and recall of the classifier on validation data, evaluated after each active learning iteration. To get a better idea of the quality of uncertainty quantification, these figures also depict the same quantities, evaluated on data sets where the most uncertain inputs (as judged by mutual information, which we find to be the most effective discriminator between correctly and incorrectly classified points) are omitted from the validation data-- specifically, we remove the inputs within the top 5\% of mutual information estimates within their predicted class.\footnote{We find that in general, due to the bounded-from-below points being comparatively clustered close together in $\vec{\lambda}$ space, points with this classification tend to have somewhat larger mutual information values than points which aren't bounded-from-below, at least for validation points drawn uniformly from the $\vec{\lambda}$ hypersphere. Therefore it is somewhat more useful to consider relative uncertainties of points separately for each predicted label class rather than computing uncertainty quantiles across all inputs at once. We emphasize, however, that by segregating quantile computation based on the classifier's prediction, we can perform this same discrimination with unlabelled data sets that might be encountered in, \eg, a parameter space scan.} For clarity, we have also included corresponding tables of the final values achieved for each classifier metric at the end of training, which provides a more quantitative picture of our results-- Table \ref{tab:fig-results-1} outlines the final performance metrics achieved for classifiers of all three models on full sets of validated data, while Table \ref{tab:fig-results-2} depicts the same performance metrics after points with a high mutual information score are removed from the validation sets in the manner that we have described.

\begin{table}[]
    \centering
    \begin{tabular}{| c | c | c |}
        \hline
        Parameter &  Definition & Default\\
        \hline
        $l$ & Length scale of weight prior $\mathcal{N}(0, l^{-2})$ (see Appendix \ref{appendix:Learner})& 0.1\\
        \hline
        Epoch Patience & Number of epochs without loss improvement before Adam terminates & 100\\
        \hline
        Layers & Number of hidden neural network layers (with 128 neurons each) & 5\\
        \hline
        $n_{\textrm{iter}}$ & Oracle accuracy hyperparameter (see Section \ref{sec:oracle}) & \begin{tabular}{@{}c@{}}100 (2HDM,3HDM) \\ 250 (Precustodial)\end{tabular}\\
        \hline
    \end{tabular}
    \caption{\footnotesize The most relevant hyperparameters for our active learning experiments in this Section. In our analysis, the value in the ``Default'' column for each parameter is used unless we specify otherwise.}
    \label{tab:default-hyperparams}
\end{table} 

\begin{figure}
    \includegraphics[width=6in]{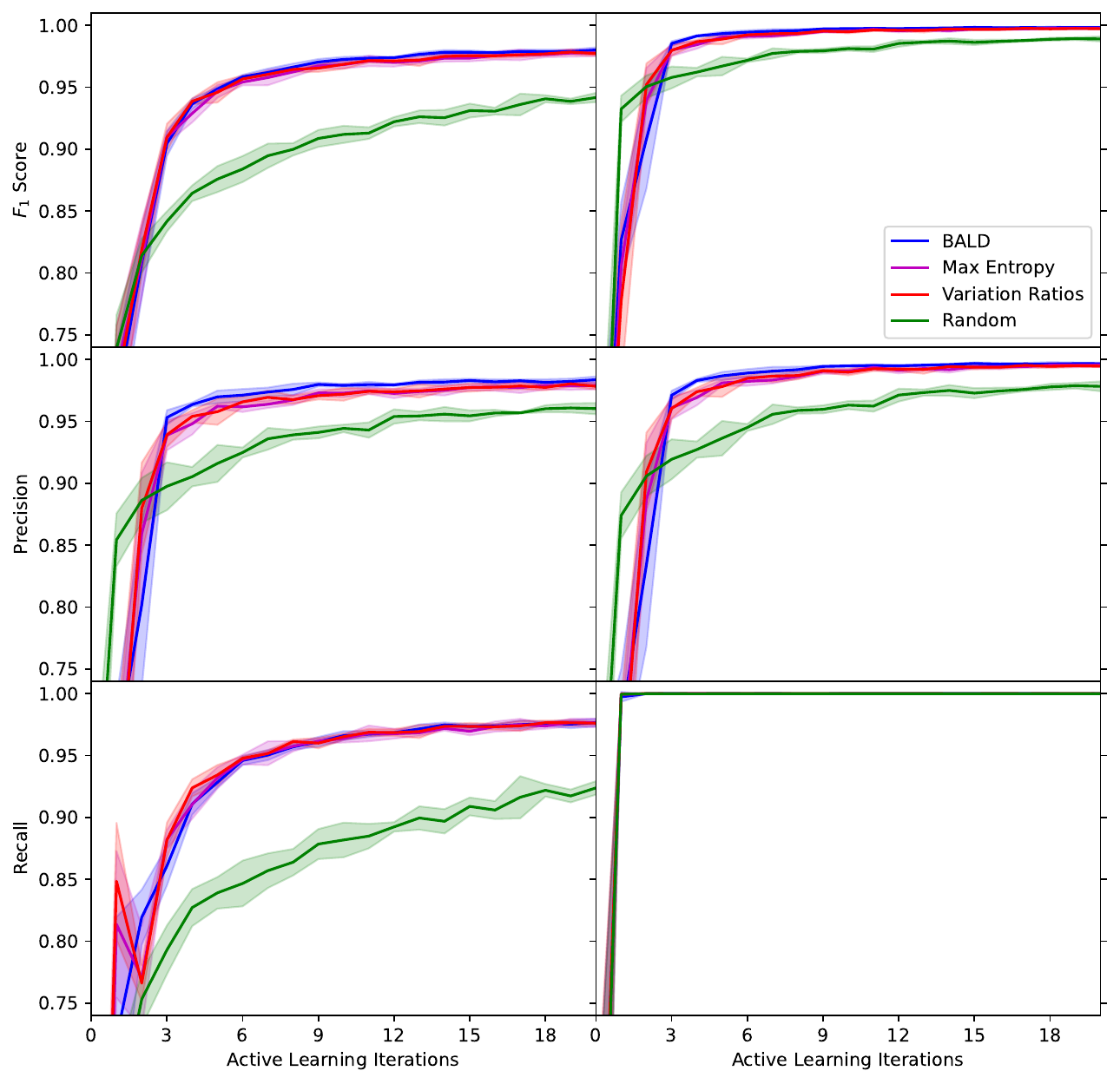}
    \caption{For the general 2HDM potential: (Left) The $F_1$ score (top), precision (middle), and recall (bottom) of the boundedness-from-below classifier with 5 hidden layers of 128 neurons each for BALD (blue), maximum entropy (magenta), variation ratios (red), and random (green) acquisition functions,
    as a function of the number of active learning iterations performed, recorded over the course of executing the active learning loop as described in Section \ref{sec:active-learning}. Each experiment is performed 5 times with different starting weights, initial training data, and initial validation data; the lines depicted represent the mean performance of all 5 trials, with the standard deviation being depicted as the transparent filled regions. Final classifier performances are listed in Table \ref{tab:fig-results-1}. (Right): As on the left, but with the validation set altered by removing the points for which the classifier predicts a mutual information (defined in Eq.(\ref{eq:mutual-information})) score greater than  $95^{\textrm{th}}$ percentile of all points with the same predicted classification. Final classifier performances on validation data subject to these restrictions are listed in Table \ref{tab:fig-results-2}.}
    \label{fig:2HDM-F}
\end{figure}
\begin{figure}
    \includegraphics[width=6in]{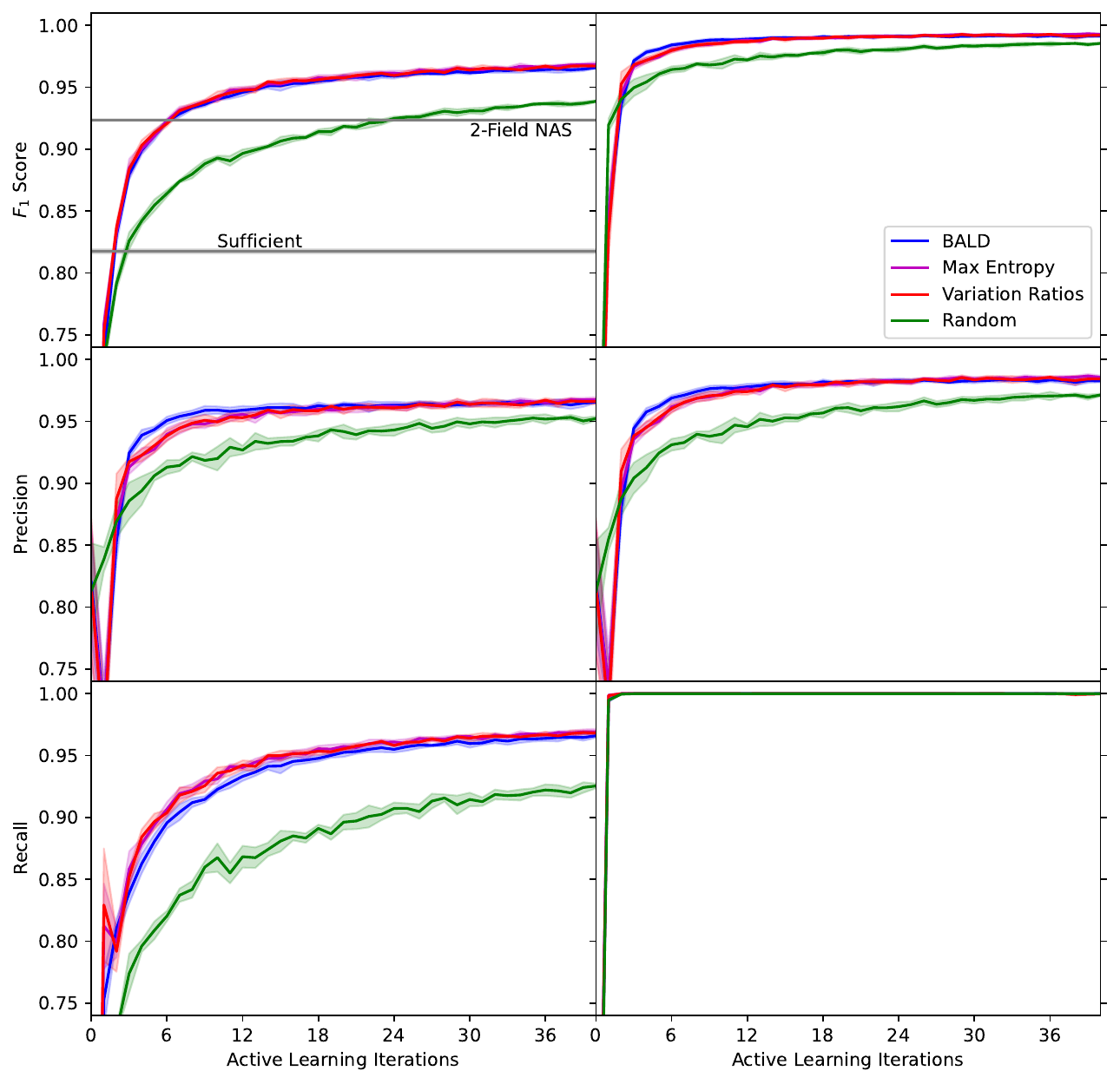}
    \caption{As Figure \ref{fig:2HDM-F}, but for the 3HDM potential given in Eq.(\ref{eq:V-3HDM}). Final performances for the data depicted in the left (right) column are listed in Table \ref{tab:fig-results-1}(\ref{tab:fig-results-2}). Notice that due to slower convergence, we have allowed active learning to continue for 40 iterations here, rather than the 20 iterations considered for the other potentials. For comparison, we have included the $F_1$ score on the validation data of the sufficient conditions of \cite{Grzadkowski:2009bt}, assuming that our oracle labels are accurate, as well as the $F_1$ score from applying the necessary and sufficient conditions for boundedness-from-below if only 2 of the 3 fields are allowed to achieve nonzero vev's simultaneously. As with other lines the value and uncertainty of these $F_1$ scores are taken as the mean and standard deviation of 5 independent experiments (in this case the 5 independent validation data sets).}
    \label{fig:3HDM-F}
\end{figure}
\begin{figure}
    \includegraphics[width=6in]{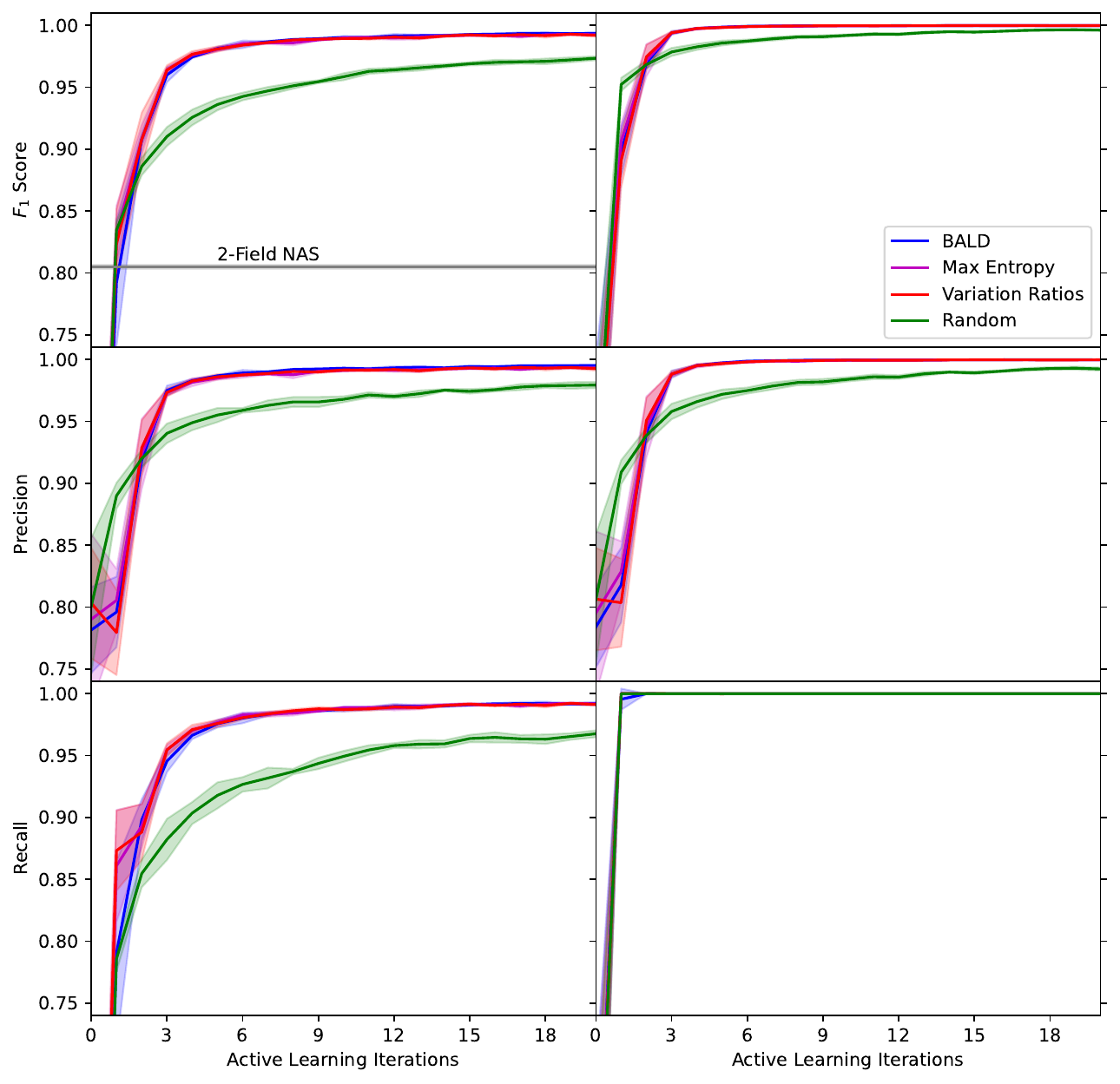}
    \caption{As Figure \ref{fig:2HDM-F}, but for the precustodial GM potential given in Eq.(\ref{eq:V-PC}). Final performances for the data depicted in the left (right) column are listed in Table \ref{tab:fig-results-1}(\ref{tab:fig-results-2}). For comparison, we have included the $F_1$ score on the validation data of the symbolic necessary and sufficient conditions for boundedness-from-below, assuming that only 2 of the 3 fields achieve nonzero vevs at a time (these may be readily extracted from \cite{Moultaka:2020dmb}). As in Figure \ref{fig:3HDM-F} we have assumed that our oracle labels are accurate, and the value and uncertainty of this $F_1$ score is taken as the mean and standard deviation of 5 independent experiments.}
    \label{fig:Precustodial-F}
\end{figure}

\begin{table}[]
    \centering
    \begin{tabular}{| c | c | c | c | c |}
        \hline
         Potential & Query & $F_1$ Score & Precision & Recall\\
        \hline
        \multirow{4}{*}{2HDM} & BALD & 0.980(2) & 0.984(3) & 0.977(3) \\
        & Max Entropy & 0.978(2) & 0.979(3) & 0.976(4)\\
        & Variation Ratios & 0.977(2) & 0.978(2) & 0.976(4)\\
        & Random & 0.942(4) & 0.960(5) & 0.924(5)\\
        \hline
        \multirow{4}{*}{3HDM} & BALD & 0.9658(7) & 0.9656(6) & 0.966(2) \\
        & Max Entropy & 0.968(2) & 0.968(2) & 0.969(2)\\
        & Variation Ratios & 0.967(2) & 0.966(2) & 0.968(1)\\
        & Random & 0.9387(8) & 0.952(2) & 0.925(2)\\
        \hline
        \multirow{4}{*}{Precustodial} & BALD & 0.9936(7) & 0.9952(7) & 0.992(1) \\
        & Max Entropy & 0.9925(4) & 0.9932(6) & 0.9918(7)\\
        & Variation Ratios & 0.992(1) & 0.993(2) & 0.991(1)\\
        & Random & 0.973(2) & 0.979(3) & 0.967(3)\\
        \hline
    \end{tabular}
    \caption{\footnotesize The $F_1$ score, precision, and recall achieved by fully trained classifiers for different scalar potentials and active learning query strategies on their full validation data sets, depicted in the left columns of Figures \ref{fig:2HDM-F}-\ref{fig:Precustodial-F}. Means and uncertainties computed by averaging the results of five independent trials.}
    \label{tab:fig-results-1}
\end{table}

\begin{table}[]
    \centering
    \begin{tabular}{| c | c | c | c | c |}
        \hline
         Potential & Query & $F_1$ Score & Precision & Recall\\
        \hline
        \multirow{4}{*}{2HDM} & BALD & 0.9983(5) & 0.997(1) & 1.0(0) \\
        & Max Entropy & 0.9978(8) & 0.996(2) & 1.0000(1)\\
        & Variation Ratios & 0.9978(8) & 0.9947(9) & 1.0(0)\\
        & Random & 0.989(2) & 0.978(4) & 1.0(0)\\
        \hline
        \multirow{4}{*}{3HDM} & BALD & 0.9913(6) & 0.983(1) & 0.99998(3) \\
        & Max Entropy & 0.9927(7) & 0.986(1) & 0.99990(6)\\
        & Variation Ratios & 0.9920(9) & 0.984(2) & 0.99983(8)\\
        & Random & 0.9855(6) & 0.971(1) & 1.0(0)\\
        \hline
        \multirow{4}{*}{Precustodial} & BALD & 0.99982(9) & 0.9998(1) & 0.99988(7) \\
        & Max Entropy & 0.99982(4) & 0.99974(6) & 0.99990(5)\\
        & Variation Ratios & 0.99986(3) & 0.9998(1) & 0.99991(7)\\
        & Random & 0.9962(6) & 0.992(1) & 1.0(0)\\
        \hline
    \end{tabular}
    \caption{\footnotesize As Table \ref{tab:fig-results-1}, but reflecting the performance on the validation sets after removing all points with mutual information greater than the $95^{\textrm{th}}$ percentile of points in their predicted classification, corresponding to the charts in the right column of Figures \ref{fig:2HDM-F}-\ref{fig:Precustodial-F}.}
    \label{tab:fig-results-2}
\end{table}

From these Figures and Tables, we can already glean a number of interesting characteristics of our method. First, for all three scalar potentials considered, the classifier achieves significant accuracy. Omitting the random query strategy (which should by design be inferior to all active learning strategies we employ) we find $F_1$ scores ranging from above 0.96 for the 3HDM potential to in excess of 0.99 for the precustodial potential. In turn, this indicates that the classifiers uniformly exhibit both comprehensive (that is, few false negatives) and precise (that is, few false positives) coverage of the bounded-from-below parameter space. 

It is important to emphasize that in the case of scalar potentials for which symbolic necessary and sufficient bounded-from-below conditions are unknown or intractable, the methodology here \emph{substantially} outperforms the more conventional techniques for approximating bounded-from-below conditions we have considered here. Applying sufficient conditions to the 3HDM potential or applying the necessary and sufficient conditions in the precustodial potential with only two fields with nonzero vevs both achieve $F_1$ scores on the validation sets of not far in excess of 0.8-- in the case of the former, this stems from a significant number of false positives, while in the case of the latter, it stems from false negatives. Applying the two-field necessary and sufficient symbolic conditions to the 3HDM potential results in somewhat improved performance, with an $F_1$ score in excess of 0.92, but our procedure outperforms even these results significantly. Meanwhile, for phenomenological studies the time necessary to classify large numbers of points is negligible: On a personal laptop with an Nvidia GeForce GTX 1660 Ti GPU, the 2HDM classifier is capable of evaluating $10^5$ inputs with 100 forward passes through the network each in $\sim 0.7$ seconds-- more than 5 times faster than even our implementation of the symbolic bounded-from-below conditions of \cite{Ivanov:2006yq}. So, we note that our methodology is capable of achieving enormously more accurate results than approximate symbolic bounded-from-below conditions while requiring comparable computation time after training, and can be applied in theory to any renormalizable scalar potential. While insufficient for work that requires extremely high precision (\eg, identifying a region of metastability for the potential), all three of our neural network classifiers can serve as excellent and efficiently evaluated expressions of the approximate bounded-from-below conditions for the purposes of a first-pass phenomenological parameter space scan.

Beyond simply noting the high performance of the models, we also see from these figures that the methodology outlined here is quite robust to both differing starting conditions and the choice of query strategy. To the former point, we see that while the models can evince significantly varying performance among the trials with different initial training data (and validation data) in early active learning iterations, these variances quickly converge to the sub-percent-level well before active learning terminates. To the second point, we see that our three uncertainty-motivated query strategies outlined in Section \ref{sec:query-strategy} all significantly outperform the baseline random query strategy, but there is little discrepancy among the results from the three strategies themselves-- especially as active learning continues and performance begins to plateau. Moreover, although mutual information is clearly an effective indicator of the reliability of the results, particularly in the case of false negatives, there seems to be no significant difference in this quantity's discriminating power among the neural networks trained with different query strategies.

As a final point we draw from Figures \ref{fig:2HDM-F}-\ref{fig:Precustodial-F}, we see that the performance of the classifiers vary significantly for different potentials: After 40 rounds of active learning, a 3HDM classifier trained with the BALD query strategy achieves an $F_1$ score of 0.9658 $\pm 0.0007$, while a 2HDM classifier with the same query strategy achieves an $F_1$ score of 0.9801 $\pm 0.0021$ after just 20 rounds, and the analogous classifier for the precustodial potential achieves 0.9936 $\pm 0.0007$. The scaling behavior of this performance with increasing $\vec{\lambda}$ space dimensionality or the number of independent vev parameters (the two obvious metrics for the complexity of a given scalar potential) is unclear-- on one hand, the 3HDM potential, with 15 real quartic coupling coefficients, demonstrates significantly poorer classifier performance than the 10-coefficient 2HDM and precustodial potentials, but on the other hand, the discrepancy between the 2HDM and precustodial classifiers' performances roughly equates the discrepancy between the 3HDM and 2HDM performances. Furthermore, the best-performing classifiers are those trained on the precustodial potential, in spite of the fact that this scalar potential has more free vev parameters than the 2HDM, and the same number of quartic coefficients. It is clear, then, that at least in the regime we have considered for these experiments (namely, potentials with $\lesssim O(10)$ independent quartic coupling coefficients and a similar number of free vev parameters), the efficacy of our procedure in characterizing the bounded-from-below region can vary somewhat (although remaining, at least for our examples, uniformly high) in a manner that is not obviously predictable. 

While a detailed exploration of the causes of these differing performance outcomes is beyond the scope of the current work, we can do some further investigation of this phenomenon by eliminating certain possible causes. In Figure \ref{fig:layer-comparison}, we depict the $F_1$ scores over the course of active learning for differing numbers of hidden layers in our network, quoting the final results for the trained networks in Table \ref{tab:layer-comparison}. If the performance of the classifier were limited solely by the capacity for the neural network architecture to learn the bounded-from-below decision rule, we might expect that a deeper network, with its larger number of weights, would achieve superior performance than a more shallow one. However, we see that for the 3HDM and the precustodial model, a shallower 3-layer network actually achieves the same (or incrementally superior) performance to a 5- and 7-layer architecture-- while suffering some instability in the quality of its results for the 2HDM potential. This result then suggests that the underperformance of our methodology on the 3HDM is not a product of an overly simple neural network.

\begin{table}[]
    \centering
    \begin{tabular}{| c | c | c | c |}
        \hline
          & 3HDM & 2HDM & Precustodial\\
        \hline
        3 Layers & 0.967(3) & 0.979(5) & 0.9938(4)\\
        5 Layers & 0.9658(7) & 0.980(2) & 0.9936(7)\\
        7 Layers & 0.962(2) & 0.979(2) & 0.9930(6)\\
        \hline
    \end{tabular}
    \caption{\footnotesize The final $F_1$ scores achieved by the trained neural networks in Figure \ref{fig:layer-comparison}, where the number of neural network layers is varied.}
    \label{tab:layer-comparison}
\end{table}

Another possible source of underperformance for the 3HDM analysis might be suboptimal training hyperparameters. Given that the Adam algorithm is usually quite robust to changing learning rates, we instead can focus on the characteristic length scale of our weight priors, $l$ (as defined and discussed in Appendix \ref{appendix:Learner}), and the epoch patience (the number of training epochs without improvement on the loss that the Adam optimizer tolerates before declaring the neural network weights converged). In these cases, decreasing $l$ (or in other words, increasing the variance of the Gaussian prior $\mathcal{N}(0, l^{-2})$ on the weights) or increasing the epoch patience should lead to a better fit of the neural network to the training data, since either less prior knowledge is assumed on the weight distributions (as $l \rightarrow 0$, the priors become those of total ignorance) or a greater number of optimizer steps is permitted. If the neural network with our original hyperparameters is failing to adequately learn a decision rule that's well-represented in our training data, decreasing $l$ and increasing the epoch patience would presumably improve the model's performance. To test this, we train 3HDM classifiers with a dramatically reduced $l$ value ($l = 0.01$), as well as classifiers with dramatically increased epoch patience (500, instead of 100). The comparative $F_1$ scores associated with tuning these hyperparameters, along with the performance of the default setting, are depicted in Figure \ref{fig:3HDM-hyperparams}, with their final performances summarized in Table \ref{tab:3HDM-hyperparams}-- we can clearly see that there is no statistically significant difference in performance when these hyperparameters are altered. In turn, this result supports the thesis that inadequacy of the model architecture is not the principal source of the 3HDM classifier's underperformance.

\begin{table}[]
    \centering
    \begin{tabular}{| c | c | c | c |}
        \hline
          & BALD & Max Entropy & Variation Ratios\\
        \hline
        Default & 0.9658(7) & 0.968(2) & 0.967(2)\\
        Epoch Patience = 500 & 0.968(2) & 0.9670(9) & 0.9687(8)\\
        $l = 0.01$ & 0.966(2) & 0.968(1) & 0.967(2)\\
        \hline
    \end{tabular}
    \caption{\footnotesize The final $F_1$ scores achieved by the trained neural networks in Figure \ref{fig:3HDM-hyperparams}, where the 3HDM bounded-from-below conditions are learned with differing training hyperparameters, as described in Figure \ref{fig:3HDM-hyperparams} and the text.}
    \label{tab:3HDM-hyperparams}
\end{table}

\begin{figure}
    \includegraphics[width=6in]{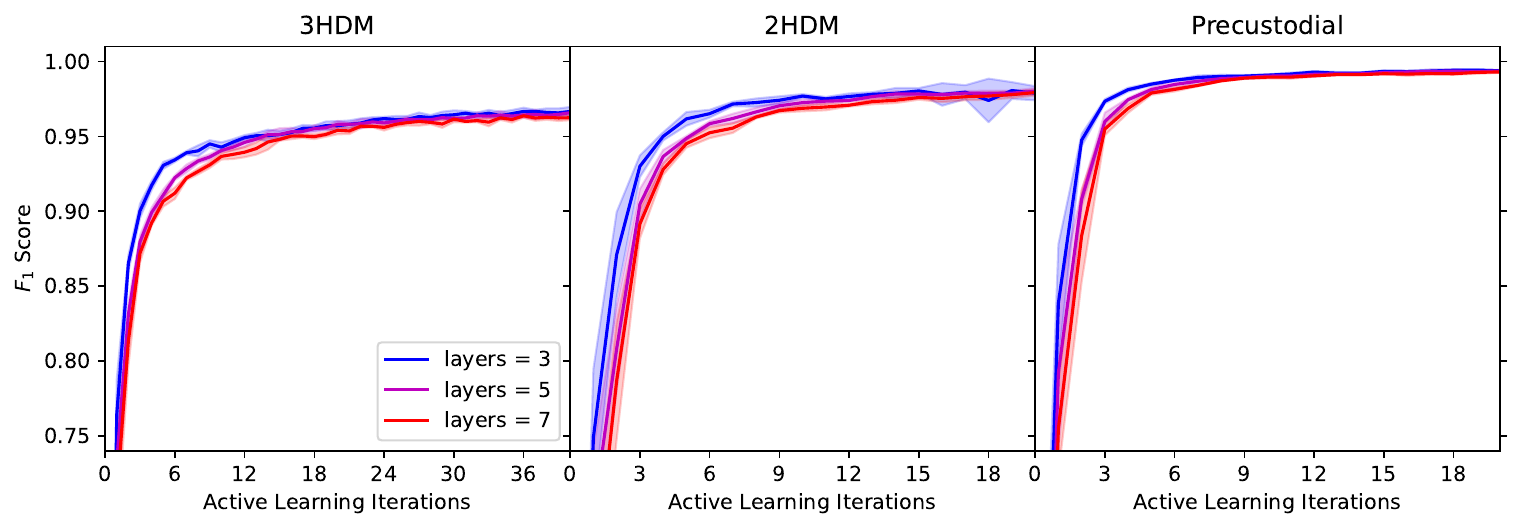}
    \caption{The $F_1$ score achieved over the course of active learning with the BALD query strategy for the 3HDM (left), 2HDM (middle), and precustodial (right) scalar potentials, for differing numbers of hidden layers (where each hidden layer is constructed with 128 neurons)-- 3 layers (blue), 5 layers (magenta), and 7 layers (red). Lines represent the mean of 5 experiments while transparently shaded areas denote the standard deviation. Final performances achieved by these trained classifiers are summarized in Table \ref{tab:layer-comparison}.}
    \label{fig:layer-comparison}
\end{figure}

\begin{figure}
    \includegraphics[width=6in]{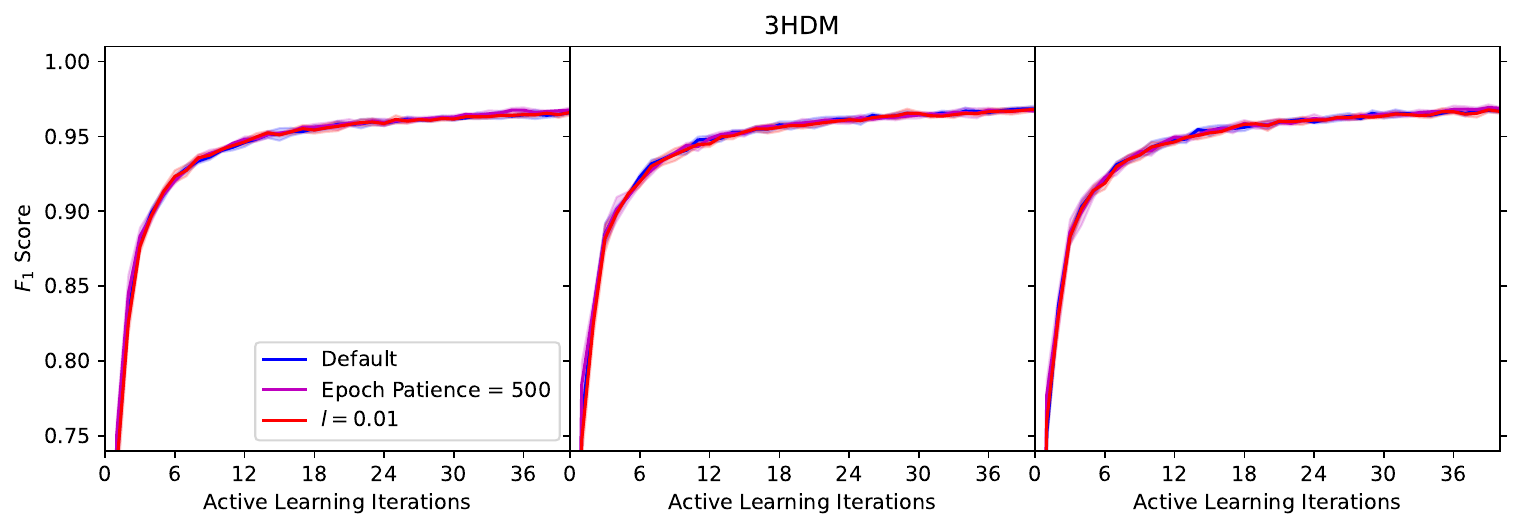}
    \caption{The $F_1$ score achieved over the course of active learning for the 3HDM potential with different hyperparameter values for the prior length scale $l$ and the patience of the Adam optimizer's early stopping condition (epoch patience), for BALD (left), maximum entropy (middle), and variation ratios (right) query strategies. Lines represent the mean of 5 experiments while transparently shaded areas denote the standard deviation. Final performances achieved by these trained classifiers are summarized in Table \ref{tab:3HDM-hyperparams}.}
    \label{fig:3HDM-hyperparams}
\end{figure}

Given our results, perhaps the most likely explanation for the degraded performance of the classifier on the 3HDM is that the active learning algorithm fails to explore the 15-dimensional $\vec{\lambda}$ space with the same efficiency that it explores the 2HDM and precustodial potentials' 10-dimensional $\vec{\lambda}$ space. We can find further support for this position by considering the accuracy that the classifier achieves on their active learning training sets. In Table \ref{tab:3HDM-hyperparams-accuracy}, we compute the binary accuracy of the 3HDM model on its training data after 40 epochs, determining the mean and error as usual from 5 independent trials.

\begin{table}[]
    \centering
    \begin{tabular}{| c | c |}
        \hline
        Hyperparameter Choice &  Training Set Accuracy\\
        \hline
        Default & 0.975(2) \\
        \hline
        Epoch Patience = 500 & 0.989(3)\\
        \hline
        $l = 0.01$ & 0.974(2)\\
        \hline
    \end{tabular}
    \caption{\footnotesize The binary accuracy on training data (which should, by its nature, be ambiguous) for fully-trained 3HDM classifiers with the specified hyperparameter choices. Mean and error are computed by conducting 5 independent trials for each parameter choice.}
    \label{tab:3HDM-hyperparams-accuracy}
\end{table}

We see that, while adjusting $l$ has no discernible effect on the neural network's accuracy on its training set (and therefore, likely a minimal effect on the neural network's performance in general), the trials with increased epoch patience unsurprisingly improve the models' accuracies on their training sets. However, this improved performance on the training set doesn't translate to improved performance on the validation set. We can therefore reason that the validation set contains points which are not well-represented by the training data-- the 3HDM parameter space is not being entirely explored.
Because the geometry of each potential's bounded-from-below region will, of course, differ substantially, it is furthermore feasible that some difference in the efficiency of our active learning procedure in exploring the 2HDM and the precustodial potentials may account for the performance discrepancy between these potentials as well.

Finally, we might suspect that the performance discrepancies between potentials that we observe could be improved by simply increasing the number of active learning rounds. To explore this, we see in Figures \ref{fig:2HDM-F}-\ref{fig:Precustodial-F} that recall (which we remind the reader is an estimate of the probability that a truth level bounded-from-below point is correctly labelled by the classifier) for the 2HDM and 3HDM potentials has a visible positive slope even at the end of active learning, indicating that the neural network is continuing to discover new regions of bou233nded-from-below parameter space. However, precision (which estimates the probability that a point that the classifier labels as bounded-from-below is actually bounded-from-below at truth level) demonstrates a far less pronounced monotonic improvement for both of these potentials, so it remains unclear whether and to what degree any performance gap might be closed simply by arbitrarily extending training.\footnote{As the size of training data increases beyond what can be contained in GPU memory, it may be necessary to split the training data into batches, which we have found can slightly degrade performance-- see the discussion at the end of Appendix \ref{appendix:Learner}.}

Although Figures \ref{fig:2HDM-F}-\ref{fig:Precustodial-F} and Tables \ref{tab:fig-results-1} and \ref{tab:fig-results-2} have offered us some minor insight into the utility of our uncertainty metrics, we can now move on to explore these characteristics more rigorously. In Figure \ref{fig:uncertainty}, we depict the $F_1$ scores on the validation sets for the classifiers trained for the 2HDM, 3HDM, and precustodial potentials with our default parameters, as points with uncertainty metrics above different quantiles (determined separately for bounded-from-below and not bounded-from-below points) are excluded from the validation set.
We can see that all three uncertainty metrics are useful predictors of the classifier's accuracy on a given point: As the included uncertainty quantiles decrease, the $F_1$ scores for each potential rapidly improve-- for the 3HDM and 2HDM, we see that this improvement is most rapid when mutual information, that is, our estimate of purely epistemic uncertainty, is used as the uncertainty metric, while performance for all three uncertainty metrics is comparable for the precustodial potential, where we have also found the classifier achieves the best performance in general.

\begin{figure}
    \centering
    \includegraphics[width=6in]{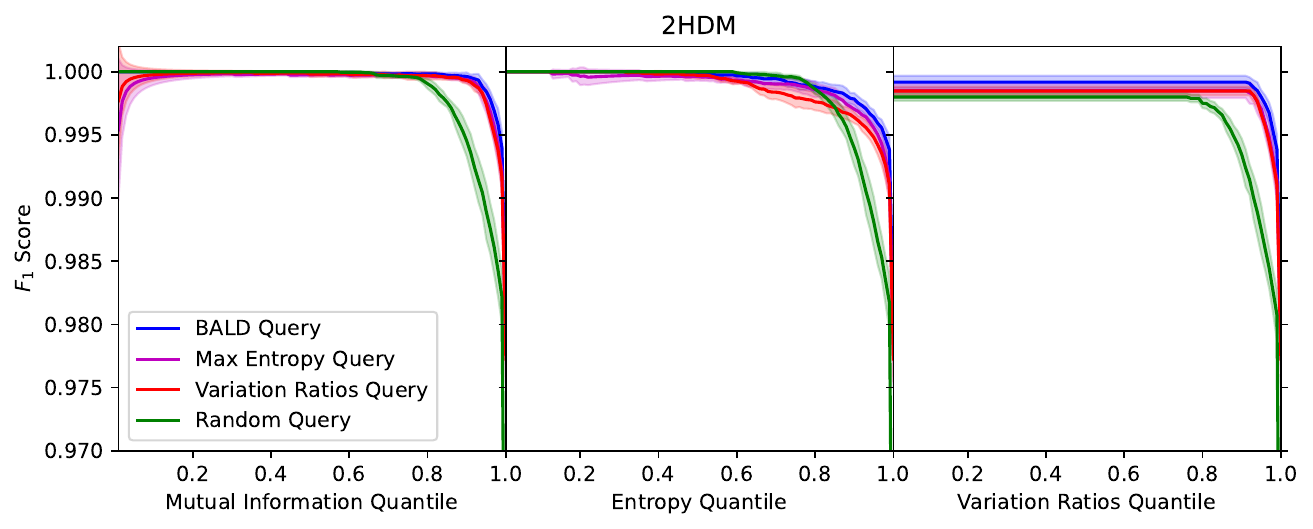}
    \includegraphics[width=6in]{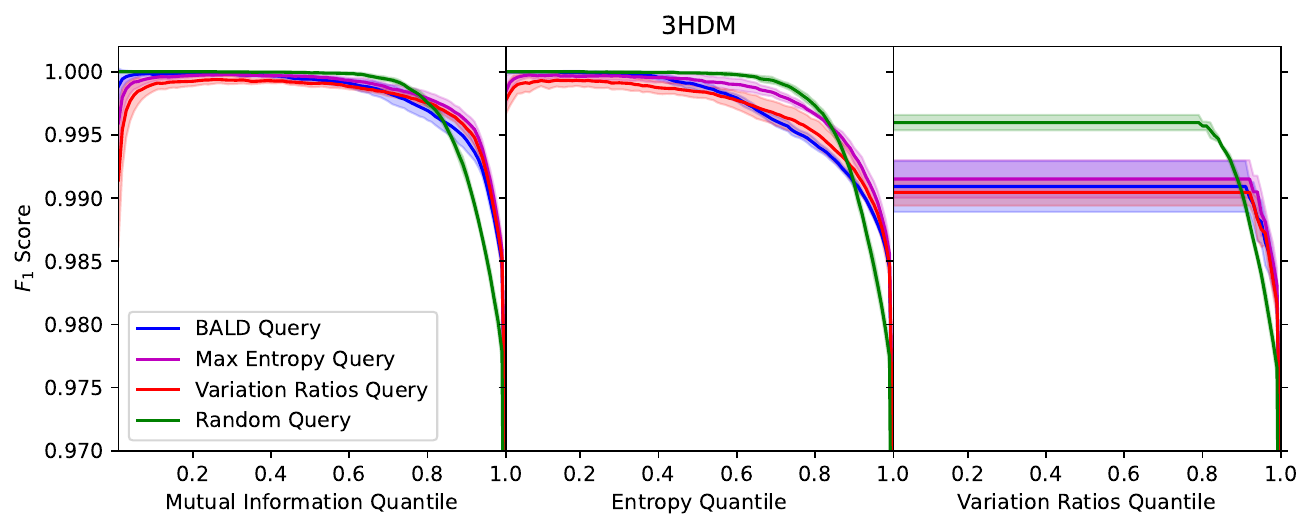}
    \includegraphics[width=6in]{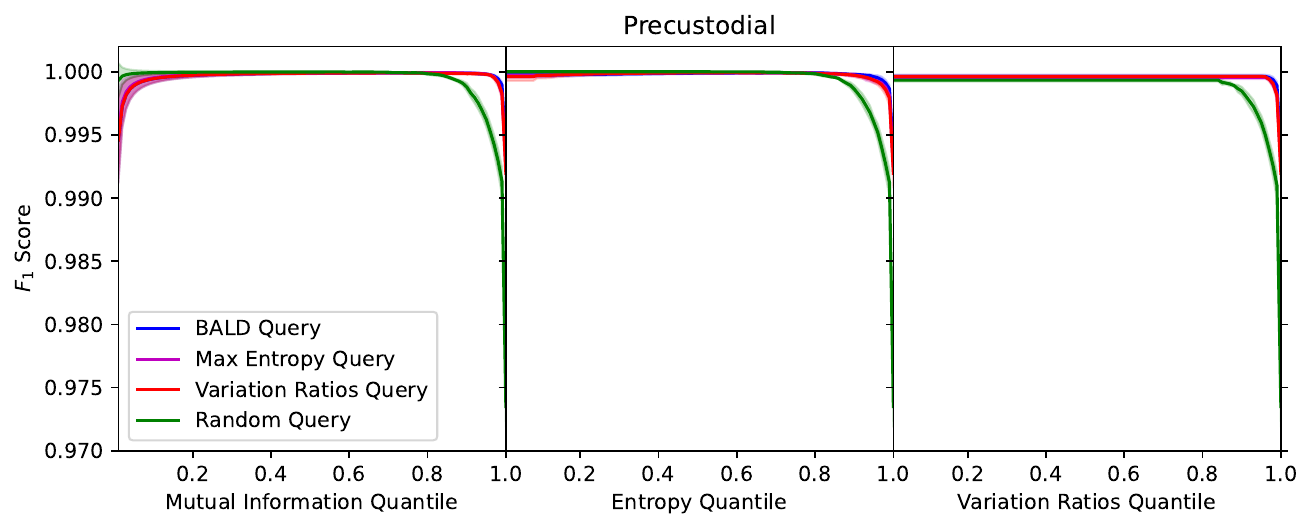}
    \caption{The $F_1$ score achieved for the 2HDM (top), 3HDM (middle), and precustodial (bottom) scalar potentials, with points which have an uncertainty score (mutual information, entropy, or variation ratios) greater than different quantiles for points with their predicted class excluded. Lines represent the mean of 5 experiments while transparently shaded areas denote the standard deviation, while the query strategy used is BALD (blue), maximum entropy (magenta), variation ratios (red), and random (green).}
    \label{fig:uncertainty}
\end{figure}

 The dropoff of the performance of some of the classifiers as the included uncertainty quantiles are very small ($\lesssim 10 \%$) suggests that these uncertainty quantification metrics are not perfect-- there exist a small number of points about which the model can be quite certain, but incorrect. However, these incorrect-but-certain points represent an extremely small fraction of the total points considered, given that the classifiers in question still exhibit $F_1$ scores of $>0.995$ for included quantiles of about $0.05$. Given that incorrect-but-certain points appear to be most present when mutual information is the uncertainty metric, this suggests that these points are likely points of low epistemic uncertainty (that is, adequate training data in the region of the points), but high aleatoric uncertainty (that is, the points themselves are very close to the decision boundary). Since mutual information is insensitive to aleatoric uncertainty, the failure of this metric to account adequately for these points is unsurprising, however we note that it is likely that such points, being near the neural network's decision boundary, will likely produce similar phenomenology to points which are correctly classified. The variation ratios uncertainty metric appears, in general, to be the least useful uncertainty metric-- producing consistent and significant underperformance as a discriminator for all three potential functions. This is unsurprising given the fact that the variation ratio's expressivity as an uncertainty metric is significantly limited by the number of predictive evaluations that we perform: Using 1000 forward passes as we do, we can expect that any point which has a ``truth-level'' (that is, the limit as the number of evaluations goes to infinity) variation ratio of $\lesssim 0.001$ will be imprecisely determined due to insufficient statistics. Mutual information and entropy, which aren't based on predictions crossing a specific threshold, will be significantly less sensitive to the effect of a finite number of predictive evaluations.

A final observation we can make on the discriminating power of the various uncertainty metrics lies in the different performance of the classifiers trained with different query strategies. Notably, the random query strategy appears to be robust against the dropoff in $F_1$ score for very small included uncertainty quantiles, even in cases where other query strategies exhibit this behavior. This suggests, then, that a classifier trained with the random acquisition function possesses slightly higher-quality uncertainty estimates, in that there are significantly fewer points for which the classifier is highly certain but incorrect than there are for classifiers with our various active learning query strategies. We might expect that this would be a consequence of statistical bias in active learning \cite{farquhar2021statistical}, namely that the training set for active learning, not being an independent and identically distributed (i.i.d.) sample from the same distribution as the validation set, might not predict uncertainties that are optimally calibrated for it. Of course, the training set points used during training with the random acquisition function aren't from the same distribution as the validation set uniformly sampled from the unit hypersphere surface (or, indeed, independently distributed between active learning iterations), but by not being disproportionately selected to have high uncertainty, it is perhaps not unreasonable to suggest that the randomly-queried training set is a closer representation of the validation set's distribution than those which are generated through the active learning strategies. Of course, if this is the cause of the discrepancy, it is both unclear whether this advantage for the random query strategy will hold for \emph{other} distributions, such as those which might emerge in a parameter space scan, and in any case the discrepancy in performance only becomes statistically significant in most cases once an impractically small $(\lesssim 0.1)$ included uncertainty quantile is used, in which case the neural network classifier is not especially useful. A detailed exploration of the effect of training set bias in uncertainty quantification for active learning for this application would require the development of a method of removing this bias from the loss estimate for our active learning application. Such an estimator was developed in \cite{farquhar2021statistical}, for example, in the case of active learning with a finite and not-replenished pool of unlabelled candidate training data $L$, but identifying an analogous estimator for our active learning implementation is beyond the scope of this work.

\subsection{Experiments: Classifier Performance on Semi-Realistic Test Sets}\label{sec:exp-slices}

In the previous Section, we have evaluated our classifiers' performances on validation sets of points which uniformly sample from the $\vec{\lambda}$ space hyperspheres of the 2HDM, 3HDM, and precustodial potential functions. While these data sets are useful to get a sense of a classifier's overall coverage of the complete bounded-from-below space of the potential, these sets themselves bear little resemblance to the distributions of points which are more typical of a practical parameter space scan. As such, we shall explore our models' performance on test distributions of a form that more commonly appears in parameter space scans: 2-dimensional slices in $\vec{\lambda}$ space. For each potential, we have selected 4 pairs of quartic coefficients to scan over while we hold the remainder fixed. Then, for each pair, we randomly generate 500 sets of the remaining parameters, and then for each of these 500 sets, we randomly generate 2000 pairs of the two scanned parameters. In all cases, we randomly sample from uniform distributions over each quartic coefficient between -5 and 5, except where a negative value of a quartic parameter violates a necessary condition for boundedness-from-below, in which case we restrict the value to be positive (this helps maximize the fraction of our slices which are ``interesting''-- that is, contain a decision boundary). In the case of the precustodial model, two coefficients $\lambda_2$ and $\lambda_3$ in Eq.(\ref{eq:V-PC}), must satisfy $\lambda_2 + \lambda_3 > 0$ and $\lambda_2 + \lambda_3 / 2 > 0$ in order to avoid bounded-from-below conditions for a potential with a single complex $SU(2)_L$ triplet, so in cases where $\lambda_2$ and $\lambda_3$ aren't scanned over, we uniformly sample these two combinations of the coefficients along intervals between 0 and 5 instead of sampling $\lambda_2$ and $\lambda_3$ directly.

We can use these ``slice ensembles'' to get a sense of the level of performance we can expect from our trained classifiers on generic 2-dimensional slices (and more broadly, typical lower-dimensional cross sections in parameter space that appear often in phenomenological studies), and, crucially, how accuracy of the model on these slices might be estimated from the predictive uncertainty metrics described in Section \ref{sec:bayesian-nn}. To do this, we label each point in our slices using the corresponding oracle. Then, taking our models trained using the BALD query strategy as a baseline, we check the predictive accuracy for the model along each 2-dimensional parameter space slice, along with the uncertainty (as determined by mutual information, which we found in Section \ref{sec:exp-performance} to provide the best discriminating power between incorrectly and correctly classified points in that setting) associated with each prediction. By scanning a large number of these slices, we can also get a qualitative and quantitative understanding of the ``worst-case'' performance of the classifiers on these types of samples (in effect, executing a primitive search for adversarial inputs in the language of machine learning), as well as what constitutes more typical performance and how one can differentiate between the two. In Table \ref{tab:slice-data}, we list data related to the binary accuracy of the classifiers on these ensembles of slices, including the average binary accuracy on a slice, the minimum binary accuracy, and the median performance among slices for which $\geq 5 \%$ of points are predicted as bounded-from-below (what might be considered ``interesting'' parameter space-- computing the median accuracy across all slices will tend to uninformatively yield a value of 1.0, due to a significant number of slices which are entirely correctly classified as not bounded-from-below).

\begin{table}[]
    \centering
    \begin{tabular}{| c | c | c | c | c |}
        \hline
        & Scanned Variables &  Min. Accuracy & Med. Accuracy & Accuracy-Uncertainty Correlation\\
        \hline
        \textbf{2HDM} & $\lambda_4 - Re \, \lambda_5$ & 0.978(2) & 0.9951(8) & -0.62(10)\\
        & $\lambda_2 - \lambda_3$ & 0.979(5) & 0.9942(9) & -0.52(14)\\
        & $\lambda_3 - \lambda_4$ & 0.977(7) & 0.9964(3) & -0.45(20)\\
        & $Re \, \lambda_6 - Re \, \lambda_7$ & 0.977(4) & 0.9946(6) & -0.80(5)\\
        \hline
        \textbf{3HDM} & $\lambda_4-\lambda_7$ & 0.65(10) & 0.986(1) & -0.70(7)\\
        & $Re \, \lambda_{10} - Re \, \lambda_{11}$ & 0.16(11) & 0.982(2) & -0.82(7)\\
        & $\lambda_1 - \lambda_7$ & 0.40(25) & 0.9856(6) & -0.82(10)\\
        & $\lambda_1 - \lambda_5$ & 0.59(8) & 0.9849(6) & -0.75(6)\\
        \hline
        \textbf{Precustodial} & $\lambda_5-\lambda_6$ & 0.85(8) & 0.9975(2) & -0.82(9) \\
        & $\lambda_8-\lambda_9$ & 0.93(3) & 0.9986(2) & -0.87(5)\\
        & $\lambda_2-\lambda_5$ & 0.97(2) & 0.9989(2) & -0.75(7)\\
        & $\lambda_4-\lambda_7$ & 0.96(1) & 0.9980(3) & -0.74(3)\\
        \hline
    \end{tabular}
    \caption{\footnotesize The minimum binary accuracy for each slice ensemble, the median accuracy for slices in the ensemble which fall in the ``interesting region'' ($\geq 5 \%$ of points are predicted to be bounded-from-below), and the Pearson correlation coefficient between binary accuracy and the average mutual information of the points in elements of the slice ensemble. The averages and errors for these quantities are computed from the results for five independently trained classifiers. The quartic coefficients are defined as in Eq.(\ref{eq:V-2HDM}) for the 2HDM, Eq.(\ref{eq:V-3HDM}) for the 3HDM, and Eq.(\ref{eq:V-PC}) for the precustodial model.}
    \label{tab:slice-data}
\end{table}

The results outlined in Table \ref{tab:slice-data} give us a feeling for the performance of the classifiers on the slice ensembles. We can see that, with the possible exception of the universally $\geq 97 \%$ accuracy achieved by the 2HDM classifiers, we have found specific 2-D slices for which the classifier has a significant error rate. In the case of the 3HDM potential, we even find slices with accuracy values which are substantially worse than the $\sim 50\%$ expected performance of an untrained classifier. However, it also appears that a strong negative correlation universally exists between the mean mutual information in the points of a given slice and the classifier's binary accuracy for that slice, as we might expect: Just as with the validation samples in Section \ref{sec:exp-performance}, this uncertainty metric provides a reliable estimate of the accuracy of a given prediction. Furthermore, this correlation appears to increase as less accurately classified slices are present in the ensemble: While the uniformly accurate 2HDM classifiers tend to have the most modest negative correlations between accuracy and mean uncertainty, the sometimes seriously inaccurate 3HDM classifiers uniformly enjoy correlation coefficients of less than or equal to -0.7.

While Table \ref{tab:slice-data} presents some quantitative data on the performance of our methodology on different slice ensembles, a visual approach can give us a better understanding of the utility of our method-- specifically, it can better illustrate how uncertainty quantification can allow us to robustly exclude (or depending on the needs of the model builder, re-examine with a more robust, but computationally expensive, methodology for determining boundedness-from-below) the substantially inaccurate slices while correctly identifying the slices for which we can trust the classifier's results. In Figures \ref{fig:2HDM-slice-image}, \ref{fig:3HDM-slice-image}, and \ref{fig:Precustodial-slice-image}, we depict some illustrative examples of classifier slices for the 2HDM, 3HDM, and precustodial potentials, respectively. To use uncertainty estimates to meaningfully gauge the reliability of individual model predictions, we will of course have to have a notion of what a ``large'' uncertainty means for each model. To establish such a notion, we compute our highest-quality uncertainty estimator, mutual information, for a ``calibration set'' of $10^6$ input $\vec{\lambda}$'s sampled uniformly from the surface of the unit hypersphere (for our purposes here, we use our existing labelled validation sets, but an unlabelled and randomly generated set of points will work just as well for use cases in which labelling such a large set of points with an oracle is impractical; the BFBrain package possesses methods to produce such a calibration set automatically). Noting that points which are predicted as bounded-from-below tend to have larger mutual information than points of the other class, we then compute the 0.95 and 0.99 mutual information quantiles for the set of calibration set points which the model predicts to be bounded-from-below. These quantiles allow us to differentiate highly uncertain inputs from others, without reference to the distribution of our input slices or the precise numerical scale of the mutual information measurements, which is nontrivially dependent on neural network architecture, training data, and hyperparameters.

To examine both the circumstances under which the classifier performs poorly, as well as get a sense for each classifier's more typical performance, Figures \ref{fig:2HDM-slice-image}-\ref{fig:Precustodial-slice-image} depict both the member of the given slice ensemble which exhibits the lowest binary accuracy (as discussed above, our ``adversarial'' examples), as well as one which exhibits the median binary accuracy for the ensemble (just as in Table \ref{tab:slice-data}, we have restricted each slice ensemble to only those points with at least $5 \%$ of the points being predicted to be bounded-from-below). Unlike the figures thus far displayed in this work, these classifications are the product of individual trained classifiers, and not averaged results among 5 independent identical experiments-- this is because individual predictions of trained classifiers can't be meaningfully combined without likely enhancing the networks' collective predictive power via ensembling \cite{barber1998ensembling}, a use case which we do not explore in this work. For completeness, plots for the slice ensembles not depicted in the main text are included in Figures \ref{fig:2HDM-slice-image-appendix} - \ref{fig:Precustodial-slice-image-appendix} in Appendix \ref{appendix:slice-ensemble}, along with Tables \ref{tab:2HDM-slice-params-worst}-\ref{tab:Precustodial-slice-params-med}, which give the specific quartic coefficient values used to produce all plots depicted here and in that appendix.

\begin{figure}
    \centerline{\includegraphics[width=3.2in]{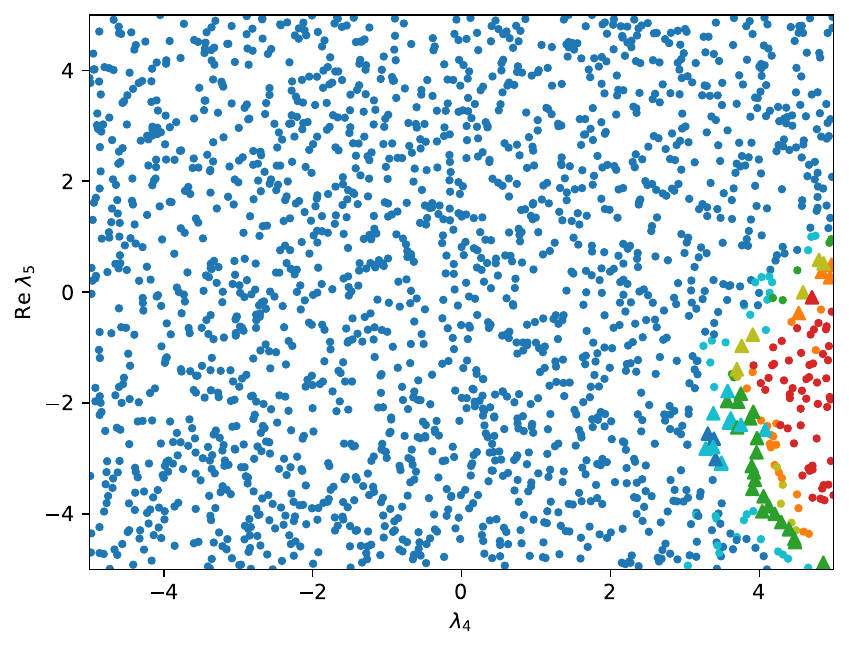}
    \hspace{-0.25cm}
    \includegraphics[width=3.2in]{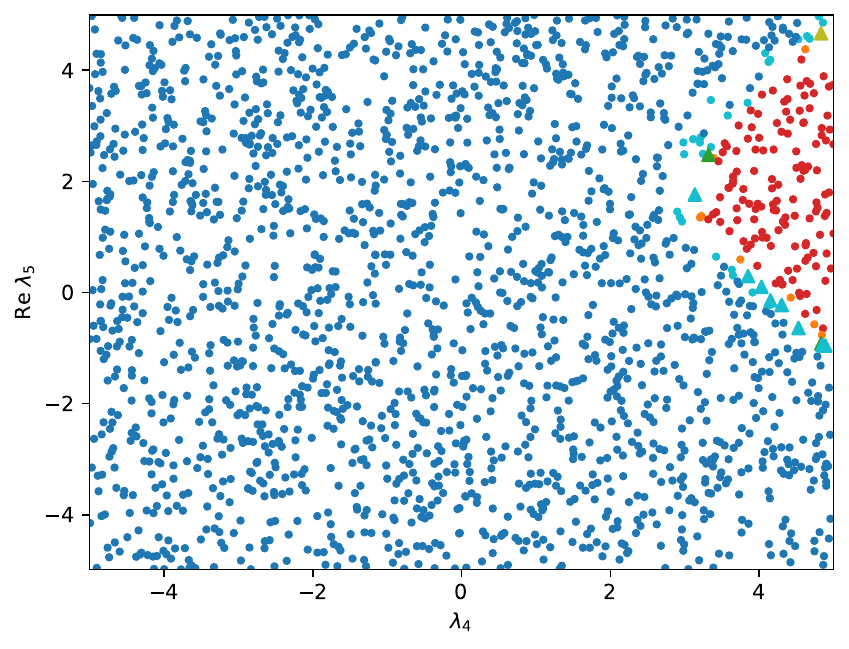}}
    \caption{Assuming a BALD query strategy and default hyperparameters, scatter plots drawn from the $\lambda_4- \textrm{Re} \, \lambda_7$ slice ensembles for a classifier for the 2HDM potential given in Eq.(\ref{eq:V-2HDM}). Dots represent correctly classified points, while triangles represent incorrectly classified points. Classification label and mutual information are indicated by point color: Red, orange, and olive points are predicted to be bounded-from-below, while blue, cyan, and green points are predicted to not be. Different color classifications within predicted labels indicate that the point has mutual information greater than the 0.99 quantile (olive, green) or 0.95 quantile (orange, cyan) of the mutual information in the calibration set (see text). On the left, the slice in the ensemble with the lowest binary accuracy is depicted. On the right, a more typical example is depicted, with binary accuracy equal to the median accuracy, computed as described in the text.}
    \label{fig:2HDM-slice-image}
\end{figure}

\begin{figure}
    \centerline{\includegraphics[width=3.2in]{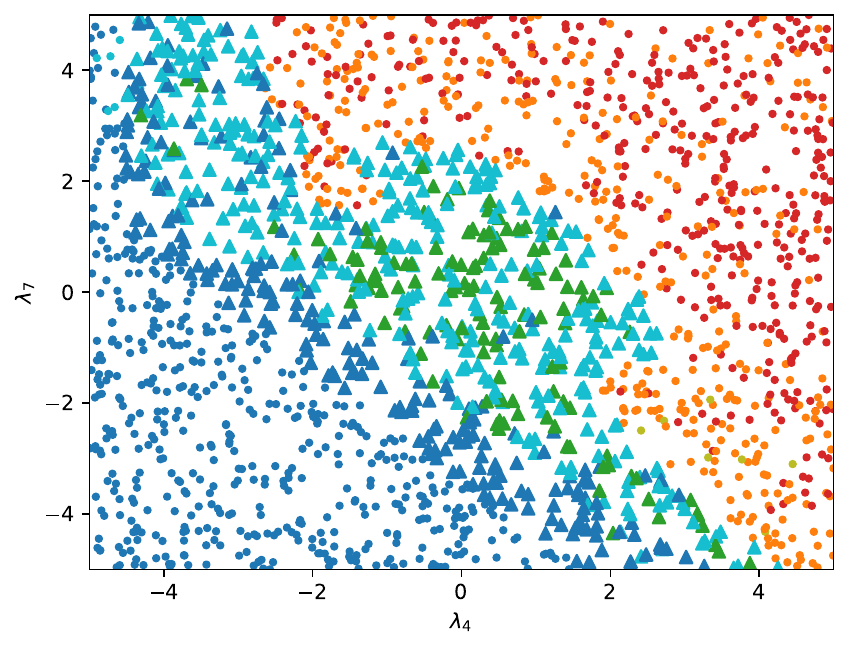}
    \hspace{-0.25cm}
    \includegraphics[width=3.2in]{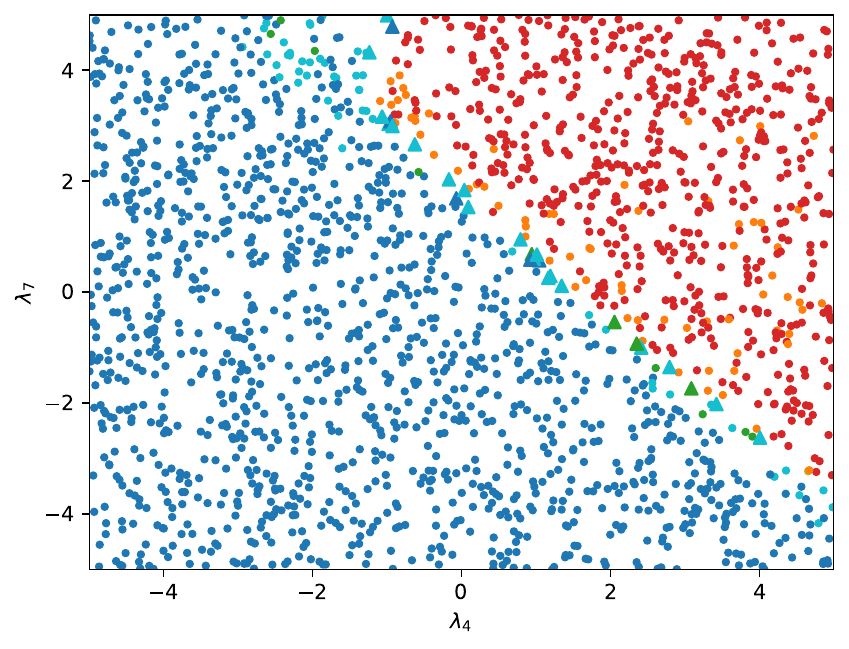}}
    \caption{As Figure \ref{fig:2HDM-slice-image}, but for the $\lambda_4 -\lambda_7$ slice ensemble of the 3HDM potential, with quartic coefficients defined in Eq.(\ref{eq:V-3HDM}).}
    \label{fig:3HDM-slice-image}
\end{figure}

\begin{figure}
    \centerline{\includegraphics[width=3.2in]{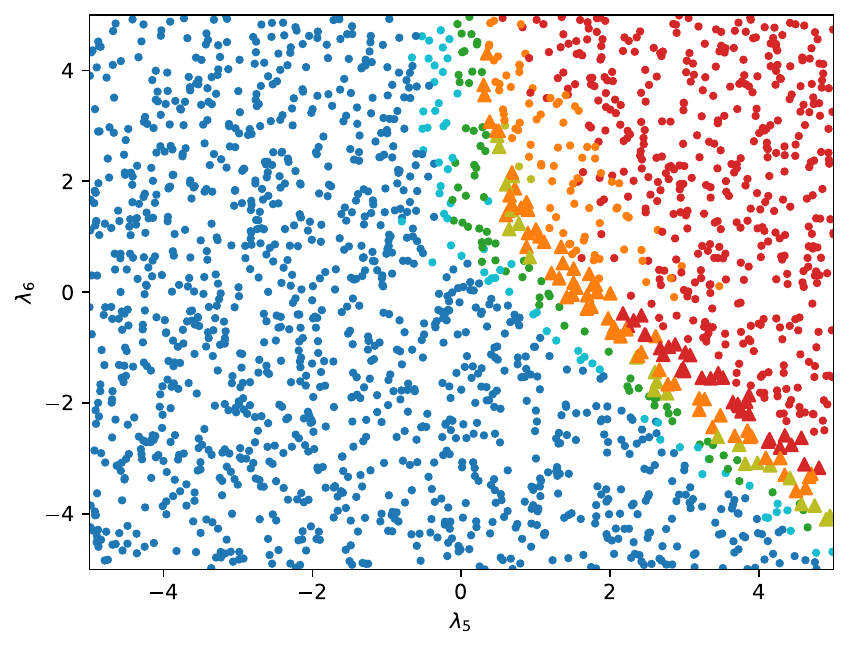}
    \hspace{-0.25cm}
    \includegraphics[width=3.2in]{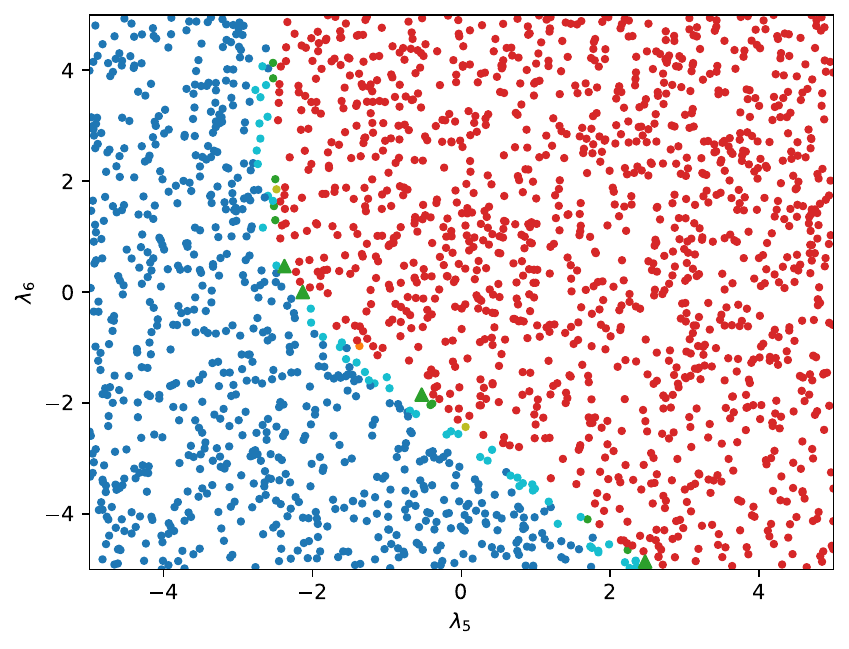}}
    \caption{As Figure \ref{fig:2HDM-slice-image}, but for the $\lambda_5 -\lambda_6$ slice ensemble of the precustodial potential, with quartic coefficients defined in Eq.(\ref{eq:V-PC}).}
    \label{fig:Precustodial-slice-image}
\end{figure}

Figures \ref{fig:2HDM-slice-image}-\ref{fig:Precustodial-slice-image} help depict the practical usefulness of classifiers trained by our method, allowing us to see both the typical (extremely precise) performance of the model on various 2-D slices, as well as what the output may look like when the classifier fails. For the plots which depict median classifier performance along the slice ensembles, we see uniformly high accuracy for all three of the 2HDM, 3HDM, and precustodial classifier. Misclassified points are extremely rare, invariably close to the decision boundary, and tend to have unusually high mutual information scores-- indicating higher-than-average epistemic uncertainty. Furthermore, in agreement with the correlations found in Table \ref{tab:slice-data}, we see that our median-accuracy slices have a significant number of points that do \emph{not} demonstrate high mutual information, so that we might identify 2-dimensional slices on which the classifier is likely to be highly performant without having to label our points with an oracle.

In the case of our classifers' worst slices, we see that at least for the 3HDM and the precustodial potentials' classifiers, there exist 2-dimensional parameter space slices for which there is significantly degraded performance. More reassuringly, we see that these slices also have a number of points with exceedingly high mutual information, in particular within the misclassified regions. These worst-case scenario plots therefore suggest that while it may be possible to find regions of parameter space over which the classifiers have unacceptably poor performance, such regions are in general recognizable as such.\footnote{In the case of 2-dimensional parameter space scans along the quartic coefficients, we can also identify likely unreliable outputs by observing non-convexities in the bounded-from-below region, which must be incorrect. However, doing this for 3-dimensional scans is considerably more difficult, and for 4-dimensional scans cannot be done visually.} We should also note that the worst-case scenario performances are in no way typical of the majority of samples in a parameter space-- unless the user is scanning over hundreds of two-dimensional slices and specifically selecting points which have a high density of high-uncertainty points, the median-performance plots are a far more typical.


\subsection{Experiments: Effects of Label Noise}

We conclude our experimentation by considering the effect of a flawed oracle on the quality of the classifier-- specifically in the case that the principal parameter governing the oracle's accuracy in our experiments, $n_{\textrm{iter}}$, is set to significantly suboptimal values. From Section \ref{sec:oracle}, we recall that this will lead to a significant number of points incorrectly being labelled as bounded-from-below, with smaller $n_{\textrm{iter}}$ values leading to a greater number of misclassifications. To gauge the severity of this effect, we train classifiers with oracles that use different small values of $n_{\textrm{iter}}$ (specifically 5, 10, and 20), and track their performance on validation data labelled using the $n_{\textrm{iter}}$ parameters of Table \ref{tab:default-hyperparams}, which satisfy our robustness criteria of Section \ref{sec:oracle}. In Figure \ref{fig:noise-comparison}, we depict the $F_1$ scores achieved over the course of training for each of our considered scalar potentials with the different $n_{\textrm{iter}}$ values, as well as the performance of each oracle with suboptimal $n_{\textrm{iter}}$ on the validation data set. We have also listed the final $F_1$ scores achieved by each of these noisy classifiers, in addition to the corresponding results for the equivalent trials with default $n_\textrm{iter}$ values first depicted in Figures \ref{fig:2HDM-F}-\ref{fig:Precustodial-F}, in Table \ref{tab:noise-comparison}.

\begin{figure}
    \centering
    \includegraphics[width=6in]{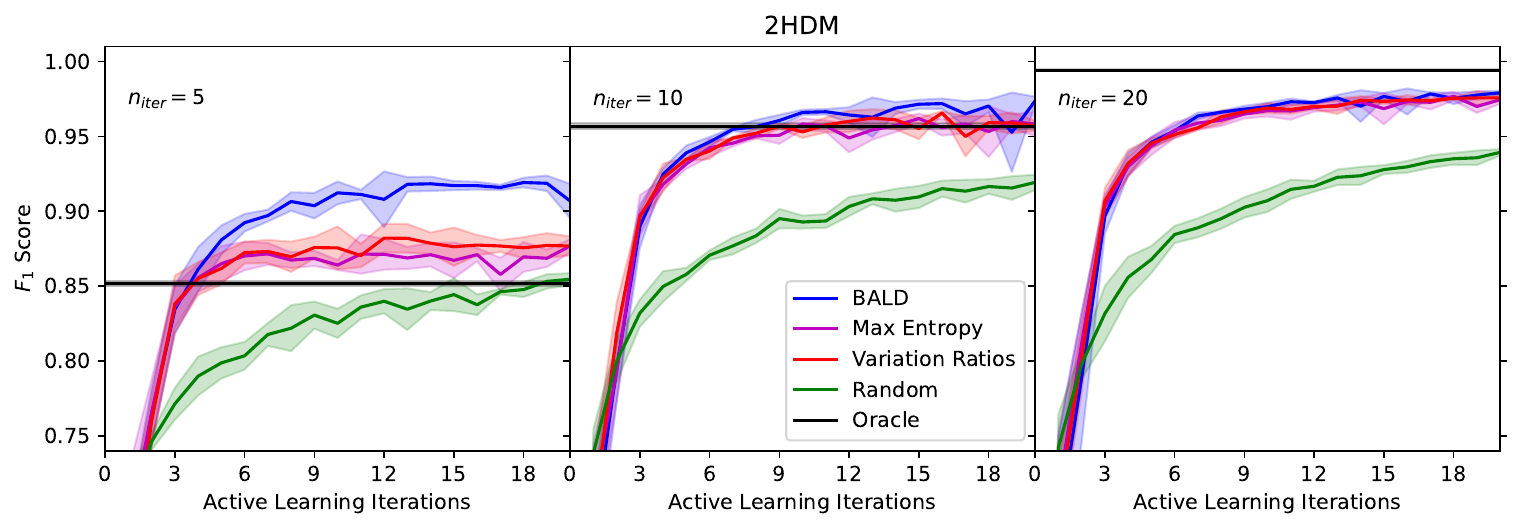}
    \includegraphics[width=6in]{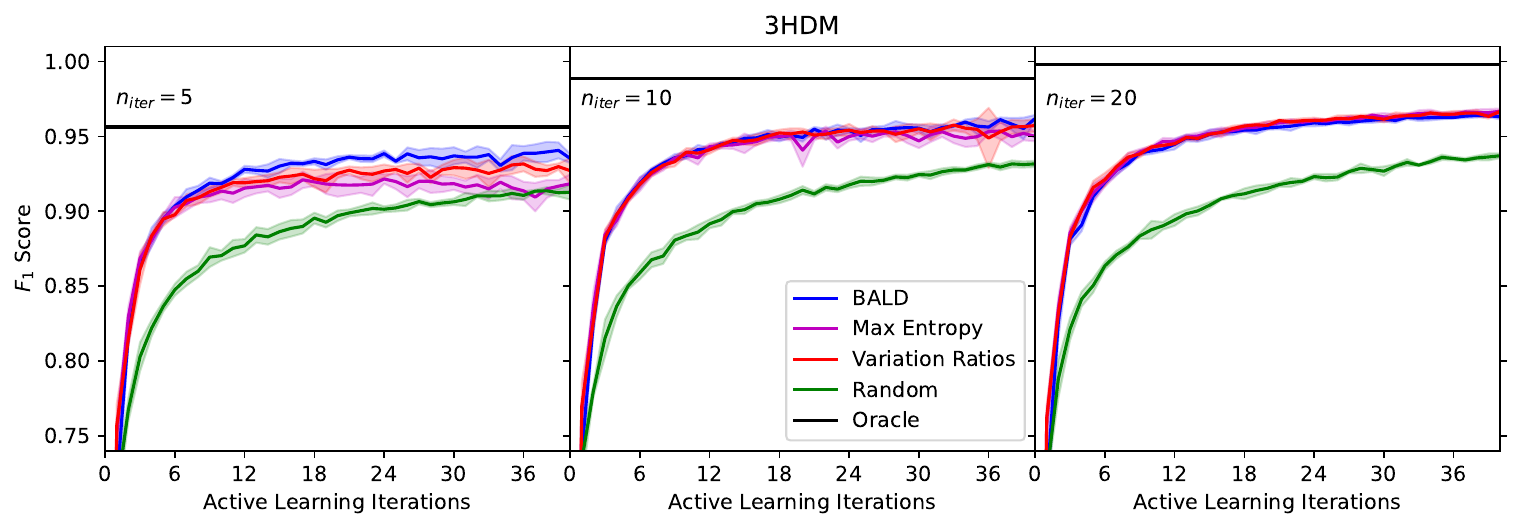}
    \includegraphics[width=6in]{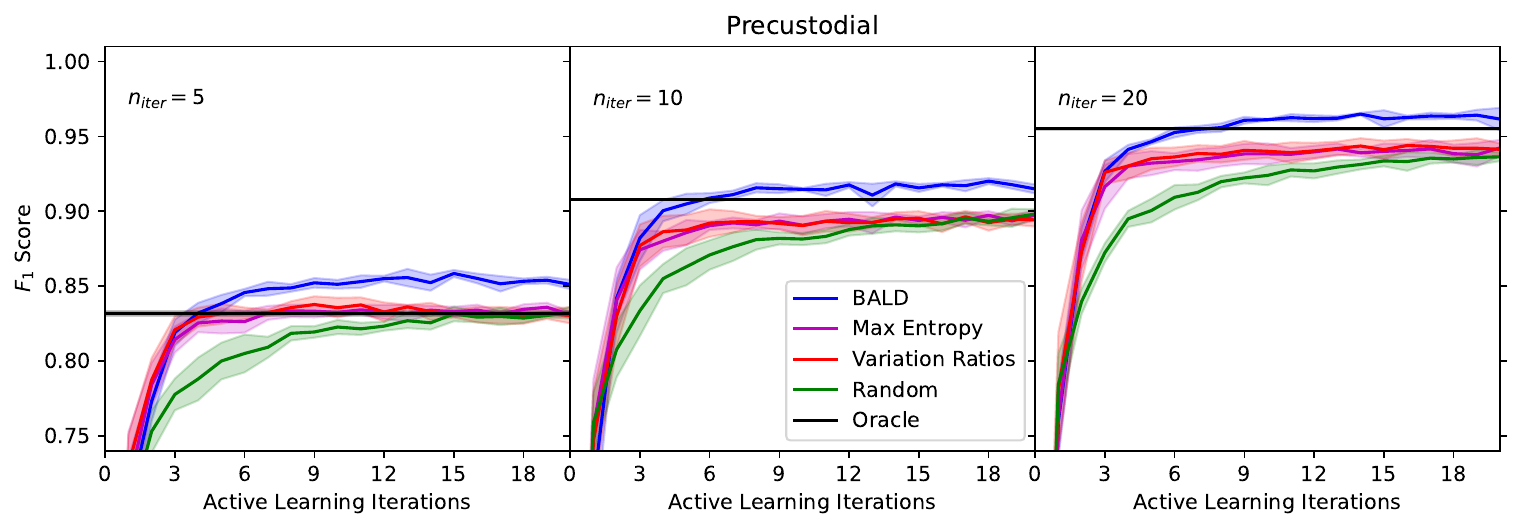}
    \caption{The $F_1$ score achieved for the 2HDM (top), 3HDM (middle), and precustodial (bottom) scalar potentials over the course of active learning, for oracles with $n_{\textrm{iter}} = 5$ (left), 10 (middle), and 20 (right). Lines represent the mean of 5 experiments while transparently shaded areas denote the standard deviation, while the query strategy used is BALD (blue), maximum entropy (magenta), variation ratios (red), and random (green). Although the training data is labelled using an oracle with the indicated $n_{\textrm{iter}}$ values, the validation data is labelled using the robust $n_{\textrm{iter}}$ parameters we discussed in Section \ref{sec:oracle} and given in Table \ref{tab:default-hyperparams}. The black lines represent the performance of the noisy (that is, low $n_{\textrm{iter}}$) oracle on the validation data sets, which is imperfect due to incorrect positive labels. Final $F_1$ scores (and associated uncertainties) attained from the classifiers depicted here are listed in Table \ref{tab:noise-comparison}.}
    \label{fig:noise-comparison}
\end{figure}

\begin{table}[]
    \centering
    \begin{tabular}{| c | c | c | c | c | c |}
        \hline
        \multirow{2}{*}{Potential} & \multirow{2}{*}{Query} & \multicolumn{4}{| c |}{$F_1$ Scores}\\
        \cline{3-6}
        & & $n_\textrm{iter} = 5$ & $n_\textrm{iter} = 10$ & $n_\textrm{iter} = 20$ & Default $n_\textrm{iter}$\\
        \hline
        \multirow{4}{*}{2HDM} & BALD & 0.91(1) & 0.974(3) & 0.979(1) & 0.980(2) \\
        & Max Entropy & 0.877(4) & 0.958(8) & 0.975(3) & 0.978(2)\\
        & Variation Ratios & 0.877(6) & 0.957(3) & 0.976(2) & 0.977(2)\\
        & Random & 0.855(4) & 0.919(5) & 0.939(2) & 0.942(4)\\
        \hline
        \multirow{4}{*}{3HDM} & BALD & 0.935(3) & 0.962(2) & 0.963(2) & 0.9658(7)\\
        & Max Entropy & 0.918(6) & 0.950(3) & 0.967(2) & 0.968(2)\\
        & Variation Ratios & 0.927(8) & 0.957(4) & 0.967(1) & 0.967(2)\\
        & Random & 0.913(5) & 0.932(2) & 0.937(2) & 0.9387(8)\\
        \hline
        \multirow{4}{*}{Precustodial} & BALD & 0.851(3) & 0.915(3) & 0.962(7) & 0.9936(7) \\
        & Max Entropy & 0.831(2) & 0.897(4) & 0.942(5) & 0.9925(4)\\
        & Variation Ratios & 0.830(5) & 0.894(4) & 0.941(4) & 0.992(1)\\
        & Random & 0.832(4) & 0.898(3) & 0.936(3) & 0.973(2)\\
        \hline
    \end{tabular}
    \caption{\footnotesize The $F_1$ scores achieved by fully trained classifiers with varying degrees of oracle noise, parameterized by lower values for the oracle hyperparameter $n_{\textrm{iter}}$-- these represent the final classifier performances achieved by the training results depicted in Figure \ref{fig:noise-comparison}.}
    \label{tab:noise-comparison}
\end{table}

From this Figure and Table, we can observe some intriguing classifier behavior for suboptimal $n_{\textrm{iter}}$. First, we see that in stark contrast to our results for our default $n_{\textrm{iter}}$ values, depicted in Figures \ref{fig:2HDM-F}-\ref{fig:Precustodial-F}, for noisy oracles (generally those which achieve $F_1$ scores of $\lsim 0.95$), we actually find mild to moderately-significant discrepancies between classifier performance based on the active learning query strategy. In particular, BALD (based on mutual information) appears to substantially outperform other query strategies. Inspecting the accuracy of the suboptimal oracles on the training set in the 2HDM classifier, we can see that this stems from the fact that BALD selects points which are significantly less likely to be misclassified than the maximum entropy or variation ratios-based strategies: For $n_{\textrm{iter}} = 5$, our 2HDM oracle achieves an average $F_1$ score of $0.87 \pm 0.03$ on the training data set generated by a BALD query strategy (accuracy is determined using the symbolic necessary and sufficient bounded-from-below conditions of \cite{Ivanov:2006yq}, while mean and error are extracted by performing 5 identical experiments). Meanwhile, it achieves an $F_1$ score of $0.60 \pm 0.04$ and $0.62 \pm 0.04$ for the maximum entropy and variation ratios query strategies. This difference vanishes when $n_{\textrm{iter}} = 20$, at which point the oracles' $F_1$ scores on the training data are identical up to statistical error, at $0.996 \pm 0.001$-- where we furthermore see no significant performance difference between the different active learning query strategies.

The reason that BALD selects points more likely to be classified well by an inferior oracle is less immediately clear. A potential intuitive reason is that, while maximum entropy and variation ratios both heavily incentivize the selection of points near the decision boundary (points with high aleatoric uncertainty). Points along the decision boundary, being extremely close to the bounded-from-below region, should intuitively have only a very small region in vev space in which the quartic potential is negative if they are not bounded-from-below, making it more probable that our approximate oracle will miss these regions if $n_{\textrm{iter}}$ is small. Meanhile BALD, by selecting points with high mutual information (which we take to estimate epistemic uncertainty), tends to prefer points for which there is a significant variance in the neural network's predicted confidence scores while being somewhat far from the decision boundary-- that is, in regions about which the neural network can expect to be quite certain about, once it has some additional information.\footnote{Simplistically, this can be shown by considering the expression for mutual information of Eq.(\ref{eq:mutual-information}) for a case with two forward passes through the network, which give $c_0 - \delta$ and $c_0 + \delta$. Then, we see that mutual information achieves a minimum when $c_0 = 0.5$, and increases as the mean departs from the decision boundary.} In fact, we see that in many cases, the neural network with the BALD query strategy significantly outperforms the oracle on which it was trained. Referring back to the training outcomes with our default $n_{\textrm{iter}}$ values, depicted in Figures \ref{fig:2HDM-F} - \ref{fig:Precustodial-F}, it appears that as a rough rule, this overperformance seems to persist whenever the $F_1$ score of the noisy oracle on the validation data is more than $O(10^{-2})$ less than the optimal performance achieved by the classifier trained with the default, robust $n_{\textrm{iter}}$-- after this point, the oracle appears to achieve uniformly higher performance than the neural network classifiers.

To get a better visual understanding of the degradation of classifier performance from oracle noise, and explore the effect of oracle noise on uncertainty quantification,in Figures \ref{fig:2HDM-noise-uncertainty}, \ref{fig:3HDM-noise-uncertainty}, and \ref{fig:Precustodial-noise-uncertainty}, we depict the performances on validation data of classifiers with our different active learning query strategies over the course of active learning for different values of $n_{\textrm{iter}}$, along with the performance over only points that have a mutual information value of below the 0.95 quantile of the respective classification in the validation set, in the same manner as Figures \ref{fig:2HDM-F} - \ref{fig:Precustodial-F}. 

\begin{figure}
    \includegraphics[width=6in]{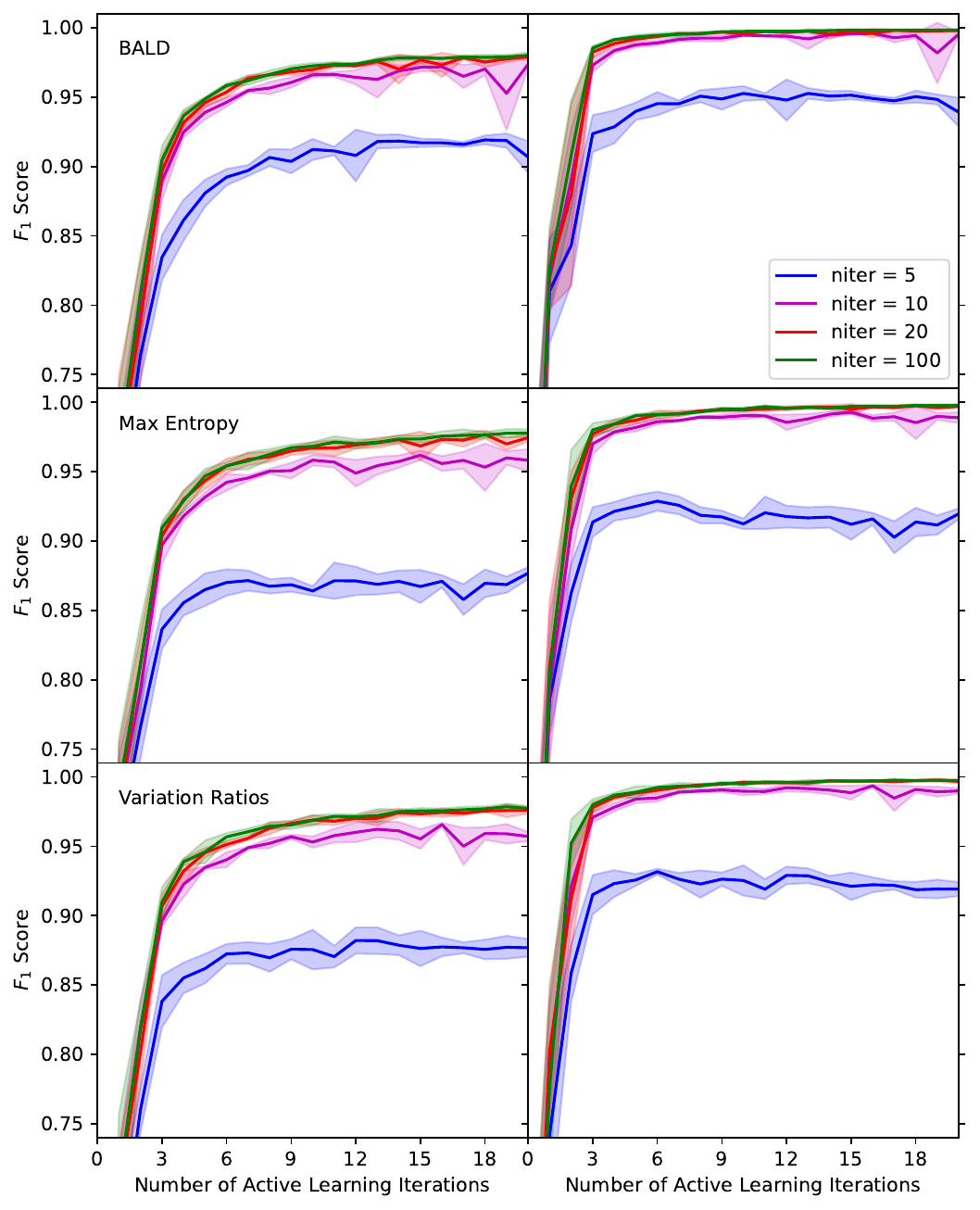}
    \caption{For the general 2HDM potential: (Left) The $F_1$ score
    as a function of the number of active learning iterations performed with the training data labelled by an oracle with $n_{\textrm{iter}} = 5$ (blue), 10 (magenta), 20 (red), and the default value of Table \ref{tab:default-hyperparams} (green). Plots are made with a classifer trained using the BALD query strategy (top), maximum entropy (middle), and variation ratios (bottom). (Right): As on the left, but with the validation set altered by removing points with mutual information greater than the 0.95 quantile for their respective predicted class, in the same manner as Figures \ref{fig:2HDM-F}-\ref{fig:Precustodial-F}.}
    \label{fig:2HDM-noise-uncertainty}
\end{figure}

\begin{figure}
    \includegraphics[width=6in]{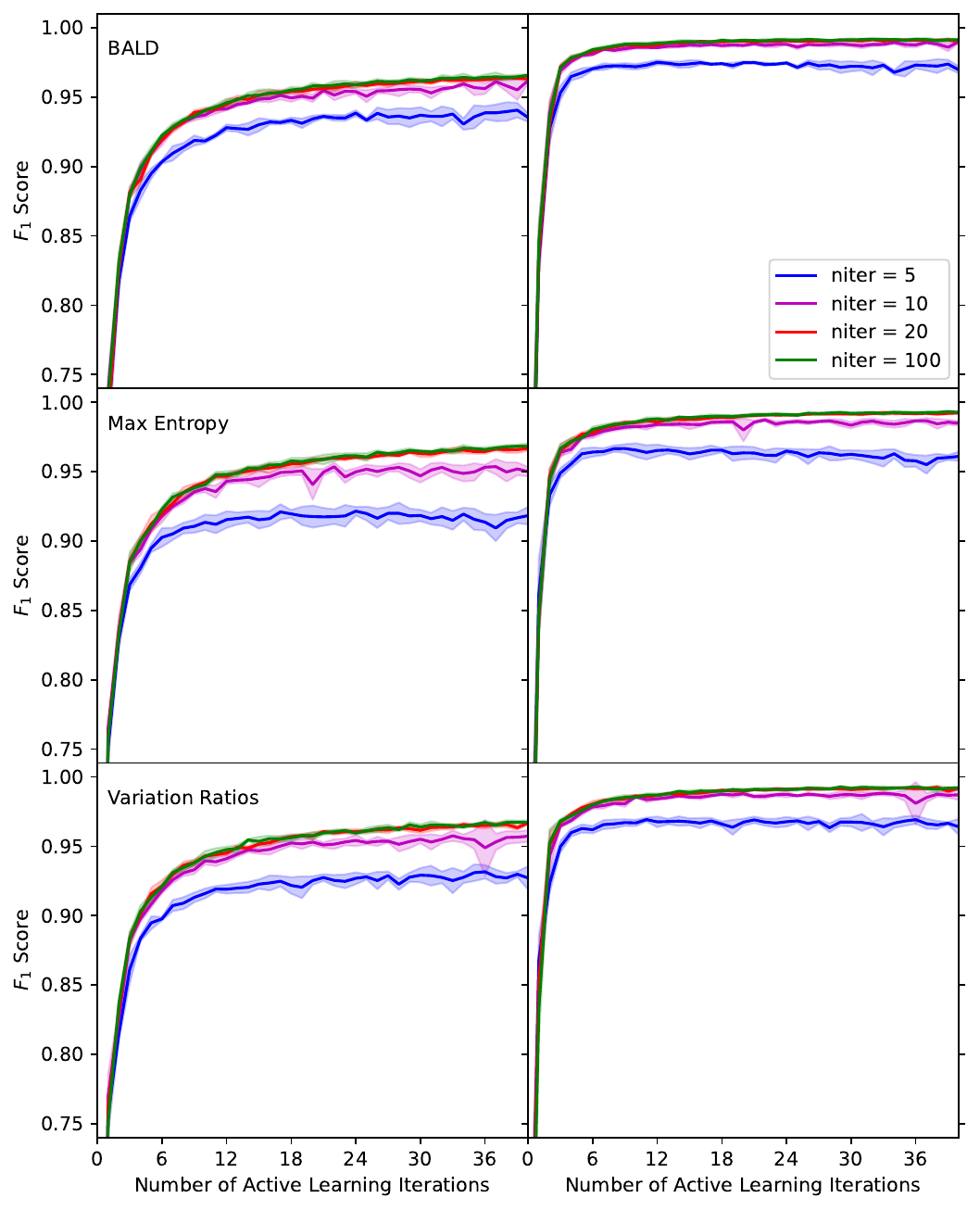}
    \caption{As Figure \ref{fig:2HDM-noise-uncertainty}, but for the 3HDM potential in Eq.(\ref{eq:V-3HDM}).}
    \label{fig:3HDM-noise-uncertainty}
\end{figure}

\begin{figure}
    \includegraphics[width=6in]{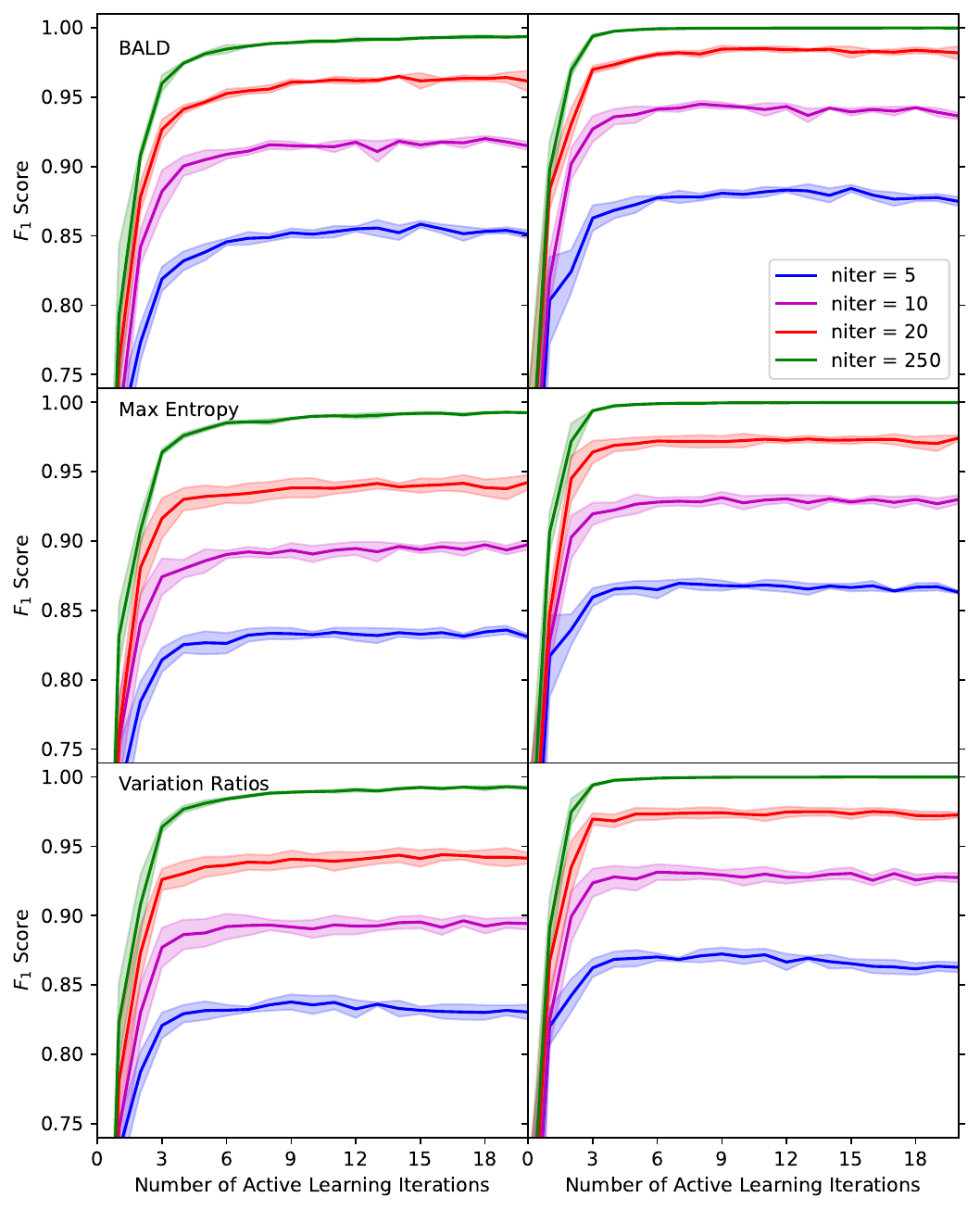}
    \caption{As Figure \ref{fig:2HDM-noise-uncertainty}, but for the precustodial potential in Eq.(\ref{eq:V-PC}).}
    \label{fig:Precustodial-noise-uncertainty}
\end{figure}

In these figures, we find several reassuring results regarding the robustness of our classifier models to oracle noise. First, we see that in the case of an oracle with an $F_1$ score $\gsim 0.99$ (\ie, the $n_{\textrm{iter}} = 20$ case for the 2HDM and 3HDM potentials), there is no significant performance degradation for any of the query strategies-- in turn this suggests that, since the oracle we have used in these experiments is imperfect (and given our robustness tests, likely at a level significantly below the discrepancy between the $n_{\textrm{iter}} = 20$ oracle and our defaults), in all probability we need not be concerned with the effects of oracle noise on our classifiers. Furthermore, we see that even in cases with significantly degraded performance, mutual information remains a highly effective gauge of the reliability of given point predictions: By omitting highly uncertain points, we see significant improvements in the classifier performance over all potentials, query strategies, and $n_{\textrm{iter}}$ values. It should be noted, however, that $F_1$ scores remain significantly degraded relative to the default oracles when a very noisy (\eg, $n_{\textrm{iter}}  = 5$) oracle is used, even when highly uncertain points are omitted from the analysis.

\section{Discussion and Conclusions}\label{sec:conclusion}

In this work, we have presented a novel procedure using active learning to automatically approximate the strictly bounded-from-below parameter space of arbitrary renormalizable scalar potentials. We have further shown that the classifiers produced by this method can significantly outperform existing methods of producing approximate symbolic bounded-from-below conditions, while requiring comparable execution time on personal electronics and demanding no special insight from the model builder beyond the ability to write down the quartic part of potential. These approximate methods also provide uncertainty quantification that can be readily used to gauge a prediction's reliability, allowing the outputs of the classifier to be trusted for use in phenomenological scans of the parameter space in a broad array of BSM scalar models. While in this work we have restricted our analysis to three scalar potentials with complicated (or unknown) bounded-from-below conditions, we have also created a public Python package, called BFBrain \cite{Wojcik2023}, which allows for our procedure to be replicated for any user-specified renormalizable scalar potential with detailed installation and usage instructions available online.

Our study also suggests a number of directions for further improvement and study of our procedure. For example, it remains unclear what precise factors influence the classifier's maximum performance on different potentials, although we have observed small but significant discrepancies in classifier accuracy between the scalar potentials we have explored here. Additional data regarding classifier performance for potentials with a wider range of dimensionalities for their vev parameters and quartic coefficients, therefore, could potentially establish a rough sense of scaling behavior that we cannot deduce here. Further refinements of the active learning strategy and neural network architecture may also result in significant performance gains over what we have thus far demonstrated-- for example, developing a procedure to dynamically alter the procedure for generating the pool of candidate points $L$ as active learning continues may improve the efficiency with which the algorithm explores the parameter space, while permutation symmetries of some potentials, such as those of multi-Higgs doublet models, may be implemented either explicitly or via a graph neural network to further narrow the parameter space that must be explored. Additional study can also be made of stopping criteria for active learning, namely the conditions under which we might consider the network to have achieved high performance (and can stop active learning), without relying on a labelled validation set (which may be unreasonably expensive to produce for some potentials)-- some work in this direction in other active learning contexts has been considered in, \eg, \cite{altschuler2019stopping, ishibashi2021stopping}.

More broadly, we have seen in our work that machine learning techniques continue to find new applications in BSM model building, joining other applications such as active learning experimental and theoretical parameter space constraints \cite{caron2019constraining,hammad2023exploration,goodsell2023active}, reinforcement learning of flavor physics parameters \cite{Harvey:2021oue,Nishimura:2023wdu}, and even automated generation of simplified models \cite{Waltenberger:2020ygp}. We have seen that the problem of finding boundedness-from-below conditions for BSM scalar potentials, for which approximate numerical conditions are often overlooked in favor of analytically tractable, but often highly imprecise, algebraic conditions, can be addressed to a high degree of accuracy using publicly available machine learning tools and hardware accessible to most researchers. It stands to reason that considerably more applications of recent developments in the machine learning space exist for BSM model building, where a significant number of other complex and intractable problems exist.

\section{Acknowledgements}
This work was supported by the U.S. Department of Energy under the contract number DE-SC0017647. The author would like to thank Lisa Everett and Matthew Sullivan for useful discussions related to this project.

\appendix
\section{A Review of Binary Classifier Performance Metrics}\label{appendix:PerformanceEval}
In the interest of maximizing the accessibility of this work, here we briefly review the definitions of various metrics that we employ to rate the performance of our binary classifiers-- specifically precision, recall, and $F_1$ score. The definitions we use here are all standard, and therefore a reader already familiar with these terms will find nothing new here.

We extract all three of the aforementioned metrics from a \emph{confusion matrix} of the classifier on some labelled data set (in this case our various validation data sets. Assuming a binary classifier assigns a ``positive'' or ``negative'' classification to each point in the data set, the elements of the confusion matrix are:
\begin{itemize}
    \item \textbf{True Positives (TP):} The number of points which the classifier \emph{correctly} identifies as positive.
    \item \textbf{False Positives (FP):} The number of points which the classifier \emph{incorrectly} identifies as positive.
    \item \textbf{True Negatives (TN):} The number of points which the classifier \emph{correctly} identifies as negative.
    \item \textbf{False Negatives (FN):} The number of points which the classifier \emph{incorrectly} identifies as negative.
\end{itemize}
Then, we can compute the precision, recall, and $F_1$ score as
\begin{align}
    \textrm{Precision} = \frac{TP}{TP + FP}, \;
    \textrm{Recall} = \frac{TP}{TP + FN}, \;
    F_1 = \frac{2 TP}{2 TP + FP + FN}.
\end{align}
It can be derived that the $F_1$ score is simply the harmonic mean of precision and recall.

We can equivalently think of precision and recall as conditional probabilities in terms of two events: That a point is of the positive class (which we can denote as $+$) and that a point is classified in that class (which we call $C+$). Then, precision is the conditional probability $P(+|C+)$, while recall is the conditional probability $P(C+|+)$. If we let the bounded-from-below class of potentials represent our positive class, then the $F_1$ score incorporates these two probabilities and gives us a key metric for the performance of our classifier: It's ultimately dependent on what fraction of the bounded-from-below parameter space the classifier correctly identifies (recall) and what fraction of points classified as bounded-from-below are classified correctly (precision).

\section{Bayesian Inference, Monte Carlo Dropout, and Concrete Dropout}\label{appendix:BayesianDropout}

This appendix reviews the Bayesian deep learning techniques employed in this paper, specifically Monte Carlo dropout and its variant concrete dropout, in greater detail. To begin, we note that in Bayesian inference, we wish to learn a posterior distribution on some set of model parameters $\mathbf{\theta}$ (in this case, these are distributions on the weights of the Bayesian neural network), given a set of training data $\mathbf{X}$ and their corresponding labels $\mathbf{Y}$-- in other words, we need to learn the distribution
\begin{align}\label{eq:exact-posterior}
    p(\mathbf{\theta} | \mathbf{X}, \mathbf{Y}) = \frac{p( \mathbf{Y} | \mathbf{X}, \mathbf{\theta} ) p(\mathbf{\theta})}{\int p (\mathbf{Y} | \mathbf{X}, \mathbf{\theta} ) p(\mathbf{\theta}) \, d \mathbf{\theta}},
\end{align}
where in Eq.(\ref{eq:exact-posterior}) we have invoked Bayes's theorem, and $p(\mathbf{\theta})$ is a prior distribution on the model parameters. It is unfortunately not tractable to determine the posterior distribution (or in particular, do the integration in the denominator of the equation) exactly in our problem case-- or in fact for all but a few extremely simple statistical learners. Variational inference is a common technique to approximate the posterior-- in which we posit a tractable distribution $q _\phi (\mathbf{\theta})$, which depends on some parameters $\phi$, and then find parameters $\phi$ that minimize the Kullback-Leibler (KL) divergence between $q_\phi(\mathbf{\theta})$ and the intractable exact posterior. The exact posterior is unknown, but maximizing the log evidence lower bound,
\begin{align}\label{eq:ELBO}
    \mathcal{L}_{VI} \equiv \int q_\phi (\mathbf{\theta}) \log p (\mathbf{Y} | \mathbf{X}, \mathbf{\theta} ) \, d \mathbf{\theta} - KL \big( q_\phi (\mathbf{\theta}) || p(\mathbf{\theta}) \big),
\end{align}
where $KL$ denotes the KL divergence (note that the divergence in this expression is between two known distributions, $q_\phi(\mathbf{\theta})$ and the prior), is equivalent to minimizing the KL divergence between $q_\phi (\mathbf{\theta})$ and the exact posterior \cite{bishop2006pattern}. Provided all the terms in Eq.(\ref{eq:ELBO}) are differentiable with respect to the parameters $\phi$, we can optimize the expression with respect to the various parameters of the approximate distribution (for example, mean weights and variances of the model weights).

In \cite{gal2016dropout}, the authors demonstrated a correspondence between the optimization objective in Eq.(\ref{eq:ELBO}) for certain classes of Bayesian neural networks and the optimization objective of a conventional feed-forward neural network with dropout applied before each weight layer. In general, such a network has its predictions given as a composition of layers which map some input $\mathbf{x}$ to an output. The output of the $k^{\textrm{th}}$ layer in the neural network is given by
\begin{align}\label{eq:dropout-layer}
    \mathbf{s}_{k+1} = \sigma_k ( \mathbf{W}_k \cdot \mathbf{z}_k \cdot \mathbf{s}_k + \mathbf{b}_k ),\\
    \mathbf{z}_k \sim \textrm{Bernoulli}(1-p_k). \nonumber
\end{align}
Here, $\mathbf{W}_k$ is the weight matrix for the $k^{\textrm{th}}$ layer, $z_k$ is a diagonal matrix of Bernoulli random variables so that each element of the layer input vector $\mathbf{s}_k$ has a probability $p_k$ of being set to 0 with each forward pass through the network, $\mathbf{b}_k$ is a bias vector, and $\sigma_k$ is some nonlinear activation function (in our case, a ReLU or a sigmoid activation function). For a training data set consisting of $N$ samples with a binary classifier, optimizing the loss \cite{gal2016dropout,gal2017concrete}
\begin{align}\label{eq:mc-dropout-loss}
    BCE + \frac{1}{N} \sum_{k = 1}^L \bigg\{ \frac{l^2 (1- p_k)}{2} ||\mathbf{W}_k||^2_2 + \frac{1}{2} ||\mathbf{b}_k ||^2_2 + K_{k} \big( p_k \log p_k + (1 - p_k) \log (1 - p_k) \big)   \bigg\},
\end{align}
approximates variational inference of a Bayesian neural network with prior distributions $\mathcal{N}(0, 1/l^2)$ over the weights and biases, where $BCE$ denotes binary cross-entropy loss (averaged over all samples in the training set, or a representative batch of them) and evaluated with a single forward pass with dropout), $L$ is the number of layers in the neural network, and $K_k$ is the number of neurons in the $(k-1)^{\textrm{th}}$ layer of the network (that is, the number of inputs into the $k^{\textrm{th}}$ layer). Furthermore, as discussed in the main text, obtaining the outputs of the neural network at test time by performing multiple evaluations of the same inputs with dropout will approximate drawing from the distribution of outputs of the Bayesian neural network-- this technique is known as Monte Carlo dropout.

In Monte Carlo dropout, we see that the parameters $\phi$ of our approximate posterior $q_\phi( \mathbf{\theta})$ are the weights and biases of our dropout network ($\mathbf{W}_k$ and $\mathbf{b}_k$, respectively), as well as the dropout probabilities for each layer $p_k$.
While the loss in Eq.(\ref{eq:mc-dropout-loss}) is differentiable with respect to $\mathbf{W}_k$ and $\mathbf{b}_k$, it is not possible to differentiate with respect to the dropout probability $p_k$. Therefore, in order to minimize this part of the objective one must perform a scan over possible dropout parameters for each layer. Given that our use case includes continuously augmenting the training data set, in principle we would then need to perform this scan for each active learning iteration, which is clearly unworkable. If instead we arbitrarily specified dropout probabilities, the quality of our uncertainty estimates would suffer-- as discussed in \cite{gal2017concrete}, a model with fixed dropout probability can only reduce its predictive variance by reducing the magnitude of its weights-- which can ultimately degrade model performance (after all, a network with only weight values of exactly 0 will have zero predictive variance, but also no predictive power). In practice, this means that for fixed dropout probability, a larger amount of training data may not a priori reduce the network's predictive variance \cite{verdoja2020notes}, in turn suggesting that, for example, comparing uncertainty estimates between active learning iterations will be nonsensical (or at the very least, unprincipled). More formally, failing to tune the dropout probability is ultimately failing to fully optimize the KL divergence between the approximate posterior $q(\mathbf{\theta})$ and the true posterior $p (\mathbf{\theta} | \mathbf{X}, \mathbf{Y})$. Hence, the potential predictive power (and associated uncertainty) of the Bayesian classifier is not extracted.

Concrete dropout \cite{gal2017concrete} circumvents this difficulty by approximating the Bernoulli random variables $\mathbf{z}_k$ with corresponding draws from the concrete distribution,
\begin{align}
    \textrm{sigmoid} \bigg( \frac{1}{t} (\log p - \log(1-p) + \log u - \log (1-u) )\bigg), \; u \sim \textrm{Unif}(0,1),
\end{align}
where $t$ is some temperature parameter. Hence, instead of being either retained with probability $p_k$ or dropped, an input to the $k^{\textrm{th}}$ layer of the neural network will be scaled according to a draw from this distribution. For small $t (\lesssim 0.1)$, the concrete distribution well approximates the Bernoulli distribution, but is smooth and can be differentiated with respect to the probability $p$. As a result, the loss in Eq.(\ref{eq:mc-dropout-loss}) can now be optimized over its weights, biases, \emph{and} dropout probabilities via gradient descent methods (\ie, conventional neural network training). In the BFBrain package, we implement concrete dropout in Tensorflow 2, borrowing heavily from both the code presented in \cite{gal2017concrete} and an earlier adaptation to Tensorflow 2 given by \cite{Amerio2022}, with only marginal changes made to facilitate model portability.

\section{The Learner}\label{appendix:Learner}
In this appendix, we discuss the details of our specific implementation of a Bayesian neural network in our 
As discussed in Section \ref{sec:bayesian-nn}, for our learner we have selected a Bayesian neural network, with Bayesian inference approximated using concrete dropout. Having outlined the fundamental concepts used in the Bayesian neural network in that section, here we restrict our attention to the particularities of our learner's construction. Our basic network is a Bayesian multilayer perceptron (MLP) implemented in Tensorflow \cite{tensorflow2015-whitepaper}, which takes as input a vector $\vec{\lambda}$ and outputs confidence scores in a probability distribution between 0 and 1-- if the average score is close to 0, the network is confident that the point is not bounded from below, while if the average score is close to 1, the network is confident that it is.
The quartic coefficients are fed into the MLP, which has some number (varying between 3 and 5 in our experiments) of hidden layers, each with 128 neurons.\footnote{We found little appreciable change in performance when increasing or decreasing the number of neurons in a given layer, as long this number was $\gsim O(100)$.} All weights have a Bayesian prior distribution $\mathcal{N}(0, 1/l^2)$, where we have found empirically that $l = 0.1$ gives good results for all potentials that we consider, and that somewhat smaller or larger values of $l$ do not substantially affect the results. 

Exact Gaussian process inference is approximated by concrete dropout, discussed in Section \ref{sec:bayesian-nn} and Appendix \ref{appendix:BayesianDropout}, where dropout is applied to the outputs from all hidden layers, the hidden layer activation functions are ReLU, and the output node activation function is a sigmoid. In \cite{gal2017concrete}, the authors also suggest that concrete dropout should be applied to the input layer. If the training set is sufficiently informative, then they find that the dropout probability for the input layer will converge to 0. However, we find this practice somewhat degrades the performance of our neural network, both in raw accuracy and quality of the uncertainty metrics. This discrepancy likely stems from the fact that the study in \cite{gal2017concrete} considered image classification, where the inputs are very high-dimensonal and neural network must learn an appropriate low-dimensional feature representation of the raw input. In our case, the inputs are already low-dimensional and condensed into extremely informative features, namely the quartic potential coefficients themselves -- ignoring any one of these inputs will make a rigorous determination of the potential's boundedness-from-below impossible. In lieu of applying concrete dropout to the input layer, we instead apply regularization terms to the network to render the learner equivalent to one with concrete dropout applied to the input layer (see Eq.(\ref{eq:mc-dropout-loss})), but that layer's dropout probability is kept constant at 0.

We can restrict our neural network architecture to leverage the fact that boundedness-from-below of a scalar potential is invariant to a positive rescaling of all the quartic coefficients fairly easily: By setting all neuron biases, as defined in Appendix \ref{appendix:BayesianDropout}, Eq.(\ref{eq:dropout-layer}), to 0, it is trivial to show that the neural network's output is restricted to a sigmoid of a homogenous function in the inputs, which in turn guarantees that the classifications will demonstrate scale invariance. Because a large positive rescaling will still affect the neural network's confidence scores (for example, the entropy of an uncertain output might be significantly reduced simply by scaling the input by a large positive number), in any practical inference task we will still apply a preprocessing layer to project the input onto the unit hypersphere. However, by limiting the functional form of our network's outputs to reflect the character of the problem, we still narrow our hypothesis space and have found that we achieve a more performant trained classifier.

After new training points are acquired in each training round, the neural network's weights are randomly re-initialized and the network is re-trained on the entire set of training data thus far accrued. This re-initialization prevents the network from overweighting of training data that was sampled in an early active learning iteration, which we find eventually leads to a collapse of the model's epistemic uncertainty estimates. The loss objective is, as is standard for classifiers of this type, given by binary cross entropy (as well as the regularization terms discussed in Appendix \ref{appendix:BayesianDropout}). Optimization is accomplished with the Adam algorithm \cite{kingma2014adam} with a learning rate of $0.001$, and continues for $2 \times 10^4$ epochs or until the loss on the training data set fails to decrease for 100 consecutive epochs (in which case the weights which yielded the smallest loss are restored at the end of training).

A final remark on the details of our neural network training is in order, regarding our arrangement of the training data into batches. It is common practice, in order to leverage the high degree of parallelization possible on a GPU, to compute the training loss for a neural network for a large batch of inputs simultaneously, and then apply updates to the weights to minimize the mean loss within that batch. Throughout this paper, we shall take our batches to be large enough to contain the \emph{entire} training data set-- this will result in more stable performance and uncertainty evaluation, at the expense of limiting the potential size of our training data sets to be small enough to fit within GPU memory. We have found that an alternative strategy, with small batch sizes of $5\times 10^3$ points, results in a mild degradation of performance and training stability, so for the somewhat simple models we consider for our experiments in this work (and only modestly-sized training data sets), we adhere to placing all training data in a single batch. The BFBrain package permits the use of smaller batches, to accommodate larger models and/or training data that a particular scalar potential might necessitate.

\section{Additional Slice Ensemble Data}\label{appendix:slice-ensemble}

In this Appendix, we include additional figures in the format of Figures \ref{fig:2HDM-slice-image}-\ref{fig:Precustodial-slice-image} for other slice ensembles that we have considered, in addition to a table of the quartic coefficients used to generate all figures appearing in the work.

\begin{figure}
    \centerline{\includegraphics[width=3.2in]{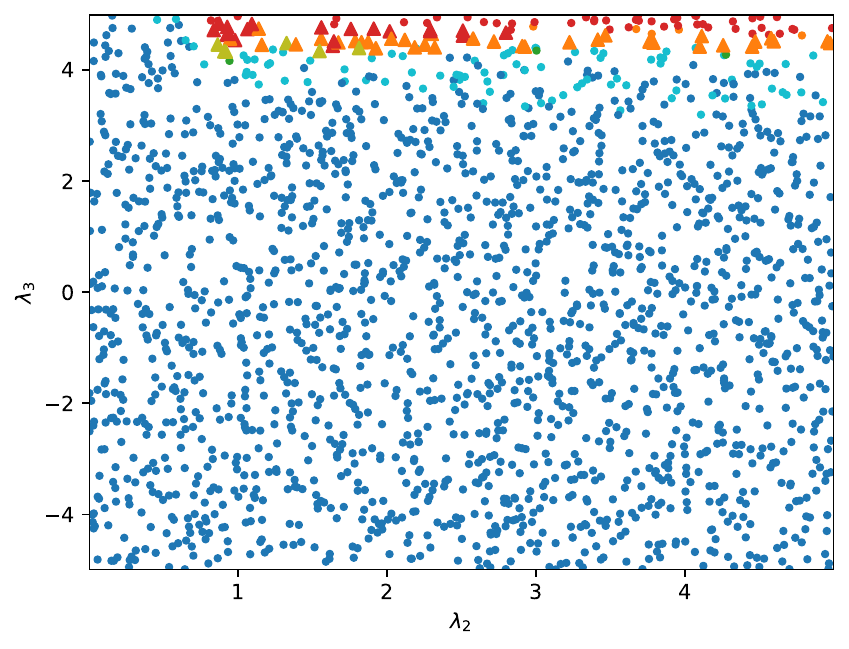}
    \hspace{-0.25cm}
    \includegraphics[width=3.2in]{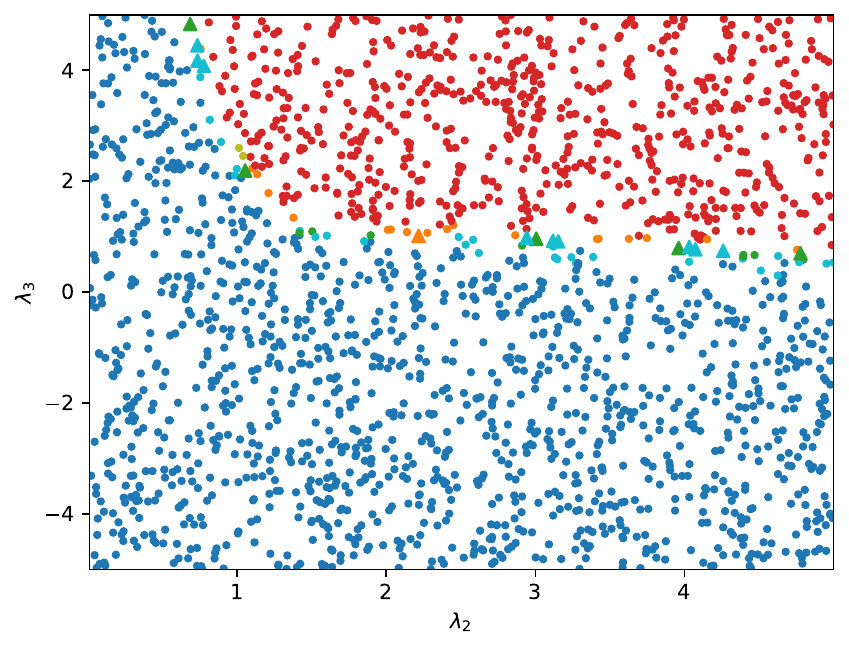}}
    \centerline{\includegraphics[width=3.2in]{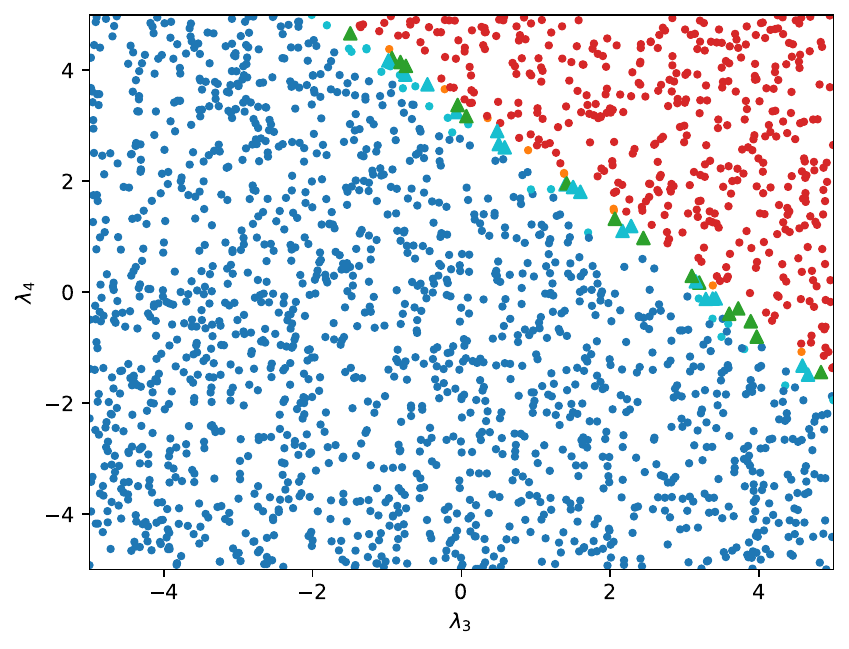}
    \hspace{-0.25cm}
    \includegraphics[width=3.2in]{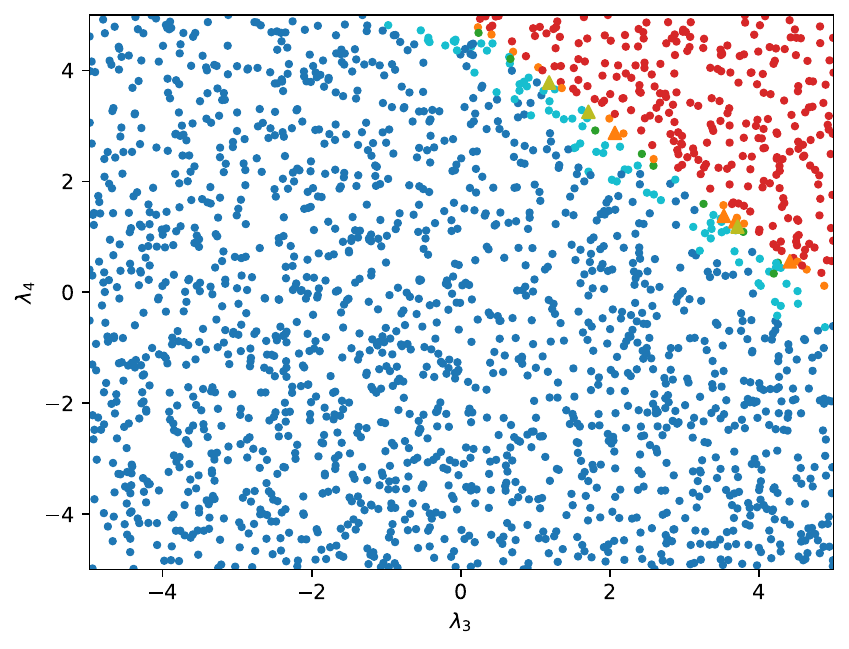}}
    \centerline{\includegraphics[width=3.2in]{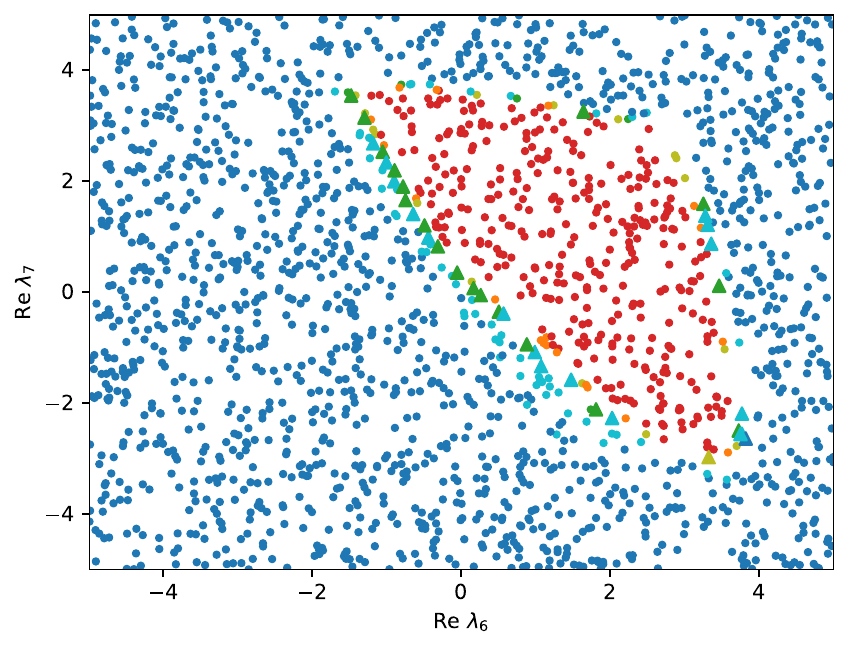}
    \hspace{-0.25cm}
    \includegraphics[width=3.2in]{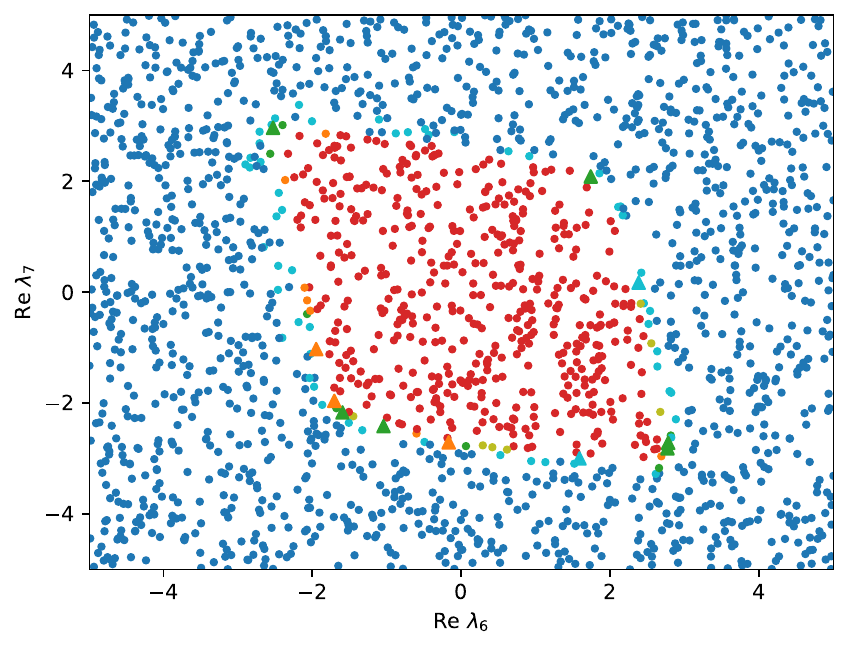}}
    \caption{As Figure \ref{fig:2HDM-slice-image}, but for the $\lambda_2 - \lambda_3$ (top), $\lambda_3 - \lambda_4$ (middle), and $\textrm{Re} \, \lambda_6 - \textrm{Re} \, \lambda_7$ (bottom) slice ensembles of the 2HDM potential, with quartic coefficients defined in Eq.(\ref{eq:V-2HDM}).}
    \label{fig:2HDM-slice-image-appendix}
\end{figure}

\begin{figure}
    \centerline{\includegraphics[width=3.2in]{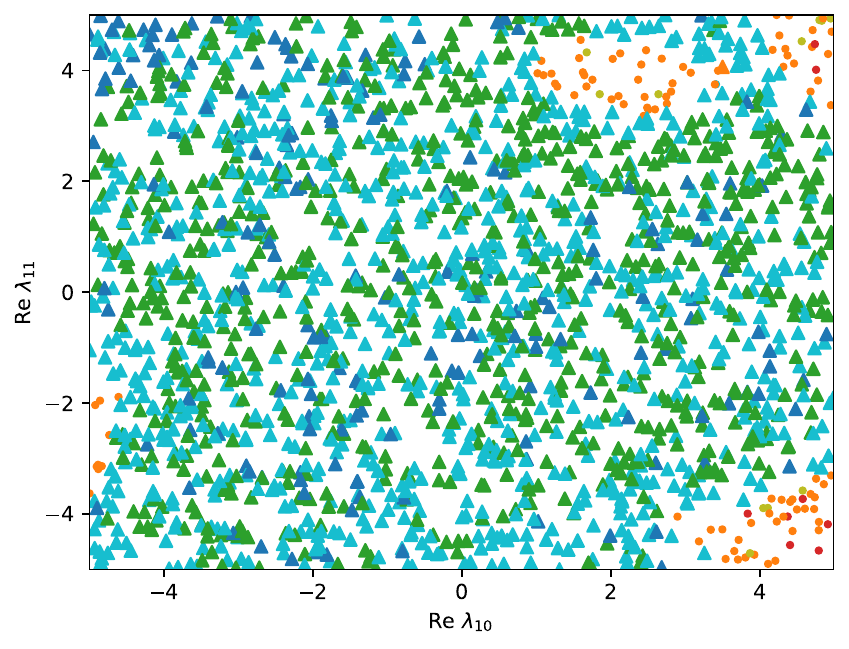}
    \hspace{-0.25cm}
    \includegraphics[width=3.2in]{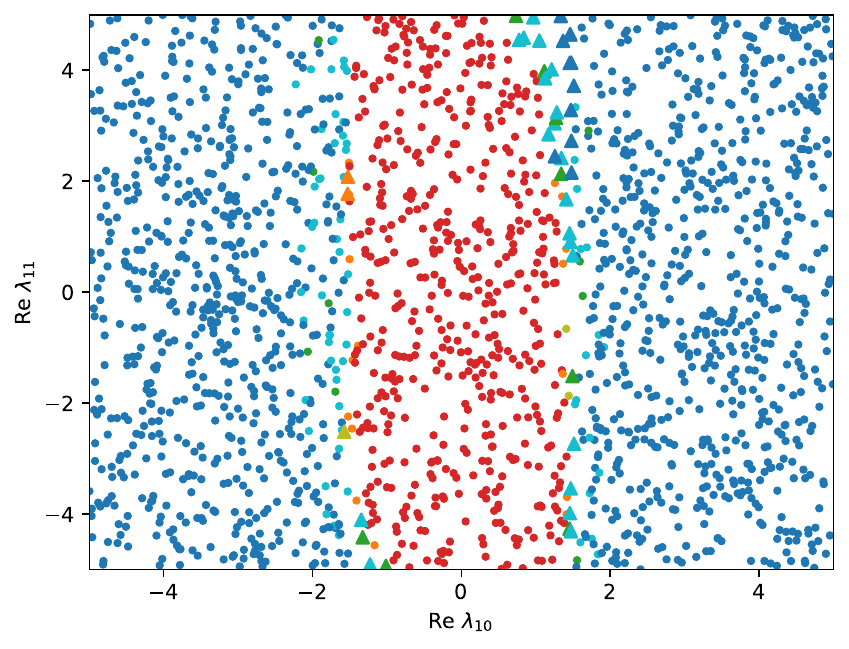}}
    \centerline{\includegraphics[width=3.2in]{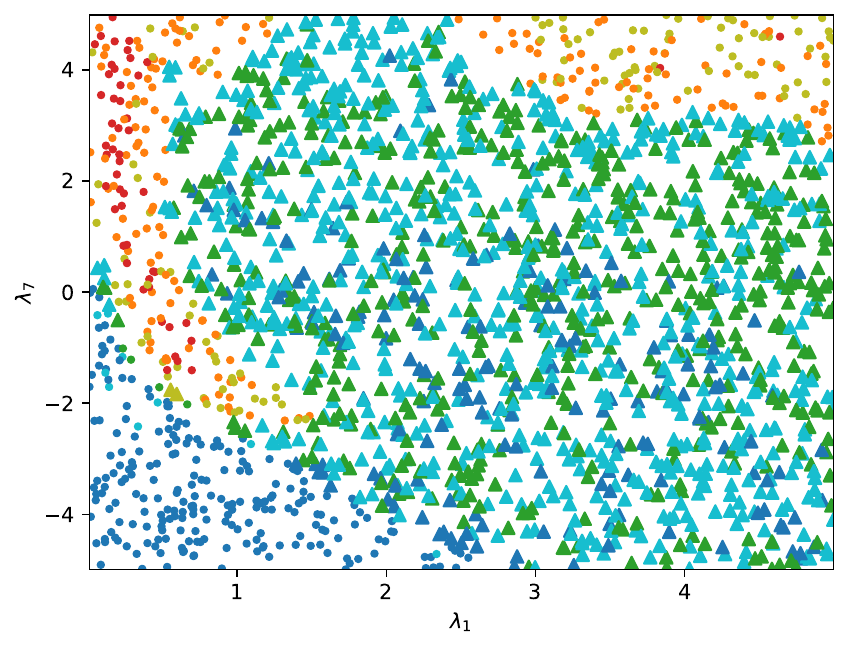}
    \hspace{-0.25cm}
    \includegraphics[width=3.2in]{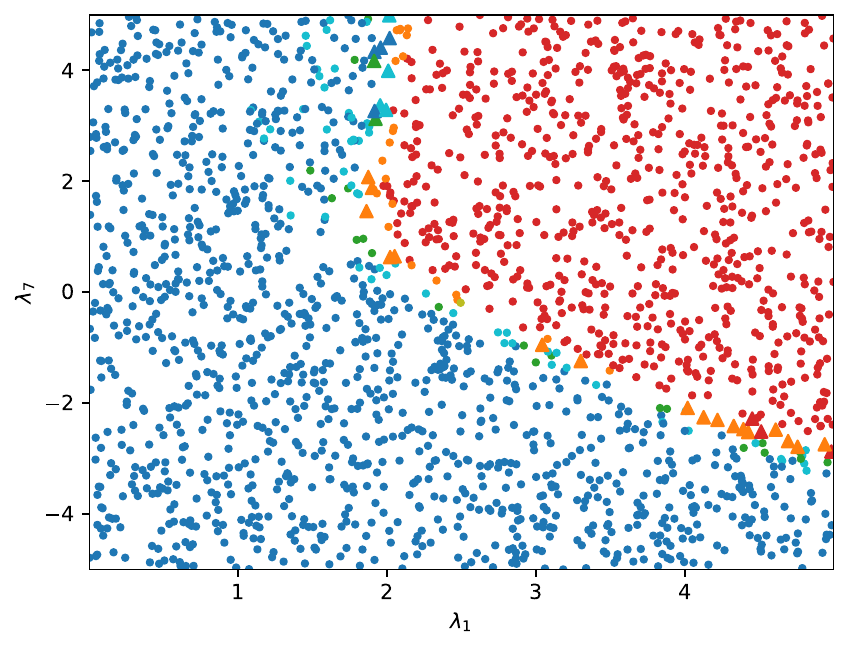}}
    \centerline{\includegraphics[width=3.2in]{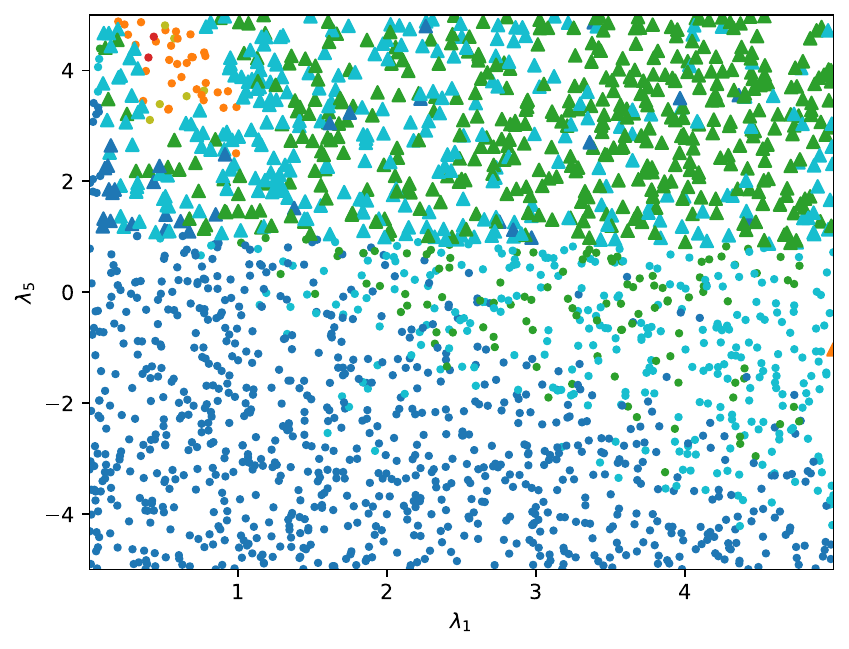}
    \hspace{-0.25cm}
    \includegraphics[width=3.2in]{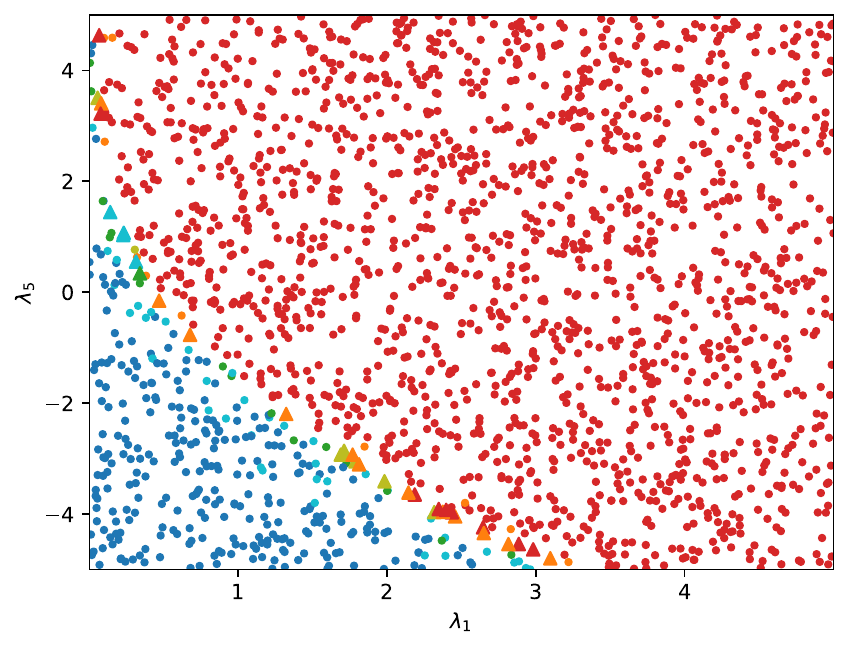}}
    \caption{As Figure \ref{fig:2HDM-slice-image}, but for the $\textrm{Re} \, \lambda_{10} - \textrm{Re} \lambda_{11}$ (top), $\lambda_1 - \lambda_7$ (middle), and $\lambda_1 - \lambda_5$ (bottom) slice ensembles of the 3HDM potential, with quartic coefficients defined in Eq.(\ref{eq:V-3HDM}).}
    \label{fig:3HDM-slice-image-appendix}
\end{figure}

\begin{figure}
    \centerline{\includegraphics[width=3.2in]{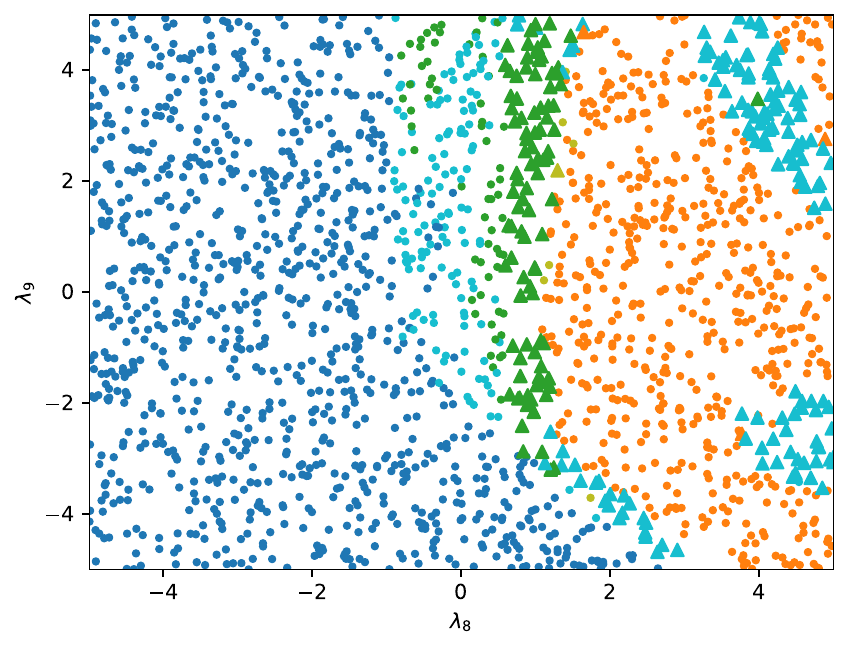}
    \hspace{-0.25cm}
    \includegraphics[width=3.2in]{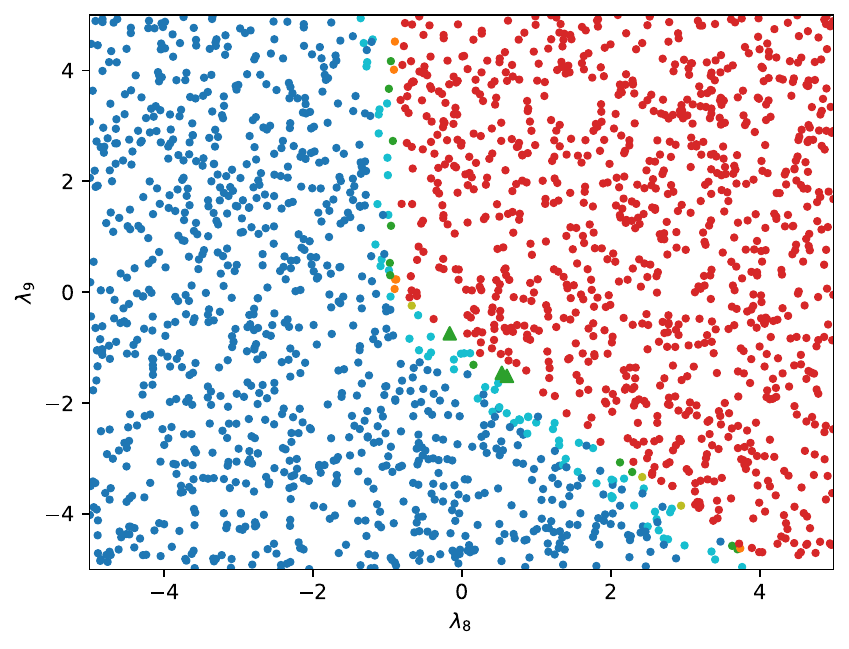}}
    \centerline{\includegraphics[width=3.2in]{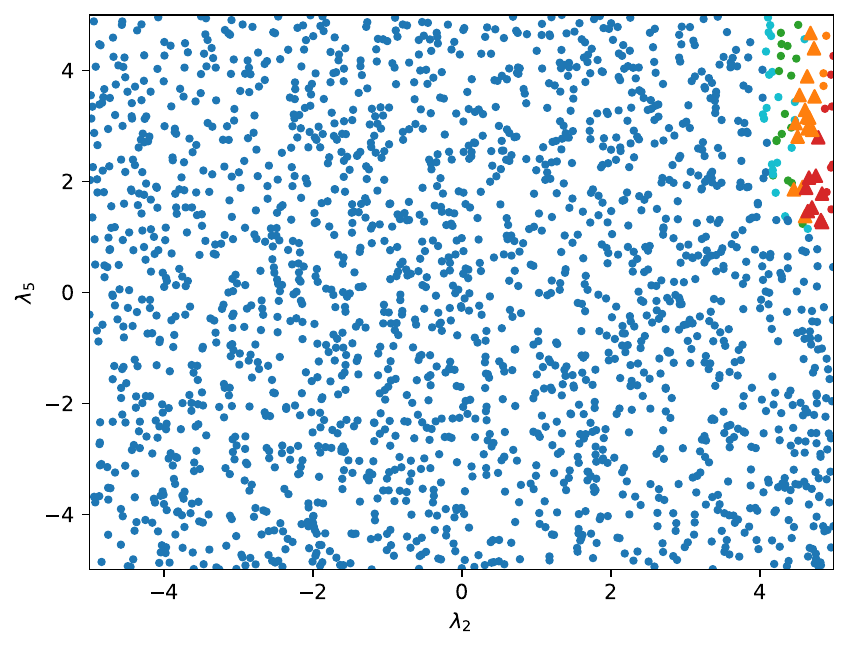}
    \hspace{-0.25cm}
    \includegraphics[width=3.2in]{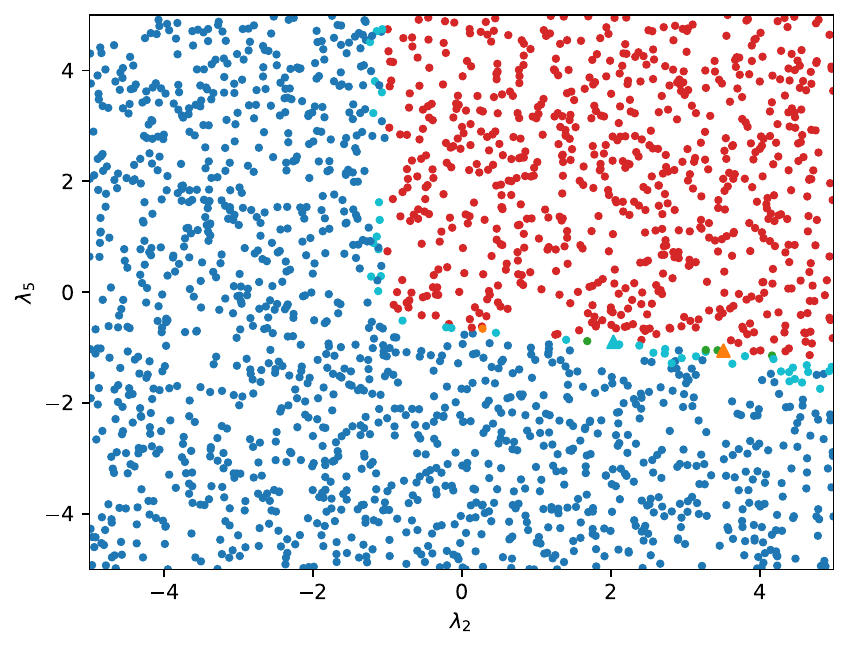}}
    \centerline{\includegraphics[width=3.2in]{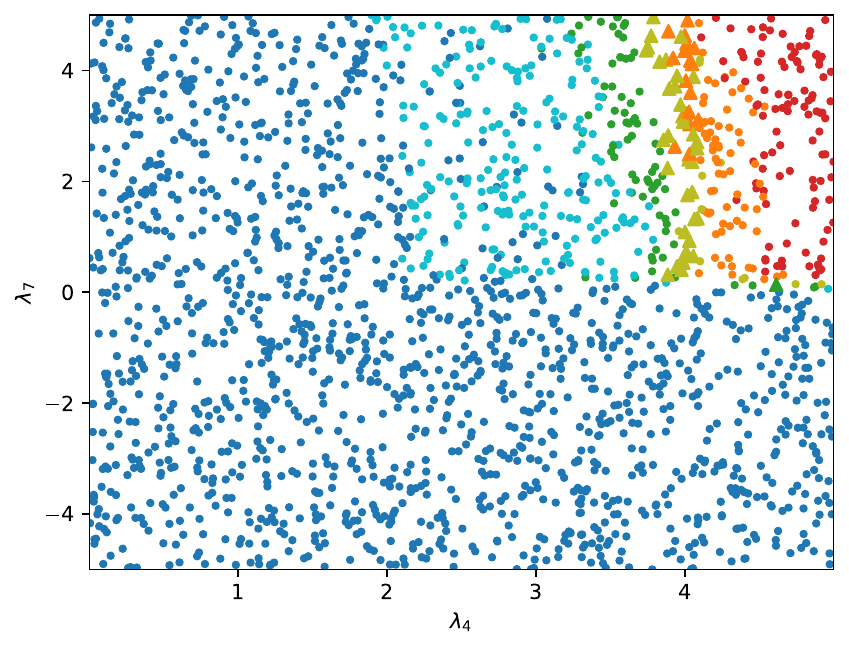}
    \hspace{-0.25cm}
    \includegraphics[width=3.2in]{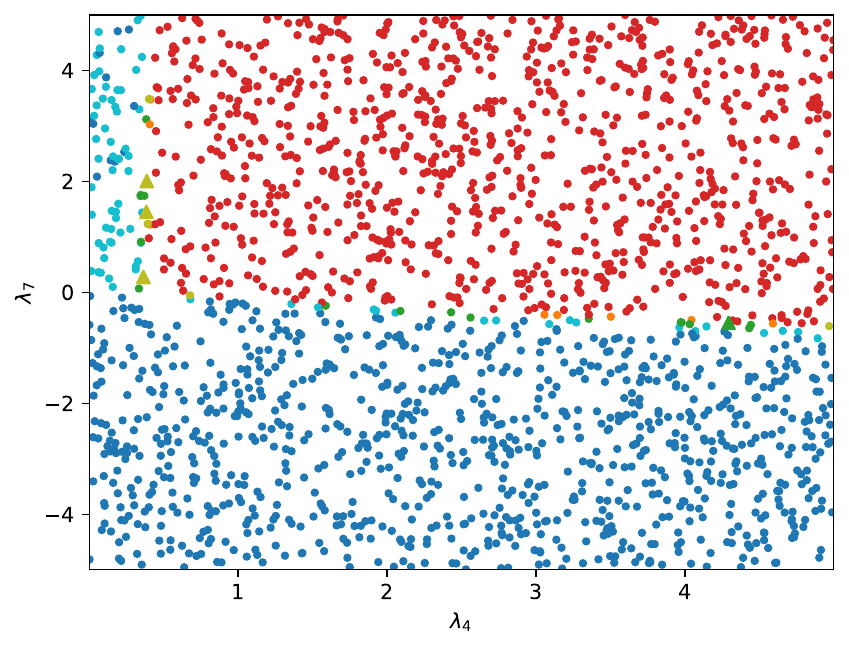}}
    \caption{As Figure \ref{fig:2HDM-slice-image}, but for the $\lambda_8 - \lambda_9$ (top), $\lambda_2 - \lambda_5$ (middle), and $\lambda_4 - \lambda_7$ (bottom) slice ensembles of the precustodial potential, with quartic coefficients defined in Eq.(\ref{eq:V-PC}).}
    \label{fig:Precustodial-slice-image-appendix}
\end{figure}

\begin{table}[]
    \centering
    \begin{tabular}{| c | c | c | c | c |}
    \hline
    & $\lambda_4 - \textrm{Re} \, \lambda_5$ & $\lambda_2-\lambda_3$ & $\lambda_3 - \lambda_4$ & $\textrm{Re} \, \lambda_6 - \textrm{Re} \lambda_7$\\
    \hline
    $\lambda_1$ & 1.99 & 2.43 & 4.32 & 2.54\\
    $\lambda_2$ & 0.36 & $\lambda_2$ & 1.46 & 2.62\\
    $\lambda_3$ & 3.54 & $\lambda_3$ & $\lambda_3$ & 2.35\\
    $\lambda_4$ & $\lambda_4$ & 3.94 & $\lambda_4$ & 4.02\\
    $\textrm{Re} \, \lambda_5$ & $\textrm{Re} \, \lambda_5$ & -3.53 & -1.84 & 4.27\\
    $\textrm{Im} \, \lambda_5$ & 2.19 & 3.59 & -0.92 & -3.25\\
    $\textrm{Re} \, \lambda_6$ & -1.64 & 0.85 & 0.34 & $\textrm{Re} \, \lambda_6$\\
    $\textrm{Im} \, \lambda_6$ & -2.41 & -3.12 & -3.02 & -1.35\\
    $\textrm{Re} \, \lambda_7$ & 0.07 & -0.91 & -0.35 & $\textrm{Re} \, \lambda_7$\\
    $\textrm{Im} \, \lambda_7$ & -1.11 & 1.22 & 2.26 & -0.41\\
    \hline
    \end{tabular}
    \caption{\footnotesize The parameters corresponding to the \emph{worst} binary accuracy among each of the 2HDM slice ensembles. Each column corresponds to a different pair of parameters which are scanned over (see Section \ref{sec:exp-slices} for details). Coefficients listed here are defined in Eq.(\ref{eq:V-2HDM}).}
    \label{tab:2HDM-slice-params-worst}
\end{table}

\begin{table}[]
    \centering
    \begin{tabular}{| c | c | c | c | c |}
    \hline
    & $\lambda_4 - \textrm{Re} \, \lambda_5$ & $\lambda_2-\lambda_3$ & $\lambda_3 - \lambda_4$ & $\textrm{Re} \, \lambda_6 - \textrm{Re} \lambda_7$\\
    \hline
    $\lambda_1$ & 4.22 & 2.52 & 4.53 & 2.96\\
    $\lambda_2$ & 4.10 & $\lambda_2$ & 3.08 & 2.83\\
    $\lambda_3$ & 3.68 & $\lambda_3$ & $\lambda_3$ & 4.30\\
    $\lambda_4$ & $\lambda_4$ & 4.68 & $\lambda_4$ & 4.82\\
    $\textrm{Re} \, \lambda_5$ & $\textrm{Re} \, \lambda_5$ & 0.53 & -1.73 & -4.05\\
    $\textrm{Im} \, \lambda_5$ & -2.26 & -0.13 & 0.56 & -0.81\\
    $\textrm{Re} \, \lambda_6$ & -2.48 & -2.36 & 0.10 & $\textrm{Re} \, \lambda_6$\\
    $\textrm{Im} \, \lambda_6$ & 1.54 & 0.82 & -4.54 & -2.97\\
    $\textrm{Re} \, \lambda_7$ & -3.13 & -0.99 & 0.21 & $\textrm{Re} \, \lambda_7$\\
    $\textrm{Im} \, \lambda_7$ & 1.96 & 1.39 & 2.34 & 1.76\\
    \hline
    \end{tabular}
    \caption{\footnotesize As Table \ref{tab:2HDM-slice-params-worst}, but for the parameters for which the \emph{median} accuracy in the slice ensemble is obtained, computed as described in Section \ref{sec:exp-slices}.}
    \label{tab:2HDM-slice-params-med}
\end{table}

\begin{table}[]
    \centering
    \begin{tabular}{| c | c | c | c | c |}
    \hline
    & $\lambda_4 - \lambda_7$ & $\textrm{Re} \, \lambda_{10} - \textrm{Re} \, \lambda_{11}$ & $\lambda_1 - \lambda_7$  & $\lambda_1 - \lambda_5$\\
    \hline
    $\lambda_1$ & 2.46 & 2.24 & $\lambda_1$ & $\lambda_1$\\
    $\lambda_2$ & 3.37 & 2.37 & 3.10 & 2.69\\
    $\lambda_3$ & 4.22 & 2.28 & 0.83 & 3.98\\
    $\lambda_4$ & $\lambda_4$ & 3.61 & 3.34 & -0.97\\
    $\lambda_5$ & -3.53 & -1.17 & 2.81 & $\lambda_5$\\
    $\lambda_6$ & 4.77 & 2.07 & 0.95 & -1.63\\
    $\lambda_7$ & $\lambda_7$ & 3.75 & $\lambda_7$ & 4.29\\
    $\lambda_8$ & 1.34  & 4.54 & 4.98 & 4.40\\
    $\lambda_9$ & 4.60 & -3.34 & 0.51 & -0.53\\
    $\textrm{Re} \, \lambda_{10}$ & -1.14 & $\textrm{Re} \, \lambda_{10}$ & 2.98 & 1.36\\
    $\textrm{Im} \, \lambda_{10}$ & 1.18 & -1.56 & 3.17 & 2.44\\
    $\textrm{Re} \, \lambda_{11}$ & 2.79 & $\textrm{Re} \, \lambda_{11}$ & -0.59 & 1.61\\
    $\textrm{Im} \, \lambda_{11}$ & -3.01 & -0.76 & -3.40 & -0.45\\
    $\textrm{Re} \, \lambda_{12}$ & 1.63 & 1.65 & -4.38 & -2.74\\
    $\textrm{Im} \, \lambda_{12}$ & -4.59 & -2.87 & -1.41 & -3.41\\
    \hline
    \end{tabular}
    \caption{\footnotesize As Table \ref{tab:2HDM-slice-params-worst}, but for the slices with the worst binary accuracy for the 3HDM potential, with parameters defined in Eq.(\ref{eq:V-3HDM}).}
    \label{tab:3HDM-slice-params-worst}
\end{table}

\begin{table}[]
    \centering
    \begin{tabular}{| c | c | c | c | c |}
    \hline
    & $\lambda_4 - \lambda_7$ & $\textrm{Re} \, \lambda_{10} - \textrm{Re} \, \lambda_{11}$ & $\lambda_1 - \lambda_7$ & $\lambda_1 - \lambda_5$\\
    \hline
    $\lambda_1$ & 3.02 & 3.03 & $\lambda_1$ & $\lambda_1$\\
    $\lambda_2$ & 1.88 & 2.05 & 1.96 & 3.95\\
    $\lambda_3$ & 4.50 & 3.01 & 3.80 & 4.64\\
    $\lambda_4$ & $\lambda_4$ & -2.22 & 1.13 & 0.92\\
    $\lambda_5$ & -4.61 & 0.50 & 1.37 & $\lambda_5$\\
    $\lambda_6$ & 0.52 & 4.56 & 0.42 & 4.47\\
    $\lambda_7$ & $\lambda_7$ & -1.24 & $\lambda_7$ & 2.84\\
    $\lambda_8$ & 2.99 & -0.09 & -2.29 & 2.94\\
    $\lambda_9$ & 1.94 & 1.56 & 3.75 & 1.31\\
    $\textrm{Re} \, \lambda_{10}$ & 2.40 & $\textrm{Re} \, \lambda_{10}$ & 3.06 & 4.19\\
    $\textrm{Im} \, \lambda_{10}$ & 4.74 & -0.13 & -2.57 & -2.63\\
    $\textrm{Re} \, \lambda_{11}$ & 1.62 & $\textrm{Re} \, \lambda_{11}$ & -3.17 & 4.48\\
    $\textrm{Im} \, \lambda_{11}$ & 4.97 & -1.44 & 3.13 & 3.63\\
    $\textrm{Re} \, \lambda_{12}$ & 4.44 & -2.82 & 1.29 & -4.01\\
    $\textrm{Im} \, \lambda_{12}$ & 3.84 & -2.08 & 4.38 & -0.84\\
    \hline
    \end{tabular}
    \caption{\footnotesize As Table \ref{tab:2HDM-slice-params-worst}, but for the slices with the median binary accuracy for the 3HDM potential, with parameters defined in Eq.(\ref{eq:V-3HDM}).}
    \label{tab:3HDM-slice-params-med}
\end{table}

\begin{table}[]
    \centering
    \begin{tabular}{| c | c | c | c | c |}
    \hline
    & $\lambda_5 - \lambda_6$ & $\lambda_8 - \lambda_9$ & $\lambda_2 -\lambda_5$ & $\lambda_4 - \lambda_7$\\
    \hline
    $\lambda_1$ & 3.33 & 2.54 & 3.08 & 0.40\\
    $\lambda_2$ & 6.28 & 4.73 & $\lambda_2$ & 1.26\\
    $\lambda_3$ & 3.16 & 2.11 & -3.84 & -0.61\\
    $\lambda_4$ & 0.81 & 4.51 & 3.74 & $\lambda_4$\\
    $\lambda_5$ & $\lambda_5$ & 4.27 & $\lambda_5$ & 3.29\\
    $\lambda_6$ & $\lambda_6$ & -3.25 & -1.87 & 1.78\\
    $\lambda_7$ & -0.40 & -1.35 & 1.55 & $\lambda_7$\\
    $\lambda_8$ & -0.93 & $\lambda_8$ & -0.43 & 3.02\\
    $\lambda_9$ & 3.81 & $\lambda_9$ & -0.65 & -3.83\\
    $\lambda_{10}$ & -0.80 & -0.41 & 0.45 & 4.05\\
    \hline
    \end{tabular}
    \caption{\footnotesize As Table \ref{tab:2HDM-slice-params-worst}, but for the slices with the worst binary accuracy for the precustodial potential, with parameters defined in Eq.(\ref{eq:V-PC}).}
    \label{tab:Precustodial-slice-params-worst}
\end{table}

\begin{table}[]
    \centering
    \begin{tabular}{| c | c | c | c | c |}
    \hline
    & $\lambda_5 - \lambda_6$ & $\lambda_8 - \lambda_9$ & $\lambda_2 -\lambda_5$ & $\lambda_4 - \lambda_7$\\
    \hline
    $\lambda_1$ & 2.85 & 4.67 & 0.80 & 1.19\\
    $\lambda_2$ & 5.13 & 6.38 & $\lambda_2$ & 4.70\\
    $\lambda_3$ & 3.61 & 5.92 & 2.00 & 2.84\\
    $\lambda_4$ & 3.70 & 0.55 & 1.15 & $\lambda_4$\\
    $\lambda_5$ & $\lambda_5$ & 4.70 & $\lambda_5$ & 3.11\\
    $\lambda_6$ & $\lambda_6$ & -2.00 & 1.11 & 3.35\\
    $\lambda_7$ & 2.79 & 1.26 & 2.69 & $\lambda_7$\\
    $\lambda_8$ & 4.41 & $\lambda_8$ & 4.14 & 2.09\\
    $\lambda_9$ & -3.12 & $\lambda_9$ & -2.51 & -2.72\\
    $\lambda_{10}$ & -1.44 & -0.71 & 1.64 & -3.36\\
    \hline
    \end{tabular}
    \caption{\footnotesize As Table \ref{tab:2HDM-slice-params-worst}, but for the slices with the median binary accuracy for the precustodial potential, with parameters defined in Eq.(\ref{eq:V-PC}).}
    \label{tab:Precustodial-slice-params-med}
\end{table}

\clearpage
\setlength{\bibsep}{3pt}
\bibliographystyle{JHEP}
\bibliography{main}

\end{document}